\documentclass[a4paper,UKenglish,cleveref, autoref, thm-restate]{lipics-v2021}

\title{Induction and Recursion Principles in a Higher-Order Quantitative Logic for Probability} 

\titlerunning{Induction and Recursion Principles in a Quantitative HO-Logic for Probability} 

\author{Giorgio Bacci}{Department of Computer Science, Aalborg University, Denmark \and \url{https://homes.cs.aau.dk/~grbacci/}}{grbacci@cs.aau.dk}{https://orcid.org/0000-0003-4004-6049}{This work was supported by Digital Research Centre Denmark (DIREC), under the Robust NIDS project.}

\author{Rasmus Ejlers M{\o}gelberg}{IT University of Copenhagen, Denmark \and \url{https://www.itu.dk/~mogel/}}{mogel@itu.dk}{https://orcid.org/0000-0003-0386-4376}{{This work was supported by the Independent Research Fund Denmark, grant number 2032-00134B.}}

\authorrunning{G. Bacci and R.\,E. M{\o}gelberg} 

\Copyright{Giorgio Bacci and Rasmus Ejlers M{\o}gelberg} 



\category{} 

\relatedversion{} 

\nolinenumbers 

\EventLongTitle{42nd Conference on Very Important Topics (CVIT 2016)}
\EventShortTitle{CVIT 2016}
\EventAcronym{CVIT}
\EventYear{2016}
\EventDate{December 24--27, 2016}
\EventLocation{Little Whinging, United Kingdom}
\ArticleNo{1}

\usepackage{amsmath,amsfonts,amsthm} 
\usepackage{stmaryrd}
\usepackage{proof,xifthen,leftindex}
\usepackage{enumitem}
\usepackage{hyperref}
\usepackage{tikz}
\usepackage{mathpartir}
\usetikzlibrary{arrows,shapes,cd}
\usetikzlibrary{calc} 
  \tikzset{
    commutative diagram/.style 2 args={
    	matrix of math nodes, row sep=#1,column sep=#2,
	text height=1.5ex, text depth=0.25ex},
    commutative diagram/.default={1cm}{1cm}
    }
  \tikzset{    
    skip loop/.style n args={3}{to path={-- ++(0,#1) -| node[pos=0.25,#2] {#3} (\tikztotarget)}},
    cross line/.style={preaction={draw=white, -, line width=6pt}}
  }
\providecommand*{\ifempty}[3]{\ifthenelse{\isempty{#1}}{#2}{#3}}
\newcommand{\parensmathoper}[2]{\ensuremath{#1\ifempty{#2}{}{(#2)}}}

\newcommand{\ie}{\textit{i.e.}}

\newcommand{\lol}{\mathbin{\multimap}}

\newcommand{\cat}[1]{\mathbf{#1}}

\newcommand{\CMet}{\cat{CMet}}
\newcommand{\Set}{\cat{Set}}

\newcommand{\mprod}{\otimes}
\newcommand{\mexp}{\lol}
\newcommand{\unit}{\mathbf{1}}

\newcommand{\nat}{\mathbb{N}}
\newcommand{\bool}{\mathsf{Bool}}
\newcommand{\unitT}{1}
\newcommand{\extpreals}{[0,\infty]}
\newcommand{\uinterval}{[0,1]}
\newcommand{\D}{\mathcal{D}}
\newcommand{\emptyctx}{\langle\rangle}
\newcommand{\pair}[2]{\otimes_{#1,#2}}
\renewcommand{\pair}[2]{\leftindex_{#1}\otimes_{#2}} 
\newcommand{\inj}{\mathsf{\mathit{inj}}}

\newcommand{\infrule}[3][]{\infer[{\ifempty{#1}{}{(\textsc{#1})\;}}]{#3}{#2}}
\newcommand{\doubleinfrule}[3][]{\infer=[{\ifempty{#1}{}{(\textsc{#1})\;}}]{#3}{#2}}

\newcommand{\sminus}{\!-\!}

\newcommand{\scaling}{}
\newcommand{\semfp}{\mathsf{fp}}
\newcommand{\semfix}{\mathsf{fix}}
\newcommand{\iso}{\cong}

\newcommand{\cp}{\mathsf{copy}}
\newcommand{\splt}{\mathsf{split}}
\newcommand{\distribute}{\mathsf{dist}}

\newcommand{\iterate}[1][n]{\parensmathoper{\mathsf{iterate}}{#1}}
\newcommand{\proj}{\mathsf{proj}}
\newcommand{\semweak}{\mathsf{weak}}

\newcommand{\Proc}[1][\discfact]{\mathbb{P}_{#1}}
\newcommand{\Labels}{A}
\newcommand{\Act}{\mathcal A}
\newcommand{\State}{\mathcal S}
\newcommand{\trans}{\mathcal P} 
\newcommand{\reward}{\mathcal R}
\newcommand{\policy}{\pi}
\newcommand{\refine}{\textsf{TDstep}}
\newcommand{\TDfunc}{\textsf{TD}}
\newcommand{\strength}{\textsf{st}}

\newcommand{\ts}{\vdash}
\newcommand{\hastype}[3]{#1\ts #2 : #3}

\newcommand{\geo}[1][p]{\mathsf{geo}_{#1}}

\newcommand{\semProc}[1][\discfact]{\mathbb{P}_{#1}}
\newcommand{\semLabels}{A}

\newcommand{\Pos}{\mathsf{Pos}}
\newcommand{\flip}{\mathsf{flip}}
\newcommand{\unifZN}{\mathsf{unif}_{0,N}}
\newcommand{\hwalk}{\mathsf{hwalk}}
\newcommand{\bijarg}[2]{\sigma_{#1,#2}}
\newcommand{\bij}{\sigma}

\newcommand{\letIn}[2]{\mathsf{let~} #1 \mathsf{~in~} #2}
\newcommand{\fix}[2]{\mathsf{fix~} #1 . #2 }
\newcommand{\Zero}{\mathsf{zero}}
\newcommand{\Succ}{\parensmathoper{\mathsf{succ}}}
\newcommand{\Succbare}{\mathsf{succ}}
\newcommand{\Rec}[3]{\mathsf{rec}(#1,#2,#3)}
\newcommand{\tuple}[1]{\langle #1 \rangle}
\newcommand{\case}[5]{\mathsf{case~} #1%
    \mathsf{~of~} [\inj_1 #2 . #3 \,|\, \inj_2 #4 . #5]%
}

\newcommand{\hd}{\textsf{Hd}}
\newcommand{\tl}{\textsf{Tl}}

\newcommand{\hdstate}[1][\epsilon]{\textsf{hd}_{#1}}
\newcommand{\tlstate}[1][\epsilon]{\textsf{tl}_{#1}}
\newcommand{\hdtlstate}[1][\epsilon]{\textsf{hdtl}_{#1}}

\newcommand{\hdfair}{\textsf{hd}}
\newcommand{\tlfair}{\textsf{tl}}

\newcommand{\mean}[3]{E_{#1 \sim #2}[#3]}
 
\newcommand{\Prop}{\mathsf{Prop}}

\newcommand{\true}{\mathsf{tt}}
\newcommand{\false}{\mathsf{ff}}
\newcommand{\ltensor}{*}
\renewcommand{\ltensor}{\bullet}
\newcommand{\lexp}{\mathbin{{-}\kern-.2em\raisebox{0ex}{$*$}}}
\renewcommand{\lexp}{\mathbin{{-}\kern-.2em\raisebox{0ex}{$\bullet$}}}
\newcommand{\conj}{\wedge}

\newcommand{\curry}{\mathsf{curry}}

\newcommand{\eval}{\mathsf{ev}}
\newcommand{\unfold}{\mathsf{ufld}}
\newcommand{\fold}{\mathsf{fld}}
\newcommand{\proc}[2]{#1 \cdot #2}
\renewcommand{\proc}[2]{#1 ; #2}

\newcommand{\TD}{\mathcal{D}}

\newcommand{\kant}[2]{K(#1,#2)}

\newcommand{\coup}[4]{#2 \in \mathsf{Cpl}_{#1}(#3,#4)}
\newcommand{\Rel}[2]{\mathsf{Rel}(#1,#2)}
\newcommand{\RelQ}[4]{\mathsf{Rel}_{#3,#4}(#1,#2)}
\newcommand{\Pred}[1]{\mathsf{Pred}(#1)}
\newcommand{\PredQ}[2]{\mathsf{Pred}_{#2}(#1)}
\newcommand{\Bisim}[1]{\mathsf{Bisim(#1)}} 
\newcommand{\bisim}{\sim} 

\newcommand{\discfact}{c} 

\newcommand{\logicJ}[3][\Delta]{#1 \mid #2 \vdash #3}

\newcommand{\peq}{=}
\newcommand{\jeq}{\equiv}
\newcommand{\defeq}{\triangleq}

\newcommand{\den}[1]{\llbracket #1\rrbracket}
\newcommand{\denD}[1]{\left\llbracket #1 \right\rrbracket}

\newcommand{\asmap}{\beta} 

\newcommand{\store}{\mathbb{S}}

\newcommand{\tmp}{\mathsf{tmp}}

\begin{document}

\maketitle

\begin{abstract}
Quantitative logic reasons about the degree to which formulas are satisfied. 
This paper studies the fundamental reasoning principles of higher-order quantitative logic and their application to reasoning about probabilistic programs and processes. 

We construct an affine calculus for $1$-bounded complete metric spaces and the monad for probability measures equipped with the Kantorovich distance. The calculus includes a form of guarded recursion interpreted via Banach's fixed point theorem, useful, e.g., for recursive programming with processes. We then define an affine higher-order quantitative logic for reasoning about terms of our calculus. The logic includes novel principles for guarded recursion, and induction over probability measures and natural numbers.

We illustrate the expressivity of the logic by a sequence of case studies: Proving upper limits on bisimilarity distances of Markov processes, showing convergence of 
a temporal learning algorithm and of a random walk using a coupling argument. 
\end{abstract}

\section{Introduction}

In computer science, logics are traditionally designed for proving precise qualitative properties of programs, such as program equality. 
However, in many modern applications, especially those that involve probabilistic programming, one is often interested 
in proving quantitative properties of programs, such as upper limits on program distances, sensitivity of program 
outputs to program inputs, or convergence of sequences of programs. Such properties are important in diverse
application areas such as differential privacy~\cite{AmorimGHKC17,ReedP10}, security~\cite{AlmeidaBBGLLOPQSSS23,AvanziniBGMV24} and machine learning~\cite{OlmedoKKM16}. 
Similarly, in process algebra, it has long been known that for probabilistic processes, the notion of bisimilarity should be stated
quantitatively, in the form of a metric, 
to be robust to small perturbations that may otherwise compromise the exact comparison of behaviours~\cite{GiacaloneJS94,Desharnais00}.
Also program logics for probabilistic programming languages are often defined quantitatively, using 
random variables valued in $\extpreals$ as a 
quantitative notion of predicates of states~\cite{Kozen85,MorganMS96,McIverM05,BHKKM21,AvanziniBDG25}. 

What has been lacking so far is a general notion of quantitative logic in which all these applications can be expressed
as special cases. This paper develops such a logic, and in particular, the fundamental principles 
governing equality, those for reasoning about distributions, and the notion of guarded recursion. 
What we found is that these 
simple basic principles are extremely powerful, and that by reading standard definitions of concepts such as a coupling and
bisimulation in our logic we can reason about quantitative versions of these concepts in a very natural way.

\subparagraph*{A sensitivity calculus for complete metric spaces.}

Quantitative logic is the logic of metric spaces.
A metric space can be thought of as a set with
a quantitative notion of equality: The smaller the distance between two elements, the more equal they are. From the
logical perspective, the appropriate notion of morphism between metric spaces is that of a non-expansive map: a function $f$ satisfying $d(f(x),f(y)) \leq d(x,y)$, for
all $x, y$. These are precisely the maps that preserve equality in the quantitative sense, sending quantitatively equal inputs to quantitatively equal outputs.
We will restrict our attention to those metric spaces that are non-empty, complete, and $1$-bounded, and denote the resulting category by $\CMet$. 
We found this category to be the most suitable for our purposes, although alternative models are also possible, as we discuss further in Section~\ref{sec:conclusions}.

The category $\CMet$ is symmetric monoidal closed, and moreover has a strong monad $\TD$ on it, mapping a metric
space to the set of Radon distributions equipped with the Kantorovich metric. The Kantorovich metric is a natural notion
of metric on distributions, and is used in most applications, including program logics and bisimilarity distance. 
Another important operation is that of scaling a metric space by a factor $r \geq 0$. The resulting space $r X$ has the same
underlying set as $X$, but all distances are scaled by $r$ (but bounded by $1$). 
Scaling is used to express the sensitivity of functions to changes in input data: A map has type $r X \to Y$ iff it is $r$-Lipschitz continuous or, using the terminology of the programming languages Fuzz~\cite{ReedP10} and DFuzz~\cite{GaboardiHHNP13}, is \emph{$r$-sensitive} in $X$.
An important example is the operation that maps two distributions $\mu$ and $\nu$ to their convex combination $\mu \oplus_p \nu$, corresponding to choosing $\mu$ with probability $p$ and $\nu$ with probability $1\sminus p$, for $p \in (0,1)$. This is an operation of type
\begin{equation} \label{eq:oplus:type}
\oplus_p \colon p\D X \mprod (1-p) \D X \to \D X
\end{equation}
expressing that $\oplus_p$ is $p$-sensitive in the first argument and $(1-p)$-sensitive in the second,
meaning $d(\mu\oplus_p\nu, \mu'\oplus_p \nu) \leq p\cdot d(\mu, \mu')$ and $d(\mu\oplus_p\nu, \mu\oplus_p \nu') \leq (1-p)\cdot d(\nu, \nu')$ hold simultaneously, where both distances are measured in $\D X$.

Inspired by Fuzz~\cite{ReedP10}, we define a $\lambda$-calculus for programming in $\CMet$. The calculus is affine, because the $\mprod$ of the monoidal
structure has projections, but not diagonals. The calculus is simply-typed, with typing judgements using contexts with sensitivity annotations on all variables. For 
example, the sensitivity of $\oplus_p$ is expressed by the typing judgement $\hastype{\mu :^p \TD X, \nu :^{1-p}\TD X}{\mu \oplus_p \nu}{\TD X}$. 

Metric spaces do not model general recursion, but they do model a form of guarded recursion via the Banach 
fixed point theorem: Any non-expansive map $f \colon p\scaling X \to X$ has a unique 
fixed point if $p < 1$ and $X$ is complete and non-empty. Guarded recursion on this form has previously been 
studied for ultra-metric spaces~\cite{VinkR99}, but
that setting is simpler, because it models simply typed lambda calculus and intuitionistic logic. Generalising to the 
affine setting of general metric spaces has surprising consequences. For example, in the ultra-metric setting, the fixed 
point operator defines a non-expansive map $(p \scaling X \to X) \to X$ whereas in the general metric setting, this must be 
weakened to $(p \scaling X \mexp X) \to (1-p)\scaling X$ where $\mexp$ is the closed monoidal structure. In our 
calculus the \emph{guarded fixed point combinator} can be given type $\Gamma \vdash \fix{x}{t} : A$ 
if $(1\sminus p)\scaling\Gamma, x:^pA \vdash t : A$, for $p < 1$. Here, 
the multiplication refers to the operation of scaling all
sensitivity annotations in a context.

The type \eqref{eq:oplus:type} of the $\oplus_p$ combinator means that many distributions and processes can be defined as 
fixed points.
For example, the geometric distribution is recursively defined as 
$\geo = \delta(0) \oplus_p \D(\Succbare)(\geo)$, where $\delta(0)$ is the Dirac distribution, and $\Succbare \colon \nat \to \nat$ is the
successor function. 
This works because the defining 
equation for $\geo$ is productive: it only calls itself recursively with probability $1-p < 1$. 

\subparagraph*{A higher-order quantitative logic.}
Our main contribution is a quantitative logic for reasoning about terms in our calculus. 

Logical reasoning is via inference of judgements $\Psi \vdash \varphi$ which express that the predicate $\varphi$ logically follows from the sequence of predicates $\Psi = \psi_1, \dots, \psi_n$. 
Predicates are valued in 
the unit interval $\uinterval$ with $0$ being true and $1$ being false. For example, the equality
predicate is interpreted as distance between elements of a metric space. 
This leads naturally to an affine logic, because 
transitivity $x\peq y, y\peq z \ts x \peq z$ can be interpreted as the triangle inequality 
$d(x,y) + d(y,z) \geq d(x,z)$ (one of the axioms of metric spaces) if the comma is interpreted as sum. 

The first step in defining a logic is to state what the well-formed predicates are. Since the 
unit interval is itself an element in $\CMet$, we use our calculus to program with 
predicates using a special type $\Prop$ interpreted as $\uinterval$. The interval $\uinterval$
carries a closed monoid structure given by (truncated) sum and subtraction,
and universal and existential quantification can be modelled using infimum and supremum, respectively.
Since the logic allows quantification over all objects of $\CMet$, including 
exponents of $\Prop$, it is a higher-order logic. 

Like metric spaces, propositions can be rescaled by factors $r \in \extpreals$ and this 
can be used to express a logical principle of guarded recursion.
In a closed context this states that if we can 
prove $\phi$ from $p\cdot \phi$ for some $p<1$, then $\phi$ holds. 
The general rule reflects that of the fixed point combinator:
if $(1-p)\Psi, p\phi\ts \phi$ holds, then  $\Psi \ts \phi$.

We also study induction principles, both for natural numbers and for the Radon distribution monad $\TD$.
In the latter case, the principle states that $\Psi \vdash \phi[\mu/x]$ holds for all 
$\mu \in \TD X$ if it does on  
Dirac distributions and 
$p(\phi[\mu/x]), (1-p)\phi[\nu/x] \ts \phi[\mu \oplus_p \nu/x]$, for all $p \in (0,1)$.
At first sight, this may appear to only prove it for all finitely supported distributions, and $\TD X$ also includes continuous distributions.
The principle is nevertheless sound, because the finite distributions are dense in $\D X$, and because the principle has the side condition
that $x$ has finite sensitivity in $\phi$, so that $\phi$ is continuous in $x$. This principle reflects
the observation by Mardare et al.~\cite{MardarePP16,MardarePP18} 
that $\TD$ is generated freely by Dirac distributions and $\oplus_p$.

When stating the elimination rule for the 
equality predicate, one must take sensitivity into account. More precisely, if $x$ has sensitivity $p$ in $\phi$
then $\phi[t/x], p(t \peq s ) \ts \phi[s/x]$. 
This was previously observed by Dagnino and Pasquali~\cite{DagninoP22} in a propositional logic, and we adapt one of their rules as
an elimination principle for equality. We show that the rule
has wide-ranging consequences, including that equality can be proved symmetric and transitive. Equality
can also be proved a congruence when this is formulated in a way that takes sensitivities into account. 
For example, using the type (\ref{eq:oplus:type}) of $\oplus_p$, one can prove that 
\begin{equation} \label{eq:oplus:congruence}
p(x \peq y), (1-p)(z\peq w) \ts x \oplus_p z \peq y \oplus_p w \, .
\end{equation}

\subparagraph*{Case studies.}

To show the expressiveness of our logic we develop a sequence of case studies. 
The first concerns bisimilarity distance of Markov processes. We model these in $\CMet$ as the terminal
solution $\semProc$ to $\semProc \iso \semLabels\mprod\discfact\D(\semProc)$. Here $\discfact \in (0,1]$ is a 
discount factor. The smaller it is, the more the distance emphasizes short-term behaviour. In $\semProc$ the metric is exactly the bisimilarity distance. We show how to define recursive processes using the 
fixed point combinator, and prove upper bounds on distances  $d(t,u) \leq c$ by proving $c \cdot \false \ts t \peq u$ in the logic,
using the guarded recursion principle. Some examples require $c<1$, but in many cases
the productivity requirement follows from (\ref{eq:oplus:congruence}). We also show how to
define the notion of bisimilarity inside the logic and prove it equivalent to equality when $c<1$. 

We also show how to prove convergence of a temporal learning algorithm and of a random walk on a hypercube. 
The latter example is particularly interesting, as it requires a coupling proof. Coupling arguments~\cite{LindvallBook} are a technique for showing 
equivalence or closeness of probabilistic programs by relating the probabilistic choices made by one program
to those by another. We show how this technique can be used in our logic by internalising the definition of the 
Kantorovich distance: We prove that two distributions are equal if and only if there exists an equality-coupling 
between them. Note that since this biimplication is a statement in quantitative logic, its semantic meaning is an
equality of numbers in the unit interval. The proof of this illustrates the power of the principles of our logic,
in particular the elimination rule for equality and the induction principle for distributions.

\subparagraph*{Contributions.}
In summary, our contributions are: 
\begin{itemize}[topsep=0ex]
\item We present an affine lambda calculus with sensitivity annotations 
for programming in the category $\CMet$ of complete 1-bounded metric spaces; 
\item We formulate a first (to our knowledge) 
higher-order logic for quantitative reasoning in which equality is interpreted as distance in metric spaces. 
This includes new rules for recursion over natural numbers and probability measures, as well as a guarded 
recursion principle;
\item We show by example how to use the logic to reason about Markov processes, convergence in temporal learning
and for coupling arguments, illustrating the power of the combination of 
the above mentioned principles;
\end{itemize}

\subparagraph*{Overview.} The paper is organised as follows. We first discuss preliminaries on metric spaces in 
Section~\ref{sec:prelim}. Sections~\ref{sec:fixed:points} and~\ref{sec:prob:meas} discuss the 
Banach fixed point operator and the probability measure monad. Sections~\ref{sec:calculus} and~\ref{sec:logic}
discuss the syntax and semantics of the calculus and the logic, respectively. 
Section~\ref{sec:basic:prop} shows how to prove basic properties of the logic, including that equality is a congruence 
which is equivalent to an internalisation of the 
usual definition of the Kantorovich measure. Section~\ref{sec:markov} shows applications to Markov processes, 
Section~\ref{sec:temporal:learning} shows an application to temporal learning, and Section~\ref{sec:hypercube}
illustrates how to use coupling arguments in the logic. 
Finally, Section~\ref{sec:related} discusses related work, and 
Section~\ref{sec:conclusions} concludes. 

\section{Preliminaries on metric spaces}
\label{sec:prelim}

A ($1$-bounded) \emph{metric space} is a set $X$ equipped with a
distance function $d_X \colon X \times X \to \uinterval$ satisfying (\emph{reflexivity}) $d_X(x,y) =0$ iff $x=y$; (\emph{symmetry}) $d_X(x,y) = d_X(y,x)$; and (\emph{triangular inequality}) $d_X(x,z) \leq d_X(x,y) + d_X(y,z)$. 

A function $f \colon X \to Y$ between metric spaces is \emph{$r$-Lipschitz continuous}, for $r\geq 0$, if $r \cdot d_X(x,x') \geq  d_Y(f(x),f(x'))$; \emph{non-expansive} when $r = 1$; and a \emph{contraction} when $r < 1$ and $X = Y$. 
A metric space is complete if all Cauchy sequences converge. 

\smallskip
In this paper, we work with non-empty $1$-bounded complete metric spaces. 
These form a category $\CMet$ with non-expansive maps as morphisms. 
Restricting to $1$-bounded spaces allows us to include sets as \emph{discrete} metric spaces by setting all distances between distinct elements to $1$. 
This defines a left adjoint to the forgetful functor from $\CMet$ to $\Set$: if $X$ is discrete, then all maps 
$f\colon X \to Y$ are non-expansive. 
As a consequence, we can regard $\Set$ as a full subcategory of $\CMet$. 
 An important example for this paper is $\nat$, the set of natural numbers, which is an object in $\CMet$.

The category $\CMet$ has both binary products and coproducts: $X \times Y$ combines spaces by equipping the Cartesian product with the point-wise maximum distance; $X + Y$ combines them into a disjoint union by keeping elements from different components as far as possible from each other (\ie, at distance $1$). The singleton metric space $\unit$ is the final object.

There is another natural structure on the Cartesian product, the \emph{tensor} $X \mprod Y$, that combines distances by $1$-bounded truncated sum:
\begin{equation*}
    d_{X \mprod Y}((x,y),(x',y')) = \min \{ d_X(x,x') + d_Y(y,y'), 1 \} \,.
\end{equation*} 
Also the tensor product has non-expansive projections
$X \mprod Y \to X$ and $X \mprod Y \to Y$, but generally no diagonal $X \to X \mprod X$, unless $X$ is discrete. 
Note the following universal property of $\mprod$: A map from $X \mprod Y$ to $Z$ is non-expansive if
and only if it is non-expansive in each variable.

Unlike the categorical product, $\mprod$ allows currying and function application. More precisely, $(\CMet, \mprod, \unit)$ is a symmetric monoidal category, with unit $\unit$ and adjunction $(- \mprod X) \dashv (X \mexp -)$ making this structure closed. Here, $X \mexp Y$ denotes the set of non-expansive functions from $X$ to $Y$ endowed with point-wise supremum metric $d_{X \mexp Y}(f,g) = \sup_{x \in X} d_Y(f(x),g(x))$.
The counit of the adjunction is function evaluation
$\mathsf{ev}\colon (X \mexp Y) \mprod X \to Y$.

Nonexpansive morphisms in $\CMet$ subsume the notion of Lipschitz continuity through the rescaling functor $r \scaling X$, which scales distances by a  factor $r > 0$ as
\begin{equation*}
  d_{r\scaling X}(x,x') = 
  \min \{ r \cdot d_X(x,x'), 1 \} \,.
\end{equation*}
Indeed, by unpacking the definition, $f \colon r \scaling X \to Y$ is a morphism in $\CMet$ iff $f$ considered as a map $f \colon X \to Y$ is $r$-Lipschitz continuous.
For convenience we allow rescaling also for $r = 0$ and $r = \infty$,
and define $0 \scaling X$ as $\unit$, the one-point metric space, and $\infty \scaling X$ as the discrete metric space on $X$. Scaling preserves products and, under suitable restrictions, coproducts:
\begin{mathpar}
s A \times s B \iso s(A\times B) \quad \text{(for $s \in \extpreals$)} \,,
\and
r A + rB \iso r(A+B) \quad \text{(for $r \in [1,\infty]$)} \,.
\end{mathpar}

The next theorem states the properties of the interaction between $\mprod$ and 
scaling. First recall that $[0, \infty]$ is an ordered  
semiring with $\leq$, addition and multiplication defined as usual in most cases,
and by $\infty\cdot 0 = 0$ and $0 \cdot \infty = 0$. The properties together imply that 
the scaling operation is a $[0, \infty]$-graded comonad on $\CMet$ in the 
sense of \cite[Definition~13]{BrunelGMZ14} and \cite[Section~5.2]{GaboardiKOBU16}.

\begin{theorem} \label{thm:graded:comonad}
 There are natural transformations of  types
\begin{align*}
 m_{r,A,B} &\colon r\scaling A \mprod r\scaling B \to r\scaling (A \mprod B) \quad & 
 n_{r} &\colon \unit \to r\scaling \unit \\
 c_{r,s,A} &\colon (r+s) \scaling A \to r\scaling A \mprod s\scaling A &
 w_A &\colon 0\scaling A \to \unit \\
 \asmap_{r,s,A} &\colon (rs)\scaling A \to r\scaling (s\scaling A) & 
 \epsilon_A &\colon 1\scaling A \to A \\
  \kappa_{r,s,A} &\colon s \scaling A \to r \scaling A 	\quad \text{(for $r \leq s$)} 
\end{align*}
Moreover, if $s\leq 1$ or $r \geq 1$ then $\asmap_{r,s,A}$ is an isomorphism, and if $r\geq 1$ then $m_{r,A,B}$ is an isomorphism.
\end{theorem}

Note that without the conditions above, $\asmap_{r,s,A}$ and $m_{r,A,B}$
need not be isomorphisms. For example, if $A$ is discrete then $\frac12 (2A) = \frac12 A$.

\section{Fixed points}
\label{sec:fixed:points}

The Banach fixed point theorem~\cite{banach1922} states that any contractive function on a non-empty complete metric space has a unique fixed point. 
If $f \colon X \mprod p Y \to Y$ is a morphism in $\CMet$ and $p<1$ then, for any $x \in X$, the map $f(x, -) \colon p Y \to Y$ is a contraction on $Y$ and so has a unique fixed point $\semfp(f)(x)$. 

\begin{proposition} \label{prop:fp}
Let $p < 1$. If $f \colon (1-p)X \mprod pY \to Y$ then $\semfp (f) \colon X \to Y$ is non-expansive.  
There is a non-expansive map $\semfix \colon (pY \mexp Y) \to (1-p) Y$ mapping functions to fixed points.
\end{proposition}
\begin{proof}
For the proof of the second statement, define $\semfix$ as the function mapping a contraction
$f \colon pY \to Y$ to its unique fixed point. Let $f,f' \colon pY \to Y$, then 
\begin{align*}
 d(\semfix (f), \semfix (f')) & \leq d(\semfix (f), f(\semfix (f')) + d(f(\semfix(f')), \semfix(f')) \\
 & =  d(f(\semfix (f)), f(\semfix (f')) + d(f(\semfix(f')), f'(\semfix(f'))) \\
 & \leq pd(\semfix(f), \semfix(f')) + d(f,f')  
\end{align*}
so that $(1-p)d(\semfix (f), \semfix (f')) \leq d(f,f')$, as desired. For the first statement, by currying $f$ and composing with $\semfix$ we see that
\[
  \semfix \circ \curry(f) \colon (1-p)X \to (1-p) Y \,.
\]
Thus, we can define $\semfp(f) \colon X \to Y$ as the composition $\asmap^{-1}_Y \circ \frac{1}{1-p}(\semfix \circ \curry(f)) \circ \asmap_X$, where $\asmap$ is natural isomorphims from Theorem~\ref{thm:graded:comonad}.
\end{proof}

\section{Probability measures}
\label{sec:prob:meas}

We introduce the \emph{Radon probability monad} on $\CMet$ and recall its presentation as the free complete interpolative barycentric algebra~\cite{MardarePP18}.

A (Borel) probability measure $\mu$ on a metric space $X$ is \emph{Radon} if for any Borel set $E \subseteq X$, $\mu(E)$ is the supremum of $\mu(K)$ over all compact subsets $K$ of $E$. Examples of Radon probability measures are Dirac measures, (the Borel restriction of) the Lebesgue measure over the unit interval, any probability measure with finite support or over a Polish space (\ie, a complete metric space with a countable dense subset).

A \emph{coupling} between
two probability measures $\mu$ and $\nu$ on $X$ is a probability measure $\omega$ on $X \times X$
whose left and right marginals are,
respectively, $\mu$ and $\nu$ (\ie, 
$\omega(E \times X) = \mu(E)$ and $\omega(X \times E) = \nu(E)$, for all Borel sets $E$). The product measure $\mu \times \nu$ is always a coupling between $\mu$ and $\nu$. Moreover, if $\mu$ and $\nu$ are Radon, so is any coupling $\omega$ between them.

For $X$ a metric space, denote by $\D X$ the space of Radon probability measures over $X$ equipped with the 
\emph{Kantorovich distance}, defined by
\begin{equation} \label{eq:Kantorovic}
d_{\D X}(\mu,\nu) = \min_{\omega} \int \, d_X(x,x') \; {\omega(\mathrm{d}x,\mathrm{d}x')}
\end{equation}
where $\omega$ runs over the couplings between $\mu$ and $\nu$.
If $X$ is a complete metric space, so is $\D X$. 

The definition above extends to a monad $\D$ on $\CMet$, called \emph{Radon probability monad}, with underlying functor acting on morphisms $f\colon X \to Y$ as $\mu \in \D X \mapsto \mu \circ f^{-1} \in \D Y$ (a.k.a., the pushforward measure along $f$).
The unit of $\D$ is the Dirac measure $\delta_X \colon X\to \D X$, but rather than describing the multiplication,
we recall that this monad has an algebraic presentation as the free complete
\emph{interpolative
barycentric algebra}~\cite{MardarePP16,MardarePP18}.
\begin{definition}[IB Algebra]
\label{def:IB:alg}
An \emph{interpolative barycentric algebra} is a complete metric space $X$ with non-expansive operations $\oplus_p \colon pX \mprod(1-p) X \to X$, for all $p \in (0,1)$, such that
\begin{align*}
    x \oplus_p x &= x \tag{\sc idem} \\
    x \oplus_p y &= y \oplus_{1-p} x\tag{\sc comm} \\
    (x \oplus_p y) \oplus_q z 
    &= x \oplus_{pq} (y \oplus_{\frac{q-pq}{1-pq}} z) \tag{\sc assoc} 
\end{align*}

A homomorphism $f \colon X \to Y$ of IB algebras is a non-expansive map such that $f(x \oplus_p y) = f(x) \oplus_p f(y)$ for all $p \in (0,1)$.
\end{definition}
The axioms are those of barycentric algebras (a.k.a., convex algebras), axiomatizing  probabilistic choice by means of binary convex combinations $x \oplus_p y$. 
The definition above is equivalent to that in~\cite{MardarePP18}, as the
type imposed on the operations $\oplus_p$ is 
equivalent to requiring
\[
   d_X(x \oplus_p y, x' \oplus_p y')  \leq p d_X(x,x') + (1-p) d_X(y,y') \,. 
\]
The formulation proposed in Definition~\ref{def:IB:alg} is preferable in our context because it incorporates the Lipschitz constants directly into the type of the operation. This not only enables the remaining conditions to be expressed purely as equations but also ensures that many probabilistic processes can be defined using the Banach fixed point combinator
as we shall see below.

It is not difficult to show that, for any $X \in \CMet$,
$\D X$ is a complete interpolative barycentric algebra, by interpreting the operations as 
$\mu \oplus_p \nu = p \mu + (1-p) \nu$, for all $\mu,\nu \in \D X$. 
The next result states that $\D X$ is the free algebra with respect to all Lipschitz maps, which follows as a corollary of~\cite[Theorem~3.8]{MardarePP18}.
Before we state it,  note that if $A$ and $B$ are IB-algebras, the equational definition of IB-homomorphism extends
to Lipschitz maps $f \colon rA \to B$. In terms of diagrams,  this can be stated as the commutativity of
\[
 \begin{tikzcd}[column sep=10ex]
  (pr)A \mprod (\bar pr) A \ar{r}{m \circ (\asmap\mprod\asmap)} \arrow[d,"\asmap\mprod\asmap"'] & r(pA \mprod \bar pA) \ar{r} & rA \ar{d}{f}  \\
  p(rA) \mprod \bar p(rA) \ar{r}{pf \mprod \bar p f} & pB \mprod \bar p B \ar{r} & B
 \end{tikzcd}
 \tag{where $\bar p = 1-p$}
\]

\begin{proposition} \label{prop:D:free:IB}
 If $f \colon \Gamma \mprod rX \to Y$ (with $r<\infty$) and $Y$ is an IB algebra, there exists a unique  $\overline f \colon \Gamma \mprod r\D X \to Y$
 which is a homomorphism in its second argument, satisfying 
 $f = \overline f \circ (\Gamma \mprod r\delta_X)$. 
\end{proposition}

Observe that the special case of $r=1$ and $\Gamma=\unit$ 
in Proposition~\ref{prop:D:free:IB} can be rephrased as the existence of a left adjoint to the forgetful functor from the category of IB algebras to $\CMet$.
Thus, as anticipated, $\D$ forms a monad on $\CMet$. This monad is moreover strong. The restriction on $r$ being finite means that 
$\overline f$ is continuous. This is necessary because the other requirements otherwise only
determine the value of $\overline f$ on finite distributions, which form a dense subset of $\D X$. 

\begin{lemma} \label{lem:IB:mexp:r}
 If $X,Y$ are IB algebras, so are, $X \mprod Y$, $Z \mexp X$, and $qX$ for any $Z$ and $q \leq 1$. 
\end{lemma}

\section{A calculus for $\CMet$}
\label{sec:calculus}

We now define a calculus for programming in the category $\CMet$. 
\subparagraph*{Syntax.}
The syntax is based on a simply-typed $\lambda$-calculus with products and sums, extended with primitives for probabilistic distributions, recursion, and fixed points. 
\begin{align*}
 t, u, v &{} ::= {}
  x \mid \lambda x. t \mid t u \mid
 () \mid 
 \tuple{t,u} \mid \pi_1 t \mid \pi_2 t \mid
 \inj_1 t \mid \inj_2 t \mid \case{t}{x}{u}{y}{v}
 \\
 &\mid  (t,u) \mid \letIn{(x,y) = u }{t} 
 \mid \delta t \mid  t \oplus_p u \mid \letIn{x = u}{t} \\
 &\mid \Zero \mid \Succ{t} \mid \Rec{u}{(x,y).t}{v} 
 \mid \fix{x}{t}
\end{align*}
There are two pairs constructors, $\tuple{t,u}$ and $(t,u)$, corresponding to the Cartesian and monoidal products, respectively. The first one is eliminated using the projections $\pi_i t$, whereas the second one is eliminated using $(\letIn{(x,y) = u }{t})$. The term $()$ is the unit value. The injections $\inj_i t$ form expressions of sum type, which are eliminated by case analysis ($\case{t}{x}{u}{y}{v}$). 
The term $\delta(t)$ denotes a Dirac distribution, and $t \oplus_p u$ the
convex sum of $t$ and $u$. Probability distributions are  sampled using $(\letIn{x = u}{t})$. The terms $\Zero$ and $\Succ{t}$ are constructors for natural numbers. $\Rec{u}{(x,y).t}{v}$ denotes a term obtained by primitive recursion on natural numbers. Finally, $\fix{x}{t}$ is the ``Banach'' fixed point combinator. 

The \emph{types} of the calculus are defined by the grammar  
\begin{align*}
    A,B ::= 
    \nat \mid
    \unitT \mid 
    A \times B \mid 
    A + B \mid 
    A \pair{r}{s} B \mid A \lol_r B \mid 
    \TD A
    && \text{($r,s\geq 0$)}
\end{align*}
essentially corresponding to the categorical operations in $\CMet$ from Section~\ref{sec:prelim}.
Although rescaling of metric spaces played a central role in Section~\ref{sec:prelim}, it is not a primitive type former in the calculus. Instead, it is part of the tensor type $A \pair{r}{s} B$ and function type $A \lol_r B$ constructors. This choice was made to minimize the book keeping necessary for scalars in terms. Finally, $\TD A$ is the type of Radon probability measures on $A$.

Terms are typed with judgements of the form $\hastype{\Gamma}{t}{A}$, where $A$ is a type and $\Gamma$ a typing context. The complete list of typing rules is given in Figure~\ref{fig:typingrules}.
The typing system is inspired by Fuzz~\cite{ReedP10}, where the sensitivity of each variable used in a term is tracked by annotations of the form $x:^r A$ in typing contexts. More precisely,
a binding $x:^rA$ in a context $\Gamma$ means that $x$ has type $A$ under $\Gamma$ and that terms typed under $\Gamma$ are $r$-sensitive with respect to
$x$.

Formally, typing contexts are constructed according to the following formation rules
\begin{align*}
    \frac{}{\emptyctx :: \textsf{ctx}}
    &&
    \frac{\Gamma :: \textsf{ctx} \quad x \notin \Gamma \quad r\in \extpreals}{ \Gamma, x :^r A :: \textsf{ctx}}
\end{align*}
As usual, $\Gamma, \Gamma'$ denotes concatenation of contexts with disjoint variable bindings.  
Most rules use the operations of sum $\Gamma+\Gamma'$ and scaling $r\scaling \Gamma$ of contexts, to keep track of the sensitivities.
These are defined as the point-wise sum and scaling of the sensitivity in each variable binding: 
\begin{align*}
 \emptyctx + \emptyctx &= \emptyctx &
 (\Gamma, x :^r A) + (\Gamma', x :^s A) &= (\Gamma + \Gamma'), x :^{r+s} A \,,
\\
 r \scaling \emptyctx &= \emptyctx &
 r \scaling (\Gamma, x :^s A) &= r \scaling \Gamma, x :^{r \cdot s} A \,.
\end{align*}
Observe that $\Gamma+\Gamma'$ is defined only when $\Gamma$ and $\Gamma'$ are compatible, \ie, they agree on the order and the types of all variable bindings ---for example $(x:^r A,y:^s B)+(y:^s B,x:^r A)$ is not defined. In the rest of the paper, whenever we take a sum $\Gamma+\Gamma'$, compatibility of the contexts is assumed implicitly.

\begin{figure*}[t]
\small
\centering
\begin{mathpar}
\infrule[var]{
    r \geq 1}{
\hastype{\Gamma, x:^rA, \Gamma'}{x}{A} }
\and
\infrule[unit]{ }{
\hastype{\Gamma}{()}{\unitT} }
\and
\infrule[abs]{
    \hastype{\Gamma, x:^rA}{t}{B} }{
\hastype{\Gamma}{\lambda x. t}{A \lol_r B} }
\and
\infrule[app]{
    \hastype{\Gamma}{t}{A \lol_r B} &
    \hastype{\Gamma'}{u}{A} }{
\hastype{\Gamma + r\scaling\Gamma'}{t\,u}{B} }
\and 
\infrule[pair]{
    \hastype{\Gamma}{t}{A} &
    \hastype{\Gamma}{u}{B} }{
\hastype{\Gamma}{\tuple{t,u}}{A \times B} }
\and
\infrule[$\pi_i$]{
    \hastype{\Gamma}{t}{A_1 \times A_2} }{
\hastype{\Gamma}{\pi_i t}{A_i} }
\and      
\infrule[$\inj_i$]{
    \hastype{\Gamma}{t}{A_i}}{
\hastype{\Gamma}{\inj_i t}{A_1+A_2} }
\and
\infrule[case]{
    \hastype{\Gamma, x:^r A}{u}{C} &
    \hastype{\Gamma, y:^r B}{v}{C} &
    \hastype{\Gamma'}{t}{A+B} &
    r\geq 1  }{ 
\hastype{\Gamma + r\scaling\Gamma'}{
    \case{t}{x}{u}{y}{v}}{C} }
\and
\infrule[$\mprod$]{
    \hastype{\Gamma}{t}{A} &
    \hastype{\Gamma'}{u}{B} }{
\hastype{r\scaling\Gamma + s\scaling\Gamma'+\Gamma''}{(t,u)}{A \pair{r}{s} B} }
\and 
\infrule[let-$\mprod$]{
    \hastype{\Gamma, x:^rA, y:^sB}{t}{C} &
    \hastype{\Gamma'}{u}{A \pair{r}{s} B} }{
\hastype{\Gamma + \Gamma'}{\letIn{(x,y) = u }{t}}{C} }
\and
\infrule[$\delta$]{
    \hastype{\Gamma}{t}{A} }{
\hastype{\Gamma}{\delta t}{\TD A} }
\and
\infrule[$\oplus_p$]{
    \hastype{\Gamma}{t}{\TD A} &
    \hastype{\Gamma'}{u}{\TD A}
    & p\in(0,1) }{
\hastype{p\scaling\Gamma + (1-p)\scaling\Gamma'}{t \oplus_p u}{\TD A} }
\and
\infrule[let]{
    \hastype{\Gamma, x:^rA}{t}{E} &
    \hastype{\Gamma'}{u}{\TD A} & 
    \text{$E$ IB algebra}  &
    r<\infty}{
\hastype{\Gamma + r\scaling\Gamma'}{\letIn{x = u}{t}}{E} }
\and
\infrule[zero]{}{
    \hastype{\Gamma}{\Zero}{\nat} }
\and
\infrule[succ]{
    \hastype{\Gamma}{t}{\nat} }{
\hastype{\Gamma}{\Succ{t}}{\nat} }
\and
\infrule[rec]{
    \hastype{\Gamma}{z}{A} &
    \hastype{\Gamma', x:^1 A, y:^1 \nat}{s}{A} &
    \hastype{\Gamma''}{n}{\nat} }{
\hastype{\Gamma + \infty \scaling \Gamma' + \Gamma''}{\Rec{z}{(x,y).s}{n}}{A} }
\and
\infrule[fix]{
    \hastype{(1-p)\scaling\Gamma, x:^pA}{t}{A} & 
    p < 1 }{
\hastype{\Gamma}{\fix{x}{t}}{A} }
\end{mathpar}
    \caption{Typing rules.}
    \label{fig:typingrules}
\end{figure*}

The rule (\textsc{var}) for variable introduction reflects that projection can be given any Lipschitz factor $r\geq 1$. We allow
all such $r$, and not just $r=1$ to incorporate weakening directly into the typing rules. 
Note that, in the rule (\textsc{case}) for sum elimination, the sensitivities of the bound variables $x$ and $y$ are required to match and be greater or equal than $1$, as scaling by $r<1$ does not commute with coproducts.
In the rule ($\mprod$) for tensor introduction, $\Gamma''$ is used merely to incorporate weakening in the rule%
\footnote{The use of a free context, is necessary for build weakening in the rule in the case both $r$ and $s$ are $0$. This detail was overlooked in~\cite{AmorimGHKC17} in their rule for introduction of scaled type as their weakening lemma fails when scaling by $0$.}. 
In the rule (\textsc{rec}) for natural number recursion, the successor case term $s$ can have additional variables to $x$ and $y$. However, 
$s$ must be applied $n$ times, and therefore the sensitivity of $\Gamma'$ must be scaled by $n$. Since $n$ is not known statically, 
the only possible upper limit is $\infty$. The type of the fixed point operator reflects Proposition~\ref{prop:fp}. The elimination rule for $\TD$ uses the judgment of a type being an IB algebra. These are defined by the grammar 
\[
 E,F ::= \TD A \mid E \pair{p}{q} F \mid A\lol_r E  
 \qquad \text{($p,q \leq 1$ and $r\geq 0$)}
\]
as justified by Lemma~\ref{lem:IB:mexp:r}.

\begin{example}[The geometric distribution] \label{example:geo:calc}
Let $p \in (0,1)$.
The geometric distribution $\geo : \TD(\nat)$ is uniquely characterized by the recurrence $\geo \jeq \delta(0) \oplus_p \TD(\Succbare)(\geo)$. It is therefore natural to define it as the following fixed point:
\begin{align*}
\geo &\defeq \fix x \delta(0) \oplus_p\TD(\Succbare)(x) \,,
\end{align*}
where, the functorial action of $\TD$ on a function $f \colon A \lol_r B$ (with $r < \infty$) is given by the term $\TD(f) \defeq \lambda x. \letIn{a=x}{\delta(f(a))} \colon \TD(A) \lol_r \TD(B)$. In particular, $\TD(\Succbare) \colon \TD(\nat) \lol_1 \TD(\nat)$. Hence, $\geo$ is well typed by (\textsc{fix}), since $x$ has sensitivity $1-p$ in the body of the fixed point.
\end{example}

As usual, terms are considered equal up to $\alpha$-equivalence.
We denote by $t[u/x]$ the capture-avoiding substitution of the term $u$ for the free variable $x$ in $t$. 

The following two forms of context weakening for typing judgements hold.
\begin{lemma}[Weakening] \label{lm:weakening} \ 
\begin{enumerate}
    \item If\, $\hastype{\Gamma,\Gamma'}{t}{A}$, then $\hastype{\Gamma,\Delta,\Gamma'}{t}{A}$;
    \item If\, $\hastype{\Gamma}{t}{A}$, then $\hastype{\Gamma+\Delta}{t}{A}$.
\end{enumerate}
\end{lemma}

Moreover, we have the substitution lemma.
\begin{lemma}[Substitution] \label{lm:substitution}
If\, $\hastype{\Gamma, x{:^r}A, \Gamma'}{t}{B}$ and $\hastype{\Delta}{u}{A}$, 
then $\hastype{(\Gamma,\Gamma') + r \scaling \Delta}{t[u/x]}{B}$.
\end{lemma}

The judgemental equality relation on terms is the least congruence relation generated by the rules in Figure~\ref{eq:judgmental:eq}. We use the symbol $\jeq$ for judgemental equality
to distinguish it from the propositional equality $t \peq u$, which is a predicate in the logic to be defined in Section~\ref{sec:logic}. 
Formally, judgemental equality is a relation on terms of the same type, and we will sometimes underline that by writing $\hastype\Gamma{t\jeq u}{A}$.
Similarly, the rules of Figure~\ref{eq:judgmental:eq} are to be understood as equalities in a typing context in which both sides have the same type. For example, in the case of the $\eta$-rule
for function types, $t$ is assumed to have function type.

\begin{figure*}[tbp]
\small
\begin{align*}
 (\lambda x. t)u \jeq t[u/x] &&& \pi_i(\tuple{t_1, t_2}) \jeq t_i \\
 t \jeq (\lambda x. t\,x) &&&
 \tuple{\pi_1(t), \pi_2(t)} \jeq t \\
t \jeq () &&&
 u_i[t/x_i] \jeq \case{(\inj_i t)}{x_1}{u_1}{x_2}{u_2} \\
 \letIn{(x,y) = (s,t)}u \jeq u[s/x,t/y] &&& 
  u[t/z] \jeq \case{t}{x}{u[\inj_1 x/z]}{y}{u[\inj_2 y/z]} \\
  \letIn{(x,y) = t}{u[(x,y)/z]} \jeq u[t/z]  &&&
  \letIn{x = \delta(t)}{u} \jeq u[t/x] \\ 
 &&& \letIn{x = (\letIn{y=s}t)}{u} \jeq \letIn{y=s}{(\letIn{x=t}u)}
   \\
 \fix xt \jeq t[\fix xt/x] &&& 
 \letIn{x = s\oplus_p t}u \jeq (\letIn{x = s}u)\oplus_p(\letIn{x=t}u) \\
 \Rec{z}{(x,y).s}{\Zero} \jeq z &&&
 \Rec z{(x,y).s}{\Succ n} \jeq s[\Rec z{(x,y).s}n/x, n/y]
\end{align*}
\caption{Judgemental equality (to these should be added the axioms of IB algebras of Definition~\ref{def:IB:alg}).}
\label{eq:judgmental:eq}
\end{figure*}

\subparagraph*{Semantics.}
Types and contexts are interpreted as objects in $\CMet$:
\begin{align*}
\begin{aligned}
\den{\nat} &\defeq \nat
&&&
\den{\unitT} &\defeq  \unit
\\
\den{A \times B} &\defeq \den A \times \den B
&&&
\den{A + B} &\defeq \den A + \den B
\\
\den{A \pair{r}{s} B} &\defeq r\den A \mprod s \den B 
&&&
\den{A \lol_r B} &\defeq r\den A \mexp \den B
\end{aligned}
&&
\begin{aligned}
\den{\TD A} &\defeq \D \den A \\
\den{\emptyctx} &\defeq \unit \\
\den{\Gamma, x :^r A} &\defeq \den\Gamma \mprod r\den A
\end{aligned}
\end{align*}
Judgements are interpreted as morphisms
\[
 \den{\hastype\Gamma tA} \colon \den \Gamma\to \den A
\]
in $\CMet$. The interpretation of terms is for most parts the usual set-theoretic interpretation. For example,
function abstraction and application are precisely the usual set-theoretic abstractions and applications. 
The exercise when defining the interpretation is ensuring the Lipschitz conditions
associated with types and typing judgements. This can be done entirely on the abstract category theoretic level, 
using maps such as those of Theorem~\ref{thm:graded:comonad}. One key point in doing so is that there
exist morphisms 
\begin{align*}
 \splt &\colon \den{\Gamma + \Gamma'} \to \den \Gamma \mprod \den{\Gamma'} \,, 
 &
 \proj &\colon \den{\Gamma,\Delta,\Gamma'} \to \den{\Gamma,\Gamma'}
 \,,
 \\
 \distribute &\colon \den{p\scaling\Gamma} \to p\scaling\den\Gamma\,, 
 &
 \semweak &\colon  \den{\Gamma + \Gamma'} \to \den \Gamma \,,
\end{align*}
easily defined by induction on $\Gamma$ and $\Gamma'$ using the morphisms $c$, $m$, $\kappa$ of 
Theorem~\ref{thm:graded:comonad} and projections. Note that neither of these is generally an isomorphism,
but $\distribute$ is when $r\geq 1$.
Using these, one can interpret function application
$\den{t\, u}$ as the composition 
\begin{equation*}
	\den{\Gamma + r\scaling \Gamma'} 
		\xrightarrow{\;(\den\Gamma\mprod \distribute) \circ \splt \;} 
	\den\Gamma\mprod r \den{\Gamma'} 
		\xrightarrow{\;\den t\mprod r\den u\;}
	\den{A \lol_r B} \mprod r\den A
		\xrightarrow{\;\mathsf{ev}\;} 
	\den{B} \,.
\end{equation*}
A similar argument can be used for the interpretation of most other constructions in the language, including $t \oplus_p u$ which is interpreted
using the IB algebra structure on $\D\den{A}$. Let binding is interpreted using Proposition~\ref{prop:D:free:IB}, and for the 
interpretation of natural number recursion, in the inductive case we use the fact that $\infty X$ is discrete for any $X$, and
so can be copied $\infty X \to \infty X \mprod \infty X$.

\begin{remark}[Independence on the choice of derivation] \label{rem:independence}
Formally, the above recipe gives an interpretation of a \emph{derivation} of a typing judgement. Ideally, one would like that $\den{\hastype\Gamma tA}$ is defined only
in terms of $\Gamma$, $t$ and $A$, but since a typing judgement can have many derivations, this is not \textit{a priori} clear. The main reason that judgements 
can have many derivations is that a given context $\Gamma$ can be split as a sum $\Gamma' + \Gamma''$ in many different ways. We prove that the interpretation 
of judgements is independent of the derivation by induction on terms. In order to do that, however, we must annotate terms with enough information to infer the 
types of all subterms from $\Gamma, t$ and $A$ in judgements $\hastype\Gamma tA$. This is not possible for the syntax given above. For example, for 
$\hastype\Gamma{\pi_1(t)}A$ to have a derivation $t$ must have type $A \times B$ for some $B$, but different derivations could use different $B$, which prevents the application of the induction hypothesis. The following meta-theoretic statements about the calculus should be read as statements about this `official' annotated
syntax. However, when writing terms, we will use the informal syntax presented above; it will always be the case that the annotations can be inferred.
\end{remark}

\begin{theorem} \label{thm:SemIndependence}
 The interpretation of typing judgements $\den{\hastype\Gamma tA}$ is well-defined and independent of the choice of derivation.
\end{theorem}

The interpretation is sound in the following sense.
\begin{theorem}[Soundness] \label{thm:Soundness:Term:Interpr}
 If $\hastype{\Gamma}{t\jeq u}A$ then $\den t = \den u$.
\end{theorem}

The proof of soundness relies on the following lemmas.

\begin{lemma}[Semantic Weakening] \label{lem:sem:weak}
For the derivations of Lemma~\ref{lm:weakening}, the following hold 
\begin{enumerate}
    \item $\den{\hastype{\Gamma,\Delta,\Gamma'}tA} = \den{\hastype{\Gamma,\Gamma'}tA} \circ \proj$;
    \item $\den{\hastype{\Gamma+\Delta}tA} = \den{\hastype{\Gamma}tA} \circ \semweak$.
\end{enumerate}
\end{lemma}

\begin{lemma}[Semantic Substitution] \label{lem:sem:subst}
If $\hastype{\Gamma, x:^rA, \Gamma'}{t}{B}$ and $\hastype{\Delta}{u}{A}$, for the derivation of Lemma~\ref{lm:substitution}, the following holds 
\[
  \den{\hastype{(\Gamma,\Gamma') + r \scaling \Delta}{t[u/x]}{B}}
  = \den{t} \circ 
  (\den{\Gamma} \mprod (r \den{u} \!\circ\! \distribute)\mprod \den{\Gamma'})
  \circ \splt \,.
\]
\end{lemma}

\section{Logic} \label{sec:logic}

In this section, we introduce a higher-order logic to reason about
the terms of the calculus. 
Compared to standard logics, which have a Boolean semantics, our logic is interpreted over the commutative unital quantale 
\begin{equation*}
    \Prop = (\uinterval, \geq, \oplus, \lol, 0)
\end{equation*}
with truncated sum $x \oplus y = \min \{x+y, 1\}$ as tensor, unit $0$, and adjoint $x \lol y = \max \{y - x, 0 \}$ defined as truncated reversed subtraction. Observe that the order in $\Prop$ corresponds to the reverse order in $\uinterval$: the bottom element is $1$, the top element is $0$, meet is $\sup$ and join is $\inf$. This fits with the idea of interpreting equality as distance in metric spaces: The 
logical statement $t \peq u$ is true in the model if its interpretation as the distance between $t$ and $u$ is $0$.

The formulas of the logic are well-typed \emph{predicates} in our calculus. Formally, we extend the calculus with $\Prop$ as base type with usual Euclidean distance on $\uinterval$ and add
\begin{align*}
 \varphi,\psi \ &{} ::= {}
  \true \mid \false \mid t \peq_A u \mid \varphi\ltensor\psi \mid \varphi\lexp\psi \mid r\psi 
  \mid \neg\varphi \mid \varphi\wedge\psi \mid \varphi\vee\psi \mid
  \exists x:A.\varphi \mid \forall x:A.\varphi
\end{align*}
to the syntax of terms, with typing rules as in Figure~\ref{fig:predicates}.
The formulas are those of a higher-order logic with equality, but extended with connectives specific to the quantale $\Prop$: $r\psi$ is interpreted as ($1$-bounded) rescaling by a factor $r\geq 0$, $\varphi\ltensor\psi$ as the tensor, and $\varphi\lexp\psi$ as its adjoint. 
Note that the latter are not new to our logic, but common connectives of fuzzy logics 
(e.g., {\L}ukasiewicz logic~\cite{Lukasiewicz,tarski1983logic} or the logic of Riesz MV algebras~\cite{NolaL11}).

For readability, we will often omit type annotations from $\peq_A$ and quantifiers when they can be inferred from the context.

\begin{figure}[t]
\small
\centering
\begin{mathpar}
\infrule{}{\hastype{\Gamma}{\true}{\Prop}}
\and
\infrule{}{\hastype{\Gamma}{\false}{\Prop}} 
\and
\infrule{
    \hastype{\Gamma}{t}{A}
    &\hastype{\Gamma'}{s}{A}
}{ \hastype{\Gamma+\Gamma'}{t \peq_A u}{\Prop} }
\\
\infrule{
    \hastype{\Gamma}{\varphi}{\Prop}
    &\hastype{\Gamma'}{\psi}{\Prop} 
}{ \hastype{\Gamma+\Gamma'}{\varphi \ltensor \psi}{\Prop} }
\and
\infrule{
    \hastype{\Gamma}{\varphi}{\Prop}
    &\hastype{\Gamma'}{\psi}{\Prop} 
}{ \hastype{\Gamma+\Gamma'}{\varphi \lexp \psi}{\Prop} }
\and
\infrule{
    \hastype{\Gamma}{\varphi}{\Prop} 
}{ \hastype{r\scaling\Gamma + \Gamma'}{r\varphi}{\Prop} }
\and
\infrule{
    \hastype{\Gamma}{\varphi}{\Prop}
}{ \hastype{\Gamma}{\neg\varphi}{\Prop} }
\and
\infrule{
    \hastype{\Gamma}{\varphi}{\Prop} & 
    \hastype{\Gamma}{\psi}{\Prop}
}{ \hastype{\Gamma}{\varphi \wedge \psi}{\Prop} }
\and
\infrule{
    \hastype{\Gamma}{\varphi}{\Prop} & 
    \hastype{\Gamma}{\psi}{\Prop}
}{ \hastype{\Gamma}{\varphi \vee \psi}{\Prop} }
\and
\infrule{
    \Gamma, x:^\infty A \vdash \varphi : \Prop }{
\Gamma \vdash \exists x: A.\varphi : \Prop}
\and
  \infrule{
    \Gamma, x:^\infty A \vdash \varphi : \Prop }{
\Gamma \vdash \forall x: A.\varphi : \Prop}
\end{mathpar}
    \caption{Typing rules for logical predicates.}
    \label{fig:predicates}
\end{figure}

We also add $\Prop$ to the grammar of IB algebra types, 
with operation defined as
\[
  \phi \oplus_p \psi \defeq p\phi \ltensor (1- p)\psi 
\]
and add the axioms of IB algebras (Definition~\ref{def:IB:alg}) to the judgemental equality theory also for this derived operator on predicates.

The interpretation of predicates is defined in Figure~\ref{fig:predicates:interp}. 
These are well-defined morphisms in $\CMet$ because of the following.
\begin{lemma} The following are non-expansive maps (where $r > 0$)
\begin{align*}
 {\oplus}, {\lol} &\colon  \Prop \mprod \Prop \to \Prop &
 \min\{r\cdot -, 1\} &\colon r \scaling \Prop \to \Prop  \\
 {\sup}, {\inf} &\colon (\infty \scaling X \mexp \Prop) \to \Prop &
 d_X &\colon X \mprod X \to \Prop \\
 \max, \min & \colon \Prop \times \Prop \to \Prop
\end{align*} 
\end{lemma}

Observe that, after the extension of the calculus, independence of the choice of the derivation (Theorem~\ref{thm:SemIndependence}) and soundness (Theorem~\ref{thm:Soundness:Term:Interpr}) are still valid, as well as weakening and substitution lemmas for typing judgments (Lemmas~\ref{lm:weakening} and \ref{lm:substitution}) and their semantic variants (Lemmas~\ref{lem:sem:weak} and \ref{lem:sem:subst}).

\smallskip
A \emph{predicate in context $\Gamma$} is a term $\phi$ such that $\hastype{\Gamma}{\phi}{\Prop}$. 
Observe that predicates can be constructed not just using 
logical connectives
but also via the other constructions in our calculus. For example, since
$\Prop$ is an IB algebra,  
if $\phi$ is a predicate in context $\Gamma, x:^r A$
and $\mu : \D A$, then ($\letIn{x = \mu}{\phi}$) is also a predicate. 

Logical reasoning on the terms of the calculus is done via the inference of logical judgments. The judgments of the logic are of the form 
\[
\logicJ[\Delta]{\Psi}{\varphi} \,,
\]
where $\Delta$ is a typing context, 
$\Psi = \psi_1, \dots, \psi_n$ is a list of predicates (\emph{logical context}), and $\varphi$ a predicate (\emph{conclusion}).
%

\begin{figure}[t]
\begin{align*}
\begin{aligned}
\den\true &\defeq 0 \\
\den\false &\defeq 1 \\
\den{t \peq_A u}
    &\defeq d_{\den A}\circ (\den{t} \mprod \den{s}) \circ \splt \\
\den{\varphi \ltensor \psi} 
    &\defeq \oplus \circ (\den{\varphi} \mprod \den{\psi}) \circ \splt \\
\den{\varphi \lexp \psi} 
    &\defeq \lol \circ (\den{\varphi} \mprod \den{\psi}) \circ \splt \\
\den{r\varphi} 
    &\defeq \begin{cases}
    0 & (r=0) \\
    {\min} \{r \cdot -, 1\} \circ r\den{\varphi}\circ \distribute 
    & (r > 0)
    \end{cases}
\end{aligned}
&&
\begin{aligned}
\den{\neg\varphi} 
    &\defeq 1-\den{\varphi} \\
\den{\varphi \wedge \psi}
    &\defeq \max \circ \tuple{\den{\varphi},\den{\psi}} \\
\den{\varphi \vee \psi} 
    &\defeq \min \circ \tuple{\den{\varphi},\den{\psi}} \\
\den{\exists x: A. \varphi} 
    &\defeq {\inf} \circ \curry(\den{\varphi}) \\
\den{\forall x: A. \varphi} 
    &\defeq {\sup} \circ \curry(\den{\varphi})
\end{aligned}
\end{align*}
    \caption{Interpretation of logical predicates.}
    \label{fig:predicates:interp}
\end{figure}

Hereafter, we always assume to work with well-formed logical judgments:
\begin{definition}[Well-formed judgments]
A logical judgment $\logicJ{\Psi}{\varphi}$ is \emph{well-formed} if
\begin{itemize}
    \item $\Delta$ is a discrete context, \ie, all variables in $\Delta$ have sensitivity annotation $\infty$. 
    \item all occurring predicates are well-typed in context $\Delta$, \ie, 
    $\hastype{\Delta}{\varphi}{\Prop}$ and $\hastype{\Delta}{\psi}{\Prop}$, for all $\psi\in\Psi$.
\end{itemize}
\end{definition}
The reason for the first condition is that sensitivity factors for term variables in logical predicates are irrelevant for logical judgments, and 
keeping track of these adds unnecessary complications to the logic. For example, allowing more general $\Delta$, for many
rules it would not be the case that well-formedness of the assumptions implies well-formedness of the conclusion, nor vice-versa. 
This would also raise meta-theoretic questions about the logic that we prefer to avoid. 
Note that by Lemma~\ref{lm:weakening},
$\infty$ is the most general 
sensitivity annotation possible: If $\hastype\Delta{t}A$ then also $\hastype{\Delta'}{t}A$ where $\Delta'$ is obtained from $\Delta$ by setting all
sensitivity annotations to $\infty$. 
In logical judgements we use the notation $\Delta, x:A$ as shorthand for the rigorous 
$\Delta, x:^\infty A$. 

\begin{figure*}
\small
\centering
\begin{mathpar}
\infrule[true]{}{\logicJ{\Psi}{\true}}
\and
\infrule[false]{
}{\logicJ{\Psi,\false}{\varphi}}
\and
\infrule[ass]{}{\logicJ{\Psi, \varphi}{\varphi}}
\\
\infrule[ex]{
        \logicJ{\Psi, \varphi, \psi, \Psi'}{\rho}}{
    \logicJ{\Psi, \psi, \varphi, \Psi'}{\rho}}
\and
\infrule[pr]{
        \logicJ{\Psi}{\varphi} }{
    \logicJ{r\Psi}{r\varphi}}
\and
\doubleinfrule[dup]{
        \logicJ{\Psi, (r+s)\varphi}{\psi} }{
    \logicJ{\Psi, r\varphi , s\varphi}{\psi}}
\and
\doubleinfrule[der]{
        \logicJ{\Psi, \psi}{\varphi}}{
    \logicJ{\Psi, 1\psi}{\varphi}}
\and
\infrule[zcon]{
        \logicJ{\Psi, 0\psi}{\varphi}}{
    \logicJ{\Psi}{\varphi}}
\and
\infrule[inc]{
        \logicJ{\Psi, r\psi}{\varphi}
        & r \leq s }{
    \logicJ{\Psi, s\psi}{\varphi}}
\and
\infrule[assoc${}_1$]{
        \logicJ{\Psi, r(s\psi)}{\varphi}}{
    \logicJ{\Psi, (rs)\psi}{\varphi}}
\and
\infrule[assoc${}_2$]{
        \logicJ{\Psi, (rp)\psi}{\varphi} 
        & p\leq 1 \text{ or } r \geq 1 }{
    \logicJ{\Psi, r(p\psi)}{\varphi}}
\and 
\infrule[g-rec]{
        \logicJ{(1-p)\Psi, p\varphi}{\varphi} 
        & p < 1}{ 
    \logicJ{\Psi}{\varphi} }
\and
\infrule[$\ltensor$-i]{
        \logicJ{\Psi}{\varphi} &
        \logicJ{\Psi'}{\varphi'} }{
    \logicJ{\Psi,\Psi'}{\varphi \ltensor \varphi'} }
\and
\infrule[$\ltensor$-e]{
        \logicJ{\Psi, \varphi, \psi}{\rho} }{
    \logicJ{\Psi, \varphi \ltensor \psi}{\rho} }
\and
\infrule[$\lexp$-i]{
        \logicJ{\Psi, \varphi}{\psi} }{
    \logicJ{\Psi}{\varphi \lexp \psi} }
\and
\infrule[$\lexp$-e]{
        \logicJ{\Psi}{\varphi \lexp \psi}
        & \logicJ{\Psi'}{\varphi} }{
    \logicJ{\Psi,\Psi'}{\psi} }
\and
\infrule[$\neg$-i]{
        \logicJ{\Psi,\varphi}{\false}
    }{\logicJ{\Psi}{\neg\varphi} }
\and
\infrule[$\neg$-e]{
        \logicJ{\Psi, \neg\varphi}{\false}
    }{ \logicJ{\Psi}{\varphi} }
\and
\infrule[$\wedge$-i]{
        \logicJ{\Psi}{r\varphi}
        &\logicJ{\Psi}{r\psi}
    }{\logicJ{\Psi}{r(\varphi \wedge \psi)} }
\and
\infrule[$\wedge$-el]{
        \logicJ{\Psi}{\varphi \wedge \psi} 
    }{ \logicJ{\Psi}{\varphi} }
\and
\infrule[$\wedge$-er]{
     \logicJ{\Psi}{\varphi \wedge \psi}
    }{ \logicJ{\Psi}{\psi} }
\and
\infrule[$\vee$-il]{
        \logicJ{\Psi}{\varphi}
    }{ \logicJ{\Psi}{\varphi \vee \psi} }
\and
\infrule[$\vee$-ir]{
        \logicJ{\Psi}{\psi}
    }{ \logicJ{\Psi}{\varphi \vee \psi} }
\and
\infrule[$\vee$-e]{
        \logicJ{\Psi,r\varphi}{\rho}
        &\logicJ{\Psi,r\psi}{\rho}
    }{\logicJ{\Psi,r(\varphi \vee \psi)}{\rho} }
\and
\infrule[$\exists$-i]{
    \Delta \vdash t : A
    &\logicJ[\Delta]{\Psi}{\varphi[t/x]}
}{\logicJ[\Delta]{\Psi}{\exists x:A. \varphi}}
\and
\infrule[$\exists$-e]{
    \logicJ[\Delta, x:A]{\Psi, r\varphi}{\psi} & r < \infty
}{ \logicJ[\Delta]{\Psi, r( \exists x:A. \varphi)}{\psi} }
\and
\infrule[$\forall$-i]{
    \logicJ[\Delta, x:A]{\Psi}{r\varphi}
}{ \logicJ[\Delta]{\Psi}{r(\forall x:A.\varphi)} }
\and
\infrule[$\forall$-e]{
    \logicJ[\Delta]{\Psi}{\forall x:A. \varphi}
    & \Delta \vdash t : A
}{\logicJ[\Delta]{\Psi}{\varphi[t/x]} }
\and
\infrule[eq-i]{
    \hastype{\Delta}{t \jeq s}{A}
}{\logicJ[\Delta]{\Psi}{t=_A s} }
\and
\infrule[eq-e]{
    \hastype{\Delta, x:^r A}{\varphi}{\Prop}
    &\hastype{\Delta}{t}{A} 
    &\hastype{\Delta}{u}{A}
    &\logicJ[\Delta]{\Psi}{\varphi[t/x]}
    &\logicJ[\Delta]{\Psi'}{r(t=_A u)}
}{\logicJ[\Delta]{\Psi, \Psi'}{\varphi[u/x]} }
\and
\infrule[ind${}_\mprod$]{
        \logicJ[\Delta, x : A, y : B]{\Psi}{\varphi[(x,y)/z]} & \hastype{\Delta}{t}{A\! \pair{r}{s} B}} 
   {  \logicJ{\Psi}{\varphi[t/z]} }
\and
\infrule[ind${}_\nat$]{
      \logicJ[\Delta]{\Psi}{\varphi[\Zero/n]} & 
   \logicJ[\Delta, n:\nat]{\varphi}{\varphi[\Succ{n}/n]}
&
\hastype\Delta t\nat
}{\logicJ[\Delta]{\Psi}{\varphi[t/n]} }
   \and
\infrule[ind${}_+$]{
        \logicJ[\Delta, x : A]{\Psi}{\varphi[\inj_1(x)/z]} &
         \logicJ[\Delta, y : B]{\Psi}{\varphi[\inj_2(y)/z]} &
         \hastype{\Delta}{t}{A + B}} 
   {  \logicJ{\Psi}{\varphi[t/z]} }
\and
\infrule[ind${}_\D$]{
    \begin{array}{r@{}l}
        r &{}< \infty \\
        \hastype{\Delta}{{}&t}{\D A} \\
        \hastype{\Delta, x:^r \D A}{{}&\varphi}{\Prop}
    \end{array}
    &
    \begin{array}{r@{}l}
    \\
    \logicJ[\Delta, y:A]{\Psi}{{}&\varphi[\delta y/x]} \\
    \forall p.\, \big(
    \logicJ[\Delta, \mu : \D A ,\nu : \D A]{p\varphi[\mu/x], (1-p)\varphi[\nu/x]}{{}&\varphi[\mu\oplus_p\nu/x] } \big)
    \end{array}
}{\logicJ[\Delta]{\Psi}{\varphi[t/x]} }
\end{mathpar}
\caption{Logic. (The logical judgments appearing above are assumed to be well-formed) 
}
\label{fig:logicrules}
\end{figure*}

The inference rules for the logic are given in Figure~\ref{fig:logicrules}. The notation $p\Psi$ means to multiply each proposition in $\Psi$ by $p$. Rules given by double line are double rules, so can be used in both directions.

The logic is sound with respect to the semantic interpretation of logical judgments in the following sense, where the notation $\den{\psi_1, \dots, \psi_n}$ means $\den{\psi_1} \oplus \dots \oplus \den{\psi_n}$.
\begin{theorem}[Soundness] \label{thm:Soundness:Logic}
If $\logicJ[\Delta]{\Psi}{\varphi}$ is derivable,
then $\den{\Psi}(\delta) \geq \den{\varphi}(\delta)$ for all $\delta\in\den\Delta$.
\end{theorem}

Note the similarity between the rule (\textsc{g-rec}) for guarded recursion  and the typing rule for fixed points. 
The logic includes the classical rule (\textsc{$\neg$-e}), because it is verified by the model, but our examples below do not use it. 
The rule (\textsc{eq-e}) for elimination of equality is perhaps most easily understood via a sketch proof of soundness: The sensitivity
of $\varphi$ in $x$ ensures that $\varphi[s/x]$ is at most at distance $r(t \peq s)$ from $\varphi[t/x]$. Therefore, since 
$\den{\Psi} \geq \den{\varphi[t/x]}$ and $\den{\Psi'} \geq \den{r(t \peq s)}$, also 
$\den{\Psi} \oplus \den{\Psi'} \geq \den{\varphi[s/x]}$. 

One consequence of the induction principles (\textsc{ind}${}_\mprod$), (\textsc{ind}${}_+$), (\textsc{ind}${}_\nat$), and (\textsc{ind}${}_\TD$) is that judgements involving 
predicates defined using recursion and let-bindings can be proved. For example, by (\textsc{ind}${}_\mprod$), to prove 
\[\logicJ[\Delta]{\Psi, \letIn{(x,y) = t}{\phi}}{\letIn{(x,y) = t}{\psi}},\] it suffices to show that 
$\logicJ[\Delta, x: A, y: B]{\Psi,\phi}{\psi}$. The induction principle for $\nat$ requires the induction
case to be proven not using $\Psi$, so that $\Psi$ is only used once, at the base case. 
The use of convex combinations in the hypothesis of the induction principle for $\TD$ ensures that any probability measure constructed inductively from Dirac distributions and convex combinations satisfies the conclusion. A priori, these only give the finite distributions
and, semantically, $\TD$ also contains continuous distributions. However, the finitely supported distributions are dense, and the requirement that $r$ be finite means that $\varphi$ is continuous in $x$, which suffices to verify soundness.

\begin{remark} \label{rem:finiteTD-Induction}
The universal quantification over $p$ in \textsc{(ind${}_\TD$)} makes it an infinitary rule. If one wants
a finitary proof system, the induction case can be replaced by
%
$\frac{1}{2} \varphi[\mu/x], \frac{1}{2}\varphi[\nu/x] \vdash \varphi[\mu \oplus_{\frac12} \nu/x]$. 
This principle is still sound as the finitely supported dyadic distributions are also dense in $\D X$ 
(see e.g., proof of \cite[Theorem 21]{BreugelHMW07}). We prefer to keep the rule as stated in Figure~\ref{fig:logicrules},
because it follows the standard induction principle for inductive types: One proof obligation for each constructor. 
\end{remark}

The rules (\textsc{$\wedge$-I}), (\textsc{$\vee$-e}), (\textsc{$\exists$-e}), and (\textsc{$\forall$-I}) include scaling factors in the conclusions. This is to recover distributivity of scalar multiplication over the logical connectives, which cannot be otherwise derived due to the necessary restrictions in (\textsc{assoc${}_2$}).

\begin{lemma} \label{lem:r:exists}
 For all $r \geq 0$, the predicates $r(\varphi \wedge \psi)$, $r(\varphi \vee \psi)$, and $r(\forall x : A. \varphi)$ are respectively equivalent to $r\varphi \wedge r\psi$, $r\varphi \vee r\psi$, and $\forall x : A. r\varphi$.  
 If $r < \infty$, then $r(\exists x:A . \varphi)$ is equivalent to $\exists x : A. r\varphi$.
\end{lemma}

The condition that $r < \infty$ in the last statement is necessary: $\infty (\exists x : (0,1] . x \peq 0)$ is interpreted as $0$ in the model, but $\exists x : (0,1] . \infty (x \peq 0)$ is $1$.

The logic is affine but not relevant, that is, weakening is derivable but contraction is not. 
\begin{lemma} \label{lem:logic:weak}
 If $\logicJ[\Delta]{\Psi}{\varphi}$ is derivable, so is $\logicJ[\Delta]{\Psi, \psi}{\varphi}$.
\end{lemma}

Moreover, derivability is closed under weakening of the typing context and term substitution in predicates.
\begin{lemma} \label{lem:weak} \
\begin{enumerate}
\item  If $\logicJ[\Delta]{\Psi}{\varphi}$ is derivable, then so is $\logicJ[\Delta, x:A]{\Psi}{\varphi}$;
\item If $\hastype\Delta tA$ and $\logicJ[\Delta, x:A]{\Psi}{\varphi}$ is derivable, then so is $\logicJ[\Delta]{\Psi[t/x]}{\varphi[t/x]}$.
\end{enumerate}
\end{lemma}

\section{Proving basic properties}
\label{sec:basic:prop}

We show some basic consequences of the rules of the logic, focusing on equality. 

We first show that the rule  for equality elimination implies that equality is symmetric and transitive. The derivations mirror those of the corresponding rules for identity types in type theory (see e.g.~\cite{hottbook}), although 
the setting of affine logic used here is different. 

\begin{proposition} \label{prop:eq:symm:trans}
 Propositional equality is symmetric and transitive in the sense that 
\begin{align*}
 x : A, y: A \mid r(x\peq y)  \ts r(y\peq x) \\
 x:A, y:A, z:A \mid r(x\peq y), r(y\peq z)  \ts r(x \peq z)
\end{align*}
\end{proposition}
\begin{proof}
We just prove transitivity. Let $\Delta \defeq x:A, y:A, z :A$ and apply (\textsc{eq-e}) to $\hastype{\Delta, w :^r A}{r(x\peq w)}{\Prop}$. 
Because we have 
  $\logicJ{r(x\!\peq \!y)}{r(x\!\peq\! y)}$
we obtain 
  $\logicJ{r(x\!\peq\! y), r(y\!\peq\! z)}{r(x\! \peq\! z)}$. 
\end{proof}
Note that since the comma is interpreted as $1$-bounded addition, in the case where $r=1$, transitivity corresponds to the triangle inequality. 
Equality is also a congruence relation.

\begin{proposition}[Congruence] \label{prop:eq:congruence}
 Let $\hastype\Delta uA$ and $\hastype\Delta vA$. If $\hastype{\Delta, x:^rA}{t}B$, then $\logicJ{r(u\peq v)}{t[u/x] \peq t[v/x]}$.
\end{proposition}
\begin{proof}
 Apply (\textsc{eq-e}) to $\hastype{\Delta, y:^r A}{t[u/x]\peq t[y/x]}{\Prop}$.
\end{proof}
As a special case, for $\Delta = x:\TD A, y:\TD A, z :\TD A, w:\TD A$, we have
\begin{equation} \label{eq:oplus:peq}
  \logicJ{p(x \peq y), (1-p)(z \peq w)}{x \oplus_p z \peq y \oplus_p w} \,.
\end{equation}

Using (\textsc{ind}${}_\mprod$) and (\textsc{ind}${}_+$), we can prove 
the following extensionality principles.

\begin{lemma} \label{lem:ext:tensor}
 If $x,y : A \pair{r}{s} B$ then $x \peq y$ is equivalent to 
 \[\letIn{(a,b) = x, (a',b') = y}{r\scaling(a\peq a') \ltensor s (b \peq b')}.\] 
\end{lemma}
\begin{lemma} \label{lem:ext:sum}
 If $x,y : A + B$ then $x \peq y$ is equivalent to 
\begin{equation*}
\big(\exists a,a'. (x \peq \inj_1 a) \ltensor (y \peq \inj_1 a') \ltensor (a \peq a')\big)
\vee 
\big(\exists b,b'. (x \peq \inj_2 b) \ltensor (y \peq \inj_2 b') \ltensor (b \peq b')\big)
\end{equation*}
\end{lemma}

Since $\Prop$ is a type in our calculus, the logic is higher order. We can define types of predicates
and relations
\begin{align}
 \PredQ Ar & \defeq A \lol_r \Prop & 
 \RelQ ABrs & \defeq A \pair rs B\lol_1 \Prop \,.
 \label{eq:Rel}
\end{align}
We simply write $\Pred A$ and $\Rel AB$
as shorthands for $\PredQ A1$ and $\RelQ AB11$, respectively, when all the sensitivities are $1$.

One can prove that equality is equivalent to Leibniz equality. 
\begin{proposition} \label{prop:Leibniz}
 The predicates $\forall \phi : \PredQ Ar . \phi(x) \lexp \phi(y)$ and $r(x \peq_A y)$ are equivalent. 
\end{proposition}

\subsection{Internalising the Kantorovich distance}

\newcommand{\eq}{\mathsf{eq}}

We show that the Kantorovich distance (\ref{eq:Kantorovic}) can be internalised in the logic. 
This, in turn, will enable reasoning via couplings to establish equalities between probability distributions.
The integral in the definition of the Kantorovich distance computes
the mean of the distance $d_X(x,y)$ when $(x,y)$ is distributed according to the given coupling $\omega$. We start by defining, more
generally, the mean $\mean x\mu\phi$ of a predicate $\phi : \Pred A$ over a distribution $\mu:\TD A$ 
as 
\begin{equation} \label{eq:mean:def}
 \mean x\mu\phi \defeq \letIn{x = \mu}{\phi(x)}
\end{equation}
This defines a mean because it satisfies the equations
\begin{align}
  \mean x{\delta(y)}\phi &\jeq \phi(y) 
  \label{eq:mean:dirac} \\
  \mean x{\mu \oplus_p \mu'}\phi &\jeq 
  p (\mean x\mu\phi) \ltensor(1-p) (\mean x{\mu'}\phi) 
  \,. \label{eq:mean:conv}
\end{align}
For $R : \Rel AB$ and $\omega : \TD(A \pair11 B)$ we write $\mean{(a,b)}\omega{R(a,b)}$ rather than 
$\mean{z}\omega{R}$ to emphasize that $R$ is a relation. 
When $R$ is the equality relation we write 
$\mean{(x,y)}\omega{x\peq y}$. 
The notion of $R$-coupling can be expressed quantitatively in the logic as 
\begin{align}
    \label{def:coup}
    \coup{}{\omega}{\mu}{\nu} 
    &\defeq (\TD(\pi_1)\omega \peq \mu) \ltensor  (\TD(\pi_2)\omega \peq \nu) \,,
    \\
    \label{def:R:coup}
    \coup{R}{\omega}\mu\nu 
    &\defeq \mean{(x,y)}\omega{R(x,y)}\ltensor \coup{}{\omega}{\mu}{\nu} \,.
\end{align}
Notably, $\coup{}{\omega}{\mu}{\nu}$ evaluates to $0$ (the distinguished truth value in $\Prop$) if and only if $\omega$ is a coupling between $\mu$ and $\nu$. However, when it evaluates to a value different from $0$, $\omega$ is not necessarily a coupling. Thus, \eqref{def:coup} captures a notion that is strictly weaker and more permissive than that of an actual coupling.

Finally, using that existential quantification is interpreted as infimum, we can internalise the Kantorovich distance as follows
\begin{equation*} \label{eq:internal:Kantorovic}
 \kant\mu \nu 
 \defeq \exists \omega : \TD(A \pair11 A) 
 .\, \coup{\eq}{\omega}{\mu}{\nu} \,,
\end{equation*}
where $\eq \defeq \lambda z. \letIn{(x,y)=z}{x=y} : \Rel AA$ 
is the equality relation.

\begin{theorem} \label{thm:internal:Kantorovic}
 The predicates $\kant\mu\nu$ and $\mu\peq \nu$ are equivalent.
\end{theorem}

\begin{proof}
 To prove that $\mu\peq \nu$ implies $\kant\mu\nu$, by (\textsc{eq-e}) it suffices to show
 \[
 \logicJ[\mu : \TD A]{\cdot}{\kant\mu\mu}
 \]
 Define $\hastype{\mu:^2 \TD A}{\omega(\mu)}{\TD(A \pair11 A)}$ as
 $\letIn{a = \mu}{\delta(a,a)}$. Then $\TD\pi_i(\omega(\mu)) \peq \mu$ for $i=1,2$.
Moreover, $\cdot \ts \mean{(a,a')}{\omega(\mu)}{a \peq a'}$ can be proved by induction on $\mu$ as follows. If $\mu = \delta(a)$, by \eqref{eq:mean:dirac} and $\omega(\mu) \jeq \delta(a,a)$, $\mean{(a,a')}{\omega(\mu)}{a \peq a'}$ reduces to $a\peq a$, which is true. If $\mu = \mu_1 \oplus_p \mu_2$, we must show
$\mean{(a,a')}{\omega(\mu_1\oplus_p \mu_2)}{a \peq a'}$ in context
\[
 p(\mean{(a,a')}{\omega(\mu_1)}{a \peq a'}), (1-p)(\mean{(a,a')}{\omega(\mu_2)}{a \peq a'}) 
\]
which holds by \eqref{eq:mean:conv} because $\omega(\mu_1\oplus_p \mu_2) \jeq \omega(\mu_1)\oplus_p \omega(\mu_2)$.
For the other direction, it suffices to show that $\mean{(x,y)}\omega{x\!\peq \!y}$ implies
$\TD(\pi_1)(\omega) \peq \TD(\pi_2)(\omega)$, for any $\omega : \TD(A \pair11 A)$. This can be done by induction on $\omega$. 
\end{proof}

\subsection{Uniqueness of fixed points}

One might hope that the uniqueness of fixed points from the Banach fixed point theorem can be internalised in our logic as the statement
\[ 
  x\peq f(x) , y \peq f(y) \ts x \peq y
\]
whenever $f : X \lol_p X$ for $p<1$. However, this is not true. Take for example $f : \Prop \lol_{\frac12} \Prop$ to be multiplication by $\frac12$, 
$x$ to be $0$ and $y$ to be $1$. Then the semantics of the above statement evaluates to the false statement $\frac 12 \geq 1$. However, if we
assume that $x$ is a fixed point for $f$ in the global sense, then it equals the unique fixed point. 

\begin{lemma} \label{lem:just:one:fix}
 Let $f :^1 X \lol_p X$ for $p < 1$. Then, $\true \ts x \peq f(x)$ implies $\true \ts x \peq \fix y f(y)$.
\end{lemma}
\begin{proof}
 By (\textsc{g-rec})
 , it suffices to prove that
 $
 p(x \peq \fix y f(y)) \ts x \peq \fix y f(y)$. This follows from $\true \ts x \peq f(x)$ and transitivity of propositional equality (Proposition~\ref{prop:eq:symm:trans}) after observing that $p(x \peq \fix y f(y)) \ts f(x) \peq f(\fix y f(y))$, which is true by Proposition~\ref{prop:eq:congruence}.
\end{proof}


\section{Case study: Markov processes}
\label{sec:markov}

Markov processes describe systems with memoryless transitions between states, governed by probabilities. Formally, they consist of a set of states $\State$, a transition function $\State \to \D(\State)$ that specifies the probabilities of moving to the next state, and a labeling function $\State \to A$ that assigns labels to states. Following~\cite{Breugel05}, we treat $A$ as a metric space and analyze the behavior of Markov processes quantitatively using distances that discount future differences in observations by means of a discount factor $\discfact \in (0,1]$. The smaller $\discfact$ is, the more the focus shifts toward short-term differences. Categorically, this corresponds to interpreting Markov processes as coalgebras $S \to A \otimes c\D(S)$ in $\CMet$.
The behaviour of a state is abstractly characterised as an element of the final coalgebra%
\footnote{The final coalgebra always exists as $\semLabels\mprod\discfact\D(-)$ is an accessible functor. See ~\cite{BreugelHMW05,Breugel05} for details.}, corresponding
to the coinductive solution $\semProc[\discfact]$ to the functorial equation $\semProc \iso \semLabels\mprod\discfact\D(\semProc)$. The behavioral distance 
is just the distance in $\Proc$ between behaviours~\cite{BreugelHMW05}. 

In order to program with $\semProc$ we extend the calculus with types $\Proc$ and $\Labels$ as well as terms 
\begin{align*}
\unfold & : \Proc \lol_1 \Labels\pair1\discfact\D(\Proc) &
\fold & : \Labels \pair1\discfact \D(\Proc) \lol_1 \Proc
\end{align*}
We will write $\proc a{m}$ for $\fold\,(a,m)$. Finally we add judgemental equalities stating that $\fold$ and $\unfold$ are inverses of each other.

As a first example, consider a process $m$ satisfying the recursive definition $m \jeq \proc a{(\delta(m) \oplus_{\frac13} \delta(z))}$ where $z$ is some other given process. 
This recursive definition is productive in the sense that it only calls itself with probability $\frac13$. Therefore it can be defined as a term of type $\Proc[1]$ 
similarly to the definition of the geometric distribution of Example~\ref{example:geo:calc}. Precisely, because 
\[
  \hastype{z :^{\frac23} \Proc[1], m :^{\frac13}\Proc[1]}{\proc a{(\delta(m) \oplus_{\frac13} \delta(z))}}{\Proc[1]}
\]
we can define $\hastype {z :^1 \Proc[1]}{m}{\Proc[1]}$ as 
\[
  m\defeq \fix{m}{\proc a{(\delta(m) \oplus_{\frac13} \delta(z))}}
\]
which then by the equality for fixed point unfolding satisfies the desired equality. 

Now, let $n$ satisfying $n \jeq \proc a{(\delta(n) \oplus_{\frac12} \delta(z))}$ be defined similarly. Using the logic we will now show that
the distance between $m$ and $n$ is at most $\frac14$, which in the logic corresponds to showing that $\frac14\false \vdash m\peq n$. 
In the following, for $r > 0$, we simply write $r$ for $r \,\false$. 
By guarded recursion (\textsc{g-rec}), it suffices to show that
\[
 \frac23\cdot \frac14, \frac13(m\peq n) \ts m\peq n
\]
Since 
\begin{align*}
 n & \peq \proc a(\delta(n) \oplus_{\frac13}(\delta(n) \oplus_{\frac14}\delta(z))) & 
 m & \peq \proc a{(\delta(m) \oplus_{\frac13}(\delta(z) \oplus_{\frac14}\delta(z)))} 
\end{align*}
by (\ref{eq:oplus:peq}) and Proposition~\ref{prop:eq:congruence} it suffices to show 
\[
  \frac23\cdot \frac14, \frac13(m\peq n) \ts \frac13(m\peq n) \ltensor \frac23\left( \frac14(n \peq z) \ltensor \frac34(z \peq z)\right)
\]
which in turn reduces to the following three judgements, all of which are true:
\begin{align*} 
 m \peq n & \ts m \peq n &
 \false & \ts n \peq z &
 \true & \ts z \peq z
\end{align*}

\subsection{A biased coin tossing process}
\label{sec:c:markov}

The next example describes a probabilistic process generated by a coin toss with a biased coin,
where the current state remembers the result of the last coin toss.
The label space is the discrete 
set $\Labels = \{\hd, \tl\}$ and the two states should satisfy the mutually recursive equations
\begin{align*}
\begin{aligned}
 \begin{tikzpicture}[every node/.style={font=\scriptsize}]
    \node[draw,circle,fill=black,text=white,inner sep=1.1pt] (hd) at (0,0) {$\hd$}; 
    \node[draw,circle,fill=black,text=white,inner sep=2pt] (tl) at (2,0) {$\tl$};
    \path[->]
        (hd) edge [loop left] node[left]{$\frac12-\epsilon$} (hd)
             edge [bend right] node[below]{$\frac12+\epsilon$} (tl)
        (tl) edge [loop right] node[right]{$\frac12+\epsilon$} (tl)
             edge [bend right] node[above]{$\frac12-\epsilon$} (hd);
    \end{tikzpicture}
\end{aligned}
&&
\begin{aligned}
 \hdstate & \jeq \proc \hd{(\delta(\hdstate) \oplus_{\frac12-\epsilon} \delta(\tlstate))} \\
 \tlstate & \jeq \proc \tl{(\delta(\hdstate) \oplus_{\frac12-\epsilon} \delta(\tlstate))} \,.
\end{aligned}
\end{align*}
Unlike the previous example, this definition is not productive, so $\hdstate, \tlstate$ cannot be defined by guarded recursion 
as elements of $\Proc[1]$. However, they can be defined as elements of $\Proc[\discfact]$ for any $\discfact \in (0,1)$. 
We define them by mutual recursion using a fixed point of a contraction on $\Proc\times \Proc$, as
$\hdstate \defeq \pi_1(\hdtlstate)$ and $\tlstate \defeq \pi_2(\hdtlstate)$ for
\begin{align*}
\hdtlstate 
&\defeq \fix x {\tuple{\proc \hd{\flip_\epsilon(x)}, \proc \tl{\flip_\epsilon(x)}}} 
\end{align*}
where $\flip_\epsilon(x) \defeq \delta(\pi_1(x)) \oplus_{\frac12-\epsilon} \delta(\pi_2(x))$. Observe that $\hdtlstate$ is well-defined as 
\begin{equation*}
    \hastype{x :^\discfact \Proc\times \Proc}
    {\tuple{\proc \hd{\flip_\epsilon(x)}, \proc \tl{\flip_\epsilon(x)}}}{\Proc\times \Proc} \,.
\end{equation*}

Consider the special case of a fair coin $\hdfair \defeq \hdstate[0]$, $\tlfair \defeq \tlstate[0]$. 
We now show that the distance between $\hdfair$ and $\hdstate$ is at most 
$\frac{\discfact\epsilon}{1-\discfact+\discfact \epsilon}$. Logically, this corresponds to the statement 
\begin{equation} \label{eq:hdfair:hdstate}
 \frac{\discfact\epsilon}{1-\discfact+\discfact \epsilon} \ts \hdfair \peq \hdstate  
\end{equation}

The statement (\ref{eq:hdfair:hdstate}) must be proved simultaneously with a similar statement for the two tail
states, and by guarded recursion it suffices to prove that, for $\mathsf{d}_{c,\epsilon} \defeq \frac{\discfact\epsilon}{1-\discfact+\discfact \epsilon}$, we have
\[
 \discfact\left( \mathsf{d}_{c,\epsilon} \!\lexp\! \left(\hdfair\! \peq\! \hdstate \!\conj\! \tlfair\! \peq\! \tlstate\right) \right) 
 \ts 
 \left( \mathsf{d}_{c,\epsilon}\! \lexp\! \left(\hdfair \!\peq\! \hdstate \!\conj\! \tlfair\! \peq\! \tlstate\right) \right) \,.
\]
We show that 
\begin{equation} \label{eq:hd:hdfair:gr}
 \discfact\left( \mathsf{d}_{c,\epsilon} \lexp \left(\hdfair \peq \hdstate \conj \tlfair \peq \tlstate\right) \right),
  \mathsf{d}_{c,\epsilon} 
 \ts 
 \hdfair \peq \hdstate
\end{equation}
The case for $\tlfair \peq \tlstate$ is similar. 
By (\ref{eq:oplus:peq}) to show $\hdfair \peq \hdstate$ it suffices to show $\discfact$ times the formula 
\begin{equation} \label{eq:hd:hdfair:split}
  \left(\frac12-\epsilon\right)(\hdfair \peq \hdstate) \ltensor \epsilon(\hdfair \peq \tlstate) \ltensor \frac12(\tlfair \peq \tlstate)
\end{equation}
and so (\ref{eq:hd:hdfair:gr}) reduces by rule (\textsc{pr}) to showing (\ref{eq:hd:hdfair:split}) in context 
\begin{equation} \label{eq:hd:hdfair:reduced}
  \mathsf{d}_{c,\epsilon} \lexp \left(\hdfair \peq \hdstate \conj \tlfair \peq \tlstate\right) ,
  \frac{\epsilon}{1\!-\!\discfact\!+\!\discfact \epsilon}
\end{equation}
Since
\[
 \frac{\epsilon}{1\!-\!\discfact\!+\!\discfact \epsilon} = \left(\frac12-\epsilon\right)\mathsf{d}_{c,\epsilon} + \epsilon + \frac12 \mathsf{d}_{c,\epsilon}
\]
Using (\textsc{dup}) and (\textsc{inc})  it suffices to prove (\ref{eq:hd:hdfair:split}) in context 
\begin{equation*} 
  \left(\frac12-\epsilon\right)\left(\mathsf{d}_{c,\epsilon} \lexp \left(\hdfair \peq \hdstate \conj \tlfair \peq \tlstate\right)\right) ,
  \frac12\left(\mathsf{d}_{c,\epsilon} \lexp \left(\hdfair \peq \hdstate \conj \tlfair \peq \tlstate\right)\right) ,
  \left(\frac12-\epsilon\right)\mathsf{d}_{c,\epsilon},  \epsilon, \frac12 \mathsf{d}_{c,\epsilon}
\end{equation*}
which can be done using (\textsc{$\ltensor$-i}).

We remark that the upper bound shown above is tight, in the sense
that it is the actual behavioual distance between
the two states. This can be checked by direct calculations because, as observed in~\cite{BreugelW05}, it corresponds to the $\discfact$-discounted bisimilarity distance of Desharnais et al.~\cite{Desharnais00}.

\subsection{Bisimulation}

Let $\xi \colon X \to \semLabels \mprod \discfact \D(X)$ be a Markov process, and decompose $\xi$ into two maps: $\xi_1 \colon X \to \semLabels$ for labels and $\xi_2 \colon X \to \discfact \D(X)$ for probabilistic transitions. 
A probabilistic bisimulation~\cite{LarsenS91} for $\xi$ is a binary relation on $X$ such that
$R(x,y)$ implies (i) $\xi_1(x) = \xi_1(y)$, states have the same label; and (ii) that there exists an $R$-coupling $\rho$ (\ie, a
coupling $\rho$ for $\xi_2(x)$ and $\xi_2(y)$ whose support is in $R$) ensuring the probabilistic behaviors of $x$ and $y$ remain related under the bisimulation.

While exact equivalence is often too rigid, metrics offer a more flexible alternative. Next, we internalise the 
definition of bisimilarity in our logic. 
For $R : \Rel XX$, define 
\begin{align*}
 \Bisim{R} \defeq \forall x,y : X. R(x,y) \lexp 
    \letIn{&(l, \mu) = \xi(x), \,
    (l', \nu) = \xi(y)\\}{&
   l \peq l' \ltensor \discfact( \exists \rho . \coup{R}\rho\mu\nu}) 
\end{align*}
using the quantitative notion of $R$-coupling defined in (\ref{def:R:coup}). 

In the case where $\discfact \in (0,1)$, we can define the bisimilarity relation ${\bisim} : \Rel XX$ by guarded recursion as 
\begin{align*}
 {\bisim} \defeq \fix R \lambda x. \lambda y.  
 \letIn{(l, \mu) = \xi(x), \,
    (l', \nu) = \xi(y)}{
   l \peq l' \ltensor \discfact( \exists \rho . \coup{R}\rho\mu\nu})
\end{align*}
using that $R$ occurs with sensitivity $\discfact$ in the body of the fixed point. 
The following can be proved using guarded recursion. 

\begin{proposition} \label{prop:bisim:eq}
  Bisimilarity is equivalent to equality in $\Proc$.
\end{proposition}

\begin{proof}[Proof (sketch)] 
 We just sketch the proof of bisimilarity implying equality. The proof is by guarded recursion, 
 reducing to  
 $\discfact(\forall x,y . x\bisim y \lexp x \peq y), x \bisim y \vdash x\peq y$. The assumption $x \bisim y$ unfolds to
 \[
 \letIn{(l, \mu) = \unfold(x), \,
    (l', \nu) = \unfold(y)}{
   l \peq l' \ltensor \discfact( \exists \rho . \coup{\bisim}\rho\mu\nu})
 \]
 and using the guarded recursion hypothesis, the last part can be rewritten to 
 $\discfact(\exists \rho . \coup{\eq}\rho\mu\nu)$, which by Theorem~\ref{thm:internal:Kantorovic} is equivalent to 
 $\discfact(\mu \peq \nu)$. We conclude by Lemma~\ref{lem:ext:tensor}.
\end{proof}

Since the distance in $\Proc$ coincides with probabilistic bisimilarity distance, 
Proposition~\ref{prop:bisim:eq} shows that ${\bisim} : \Rel \Proc\Proc$ is interpreted as 
bisimilarity distance. 

\section{Case study: temporal learning}
\label{sec:temporal:learning}

The next example, adapted from Aguirre et al~\cite{BHKKM21}, showcases the expressivity of our calculus and a use of natural number induction. 

A \emph{Markov decision process} comprises a set of states $\State$, a set of actions $\Act$, a transition function $\trans \colon \State \times \Act \to \D(\State)$ and
a reward function $\reward \colon \State \times \Act \to \D([0,r])$. We will consider $\Act$ a discrete set, 
and assume that states $\State$ is a finite discrete set. Moreover, for simplicity we will assume that $\State$ is simply the set
$\State =\{0, \dots, N-1\}$ for some $N$, so that $\State \mexp X$ can be expressed as an $N$-fold product $X^N$. Finally, we will assume that 
$r=1$, so that we can give the type of $\reward$ internally in the calculus as $\State \pair11\Act \lol \TD(\Prop)$. The latter is not a 
restriction in practice, as the actual values of rewards is inessential, and so the reward function can be appropriately scaled. However, it is necessary, as we 
need the reward space to be an IB-algebra and a $1$-bounded metric space. 

When doing reinforcement learning, one must estimate a value function $V : \Prop^N$ mapping states to rewards, 
for a given policy $\policy : \Act^N$. \emph{Temporal difference}
is one approach to doing this, which works by iteratively refining the value function $V$ as follows:
For each state $i$, compute an action $a$ from the policy distribution $\policy(i)$, sample a reward
$r$ from the reward distribution  $\reward(a,i)$, and sample a transition $j$ from the transition function
$\trans(a,i)$. From this the updated value function $V'$ at $i$ can be defined as the convex combination
\[
 V'(i) = (1-\alpha)V(i) + \alpha(r + \gamma V(j))
\]
of the previous value $V(i)$ and the reward associated with the next state $j$, for fixed values $\alpha, \gamma \in (0,1)$. 
Of course, $V'(i)$ should be a distribution, since the definition above involves sampling. 

We will show that this refinement function can be defined and proved convergent in our logic. 
We first define the function ${\refine\, V}:{\TD(\Prop^N)}$ taking one step of the refinement 
 in context 
\[
 \trans :^\infty \Act \lol \D(N)^N, \reward :^\infty  \Act \lol \TD(\Prop)^N, \policy :^\infty \Act^N, V :^k \Prop^N
\]
where $k \defeq 1-\alpha + \gamma\alpha$. 
Note that all the parameters of the reinforcement learning setup are given sensitivity $\infty$, because these
are assumed to be closed terms that will be called repeatedly. 
We define $\refine\,V$ as $\strength (\refine'\,V)$ where
\[
  \strength : \TD(\Prop)^N \lol \TD(\Prop^N)
\]
is obtained in the standard way by induction on $N$, and
\begin{align*}
\refine'\,V \defeq \tuple{\letIn{r = \reward(\policy(i))(i) , \, 
  j = \trans(\policy(i))(i) }{
  \delta((1-\alpha)V(i) \ltensor \alpha(r \ltensor \gamma V(j)))
  }}_{i\leq N} 
\end{align*}
Since $k <1$, one can define the refinement function as the fixed point of $\refine$. 
In practice, however,
refinement is only iterated a finite number $n$ of times, defining $\TDfunc : \Prop^N \mexp \nat \mexp \TD(\Prop^N)$
by recursion on the second argument as 
\begin{align*}
 \TDfunc \, V \, 0 & \defeq V & 
 \TDfunc \, V\, (n+1) & \defeq \letIn{V' = (\TDfunc\, V \, n)}{\refine\, V'}
\end{align*}
Then, since $k(V \peq W) \ts \refine\, V \peq \refine \, W$, one can prove by induction on $n$ that 
\[
  k^n(V \peq W) \vdash \TDfunc\,V\, n \peq \TDfunc\, W\, n
\]

\section{Case study: hypercube walk}
\label{sec:hypercube}

This section shows how the internalisation of the Kantorovich distance (Theorem~\ref{thm:internal:Kantorovic})
can be used for coupling proofs in our logic. The example is a random walk on a hypercube adapted from
Aguirre et al.~\cite{BHKKM21}. Much of the example is done by reasoning in the model, as is most natural. 
Our logic is then used as an internal language of the model to apply Theorem~\ref{thm:internal:Kantorovic}
in the last step.

A position on an $N$-dimensional hypercube is an element of $\bool^N$, and we consider this a metric space with the 
normalised Hamming distance: The distance between $p$ and $q$ is $\frac 1N$ times the number of positions where $p$ 
and $q$ differ. In other words, the metric space of positions can be defined as 
\[
  \Pos \defeq \bigotimes_{i=1}^N\frac1N\bool
\]
where $\bool \defeq \unit + \unit$. 
Let $\unifZN : \TD(\nat)$ be the uniform distribution on $\{0,\dots, N\}$, 
and $\flip_i : \Pos \to \Pos$ be the operation that flips the $i$-th 
coordinate of a position if $i = 1,\dots, N$, and acts as the identity if $i=0$. Define the one-step hypercube random walk
as
\begin{align*}
 \hwalk & \defeq \lambda p.  \letIn{i = \unifZN}{\flip_i(p)}  : \Pos \lol_1 \TD(\Pos)
\end{align*}
We show that
\begin{equation} \label{eq:hyper:cube:goal}
 \frac{N-1}{N+1}(p \peq q) \ts \hwalk\, p \peq \hwalk\,q
\end{equation}
from which one can show that repeated iteration of $\hwalk$ converges.
To prove this, first construct, for each pair of positions
$p$ and $q$, a bijection $\bijarg pq$ of $\{0,\dots, N\}$ to itself by cases:
\begin{itemize}
\item If $p$ and $q$ are equal take $\bijarg pq$ to be the identity
\item If $p$ and $q$ differ in exactly one position $i$, let $\bijarg pq$ be the permutation
that swaps $i$ and $0$
\item If $p$ and $q$ differ in positions $i_1,\dots, i_n$ for $n>1$, let $\bijarg pq$ be the permutation
that cycles $i_1,\dots, i_n$. 
\end{itemize}
Below we just write $\bij$ for $\bijarg pq$.
One can then show that 
\begin{equation} \label{eq:sum:flip}
\frac{N-1}{N+1}(p\peq q) \ts \sum_{i=0}^N\frac1{N+1}(\flip_i\,p \peq \flip_{\bij(i)}\,q)
\end{equation}
holds in the model. This is done by analysing cases of $p$ and $q$. For example, if $p$ and $q$ differ in
exactly one position $j$, then $\flip_i\,p \peq \flip_{\bij(i)}\,q$ is $0$ for $i=j$ and $i=0$, and at all 
other values it equals $p\peq q$, from which the judgement follows.  

Let $\rho : \D(\nat\pair11\nat)$ be the uniform distribution on the finite set $\{(0,\bij (0)), \dots, (N,\bij (N))\}$, 
and define 
\[
\rho' \defeq
\letIn{z = \rho}{\big(\letIn{(i,j) = z}{(\flip_i \,p, \flip_{j}\,q)}\big)}
: \D(\Pos\!\pair11 \Pos)
\]
Then 
\[
\mean{(x,y)}{\rho'}{x\peq y} = \sum_{i=0}^N\frac1{N\!+\!1}(\flip_i\,p \peq \flip_{\bij(i)}\,q)
\]
An easy argument shows that the equalities $\D(\pi_1)(\rho') \peq \hwalk\,p$ and $\D(\pi_2)(\rho') \peq \hwalk\,q$
can be proved in the empty context, so that we have shown 
\begin{align*}
 \frac{N-1}{N+1}(p \peq q) \ts 
 \kant{\hwalk\,p}{\hwalk\,q}
\end{align*}
which, by Theorem~\ref{thm:internal:Kantorovic} is equivalent to (\ref{eq:hyper:cube:goal}). 

\begin{remark} \label{rem:limitations}
We remark that while the example above can be done for any $N$, the quantification over $N$ must be external, as internalising
it would require dependent types. Similarly, the uniform distribution $\unifZN : \TD(\nat)$ on $\{0,\dots, N\}$ used in the example,
can be constructed in our language for each $N$ as externally quantified only. It is not possible to write a function mapping $N$ to 
$\unifZN$ in our language, because the sensitivity annotations of variables must be given externally.
\end{remark}

\section{Related work}
\label{sec:related}

Our calculus is closely related to Fuzz~\cite{ReedP10}. 
One difference is that Fuzz has recursive types,
and we have guarded recursion. Fuzz is not a calculus of metric spaces, since these
do not model general recursion. Indeed, de Amorim et al.~\cite{AmorimGHKC17} use metric CPOs to model Fuzz.
Interestingly, de Amorim et al.~\cite{AmorimGHKC17} note that a fixed point operation on terms encoded using 
the recursive types of Fuzz can be given a typing rule similar to our rule for fixed points. 
Fuzz also has a probability distributions monad, but with a different metric designed
for reasoning about differential privacy. 
Unlike the Kantorovich metric, it is not clear whether such a metric allows for a principle of induction.
Fuzz also has an operational semantics. We believe operational semantics could be given for our calculus as
well, but leave this to future work. 
Our calculus is also related to graded lambda calculus \cite{BrunelGMZ14, GaboardiKOBU16}.

Dagnino and Pasquali~\cite{DagninoP22} were the first to notice that the sensitivity of predicates must
be taken into account when expressing elimination principles for equality in quantitative logic.
They present an affine propositional logic for quantitative reasoning about terms written in a first-order
language where the operations carry sensitivity annotations, like e.g., $\oplus_p$ does in this paper. 
One of the rules they present for equality is our rule (\textsc{eq-e}), but transitivity is a separate axiom, and they do not study applications 
like the ones studied here. 
Instead, Dagnino and Pasquali~\cite{DagninoP22} study general categorical notions of models; we just study one
model. 

The idea of using metric spaces and guarded recursion using scaling factors $r<1$
for programming and reasoning about probabilistic processes goes back at least to the late 
1990s~\cite{VinkR99,BaierK00}.  
To our knowledge, all previous work has used ultra-metric spaces, which means that one
can use simply typed lambda calculus and a simpler type for fixed points, as explained
in the introduction. The category of complete bisected ultrametric spaces forms a subcategory
of the topos of trees~\cite{ToT}, so this line of work is connected to the recent developments
on reasoning about processes using guarded recursion~\cite{BBBG018,KristensenMV22}. 
In these works, guarded recursion is formulated for a modal operator $\rhd$, which is
not related to probabilities. Equality, therefore, is not interpreted in terms of Kantorovich distance, as it is here.

The work on quantitative equational logic~\cite{MardarePP16,MardarePP18,MardarePP21} is also related. 
However, their equational approach fundamentally differs from ours
by using a Boolean-valued logical relation $\peq_\epsilon$ to reason about distances. Scaling of propositions and 
guarded recursion as used here would not work for 
a Boolean-valued logic. Indeed, the upper bound shown
in Section~\ref{sec:c:markov} can be proven in quantitative equational logic only by using the infinitary rule (Arch)~\cite{BacciBLM18}.
Mardare et al.~\cite{MardarePP16} show how a range of monads on metric spaces can be considered 
generated by quantitative algebraic theories. These include 
the $p$-Wasserstein metrics on distributions and the Hausdorff metric on compact subsets of a metric space. It is unclear how many of these can be captured as quantitative algebras for theories with operations equipped with sensitivity factors, in the same way we algebraically encode the Kantorovich metric in this paper. While this approach also works for the Hausdorff metric, it is unlikely that the $p$-Wasserstein metrics for $p\neq 1$ can be encoded the same way. 
\cite{MardarePP21} use quantitative equational logic to reason about terms of a first-order language with fixed points modelled using the Banach fixed point theorem. Their syntax uses `Banach patterns' as a form of sensitivity annotations on terms. Interestingly, the Banach pattern for fixed point terms is similar to our typing rule (\textsc{fix}).
We are unaware of any extensions of quantitative equational logic to higher-order. 
Dal Lago et al.~\cite{LagoHLP22} extends quantitative equational logic to weak $\lambda$-theories, and Dahlqvist and Neves~\cite{DahlqvistN22,DahlqvistN23} to linear and affine $\lambda$-calulus. Both extensions do not deal with recursion.

Quantitative program logics for imperative languages with probabilistic choice operators have a long history going back
to Kozen~\cite{Kozen85}. For example, Avanzini et al.~\cite{AvanziniBDG25} define a relational Hoare logic where relations on stores $s : \store$ 
are random variables $\store \times \store \to \extpreals$, so assertions are quantitative, but the interpretation of Hoare triples
is qualitative. Recent research has extended this idea to give quantitative interpretations also to Hoare triples themselves. 
For example, Approxis~\cite{GregersenAHTB24} uses error credits to define an approximate separation logic. 
In future work we will explore how our logic offers a different viewpoint on these results: 
By reading the qualitative definition of Hoare quadruple in our 
logic, one obtains a quantitative interpretation of Hoare logic. Approximate statement about Hoare quadruples 
should then be provable in our logic by showing that $\epsilon\cdot \false$ implies a Hoare quadruple similarly to how 
we reason about distances between Markov processes in Section~\ref{sec:markov}. 

\section{Conclusions} 
\label{sec:conclusions}

We defined an affine calculus for sensitivity and a higher-order logic for reasoning about it. The calculus includes
a form of guarded recursion, where the guards are sensitivities in the open interval $(0,1)$, and we saw how
this could be used for programming recursive processes. The logic likewise includes guarded recursion, 
which can be used, e.g., for proving upper bounds on distances between processes. We also saw how the 
principles of induction in the logic, in particular the one for induction over $\TD$ are powerful
principles. For example, we saw how they lead to proofs by coupling. 

One might ask to what extent the semantics of our logic generalises to other settings. Our goal has been
to reason about metric spaces, and $\CMet$ is essentially the largest possible category in which the entire
logic can be interpreted soundly. For example, we need to include discrete sets into the category to 
model natural numbers. This requires either a finite upper limit on distances as we do, 
or using extended metric spaces, \ie, spaces where the distance function maps into $[0, \infty]$,
and correspondingly use $[0, \infty]$ also as the set of truth values.
The latter choice, however, invalidates both the Banach fixed point theorem (fixed points might not exist or be unique) and the
guarded recursion principle. 

On the other hand, it should be possible to adapt our language and logic to extended metric spaces by introducing 
a new form of judgements for finiteness of truth values and for extended metric spaces being metric spaces (so all distances are finite). 
The logical guarded recursion
principle (\textsc{g-rec}) should then be restricted to only prove finite $\phi$, and similarly, fixed points should
only be applicable to metric spaces. This model would have the benefit that scalar multiplication
distributes over $\mprod$ and $\ltensor$, and also rules like (\textsc{assoc${}_1$}) would be invertible. 
We leave exploring such systems to future work. 

In future work, it would also be interesting to address the limitations of our language mentioned in Remark~\ref{rem:limitations}. 
It might be possible to allow types and sensitivity annotations to depend on sets such
as $\nat$, but allowing them to depend on metric spaces seems more difficult. One possible starting point could be 
DFuzz~\cite{GaboardiHHNP13}, a version of Fuzz with lightweight dependent types. 

\bibliography{biblio}


\newpage

\onecolumn
\appendix

\section{Omitted proofs of Section~\ref{sec:prelim}}

We start by reviewing some known facts about $\CMet$ as their are crucial for the technical development of the paper.

$\CMet$ is a symmetric monoidal closed category with monoidal product $X \mprod Y$ being the metric space with underlying set $X \times Y$ and 
distance function  $d_{X \mprod Y}((x,y)(x',y')) = d_X(x,x') \oplus d_Y(y,y')$.
The internal hom $X \mexp Y$ is the set of non-expansive maps from $X$ to $Y$ with point-wise supremum metric $d_{X \mexp Y}(f,g) = \sup_{x \in X} d_Y(f(x),g(x))$. For $A \in \CMet$, the adjunction $(- \mprod A) \dashv (A \mexp -)$ has counit $\mathsf{ev} \colon A \mprod (A \mexp X) \to X$ given by $\mathsf{ev}(a,f)=f(a)$ (the evaluation map).
Note that $\mprod$ is not the categorical product in $\CMet$, for which the distance function would have $\max$ in place of $\oplus$, as $\CMet$ is not Cartesian closed~\cite{Lawvere73}.

The symmetric monoidal closed structure of $\CMet$ originates from the commutative unital quantale $(\uinterval, \geq, \oplus, \lol, 0)$ over which distances take value. The tensor is truncated sum $x \oplus y = \min \{ x+y, 1\}$, the unit is $0$, and  adjoint is truncated reversed substraction $x \lol y = \max \{y - x, 0 \}$. Indeed,
for all $x,y,z \in \uinterval$
\begin{align*}
    x \oplus y \geq z 
    \iff x+y \geq z 
    \iff x \geq z - y 
    \iff x \geq y \lol z \,.
\end{align*}
Observe that $\uinterval$ is taken with reverse order: the bottom element is $1$, the top element is $0$, meet is $\sup$ and join is $\inf$.

The monoidal bifunctor $\mprod \colon \CMet \times \CMet \to \CMet$ acts as
\begin{align*}
\textbf{Objects:}\quad
(X,d_X) \mprod (Y, d_Y) = (X \times Y, d_{X \mprod Y})  &&
\textbf{Morphisms:}\quad
(f \mprod g)(x,y) = (f(x), g(y)) \,.
\end{align*}
where $d_{X \mprod Y}((x,y),(x',y')) = d_X(x,x') \oplus d_Y(y,y')$ is the truncated sum of the $1$-bounded metrics $d_X$ and $d_Y$.

\begin{proposition}
$\mprod \colon \CMet \times \CMet \to \CMet$ is well-defined.
\end{proposition}
\begin{proof}
Let $X, Y \in \CMet$. We shall prove that $d_{X \mprod Y}$ is a metric. 
\begin{align*}
 d_{X \mprod Y}((x,y)(x',y')) = 0 
 &\iff d_X(x,x') \oplus d_Y(y,y') = 0 \\
 &\iff d_X(x,x') = 0 \text{ and } d_Y(y,y') = 0 \\
 &\iff x = x' \text{ and } y = y' \\
 &\iff (x,y) = (x',y') \,.
\end{align*}
Symmetry follows by commutativity of $\oplus$. Observe that for $1$-bounded metrics triangular inequality is equivalent to $\forall a,b,c.\; d(a,c) \leq d(a,b) \oplus d(b,c)$. Then, by associativity and commutativity of $\oplus$ we have that
\begin{align*}
 d_{X \mprod Y}((x,y)(x'',y'')) 
 &= d_X(x,x'') \oplus d_Y(y,y'')  \\
 &\leq (d_X(x,x') \oplus d_X(x',x'')) \oplus (d_Y(y,y') \oplus d_Y(y',y'')) \\
 &= d_{X \mprod Y}((x,y)(x',y')) \oplus d_{X \mprod Y}((x',y')(x'',y'')) \,.
\end{align*}

On morphisms we shall prove that, if $f \colon X \to Z$ and $g \colon Y \to W$ are non-expansive maps, so is $f \mprod g$. Let $x,x' \in X$ and $y,y' \in Y$. Then, by monotonicity of $\oplus$, we have 
\begin{align*}
d_{X \mprod Y}((x,y),(x',y')) 
&= d_X(x,x') \oplus d_Y(y,y')  \\
&\geq d_Z(f(x),f(x')) \oplus d_W(g(y),g(y')) \\
&= d_{Z \mprod W}((f(x),g(y)),(f(x'),g(y')))  \\
&= d_{Z \mprod W}((f \mprod g)(x,y),(f \mprod g)(x',y')) \,. 
\end{align*}

Functoriality of $\mprod$ (\ie, $id_X \mprod id_Y = id_{X \mprod Y}$ and 
$(f \circ g) \mprod (h \circ k) = (f \mprod h) \circ (g \mprod k))$ follows by routine calculations.
\end{proof}

Associativity and symmetry of the monoidal structure on $\CMet$ is a direct consequence of the associativity and commutativity of truncated sum $\oplus$ in $\uinterval$.

The monoidal structure just described on $\CMet$ is closed, that is, for any $A \in \CMet$ the endofunctor 
$(- \mprod A)$ has right adjoint $(A \mexp -) \colon \CMet \to \CMet$ defined by 
\begin{align*}
\textbf{Objects:}\quad
A \mexp (X, d_X) = (\CMet(A,X), d_{A \mexp X})
&&
\textbf{Morphisms:}\quad
(A \mexp f)(g) = f \circ g \,.
\end{align*}
where $\CMet(A,X)$ is the homset of the arrows $A \to X$ and 
  $d_{A \mexp X}(g,g') = \sup_{a \in A} d_X(g(a), g'(a))$.
\begin{proposition}
$(A \mexp -) \colon \CMet \to \CMet$ is well-defined.
\end{proposition}
\begin{proof}
Let $X \in \CMet$. We shall prove that $d_{A \mexp X}$ is a metric. 
Reflexivity and symmetry are trivial. Triangular inequality follows
by (Scott-)continuity of sum
\begin{align*}
    d_{A \mexp X}(g,h) 
    = \sup_{a \in A} d_X(g(a), h(a)) 
    &\leq \sup_{a \in A} d_X(g(a), k(a)) + d_X(k(a), h(a)) \\
    & = \sup_{a \in A} d_X(g(a), k(a)) + \sup_{a \in A} d_X(k(a), h(a)) \,.
\end{align*}
$(A \mexp -)$ is well-defined on morphisms because non-expansive functions are closed under composition.  
Functoriality follows routinely by definition of $(A \mexp -)$ on morphisms.
\end{proof}
\begin{proposition}
For all $A \in \CMet$, $(- \mprod A)  \dashv  (A \mexp -)$.
\end{proposition}
\begin{proof}
For $A,B \in \CMet$, the counit $\mathsf{ev} \colon A \mprod (A \mexp B) \to B$ is defined as the \emph{evaluation map}: $\mathsf{ev}(a, g) = g(a)$, for 
$a \in A$ and $g \in A \mexp B$. Non-expansiveness follows by
\begin{align*}
    d_{A \mprod (A \mexp B)}((a,g),(a',g')) 
    &= d_A(a,a') \oplus \sup_{b \in A} d_B(g(b), g'(b)) \\
    &\geq \sup_{b \in A} d_B(g(b), g'(b)) \\
    &= d_B(\mathsf{ev}(a, g), \mathsf{ev}(a', g')) \,.
\end{align*}
Naturality of $ev^A_B$ on $B$ follows by standard arguments.

We shall now prove the universal property of the adjunction. 
Let $f \colon A \mprod X \to B$ a morphism in $\CMet$ (\ie, a non-expansive map) then the map $h \colon X \to (A \mexp B)$, defined as $h(x)(a) = f(a,x)$. It easy to show that $h$ is the only map such that $f = \mathsf{ev}\circ A\mprod h = f$.
We are left to show that $h$ is non-expansive:
\begin{equation*}
    d_X(x,x') 
    = \sup_{a \in A} d_{A \mprod X}((a,x),(a,x')) 
    \geq \sup_{a \in A} d_{B}(f(a,x),f(a,x')) 
    = d_{A \mexp B}(h(x),h(x')) \,.
    \qedhere
\end{equation*}
\end{proof}

\subsection*{Proof of Theorem~\ref{thm:graded:comonad}}
The maps $n_{r}$ and $w_A$ are clearly non-expansive 
as $\unit$ is a discrete space and $0 \scaling A = \unit$.
Non-expansiveness of $\kappa_{r,s,A}$ follows by monotonicity
of multiplication.

The map $m_{r,A,B}$ is non-expansive because 
\begin{align*}
 d_{r(A\mprod B)} ((a,b), (a',b')) 
 & = \min\{rd_A(a,a') + rd_B(b,b'), r, 1\} \\
 & \leq \min\{rd_A(a,a') + rd_B(b,b'), 1\} \\
 & = d_{rA \mprod rB}((a,b), (a',b'))
\end{align*}
If $r \geq 1$ the middle inequality is an equality.

For $c_{r,s,A}$, note that 
\begin{align*}
d_{(r+s)A}(a,a')
    &= \min\{(r+s)d_A(a,a'), 1 \} \\
    &= \min \{ rd_A(a,a')+ sd_A(a,a'), 1 \} \\
    &\geq \min \{ d_{rA}(a,a') + d_{sA}(a,a') ,1\} \\
    &= d_{rA\mprod sA}((a,a),(a',a')) 
\end{align*}

Finally, for $\asmap_{r,s,A}$, note that 
\begin{align*}
 d_{(rs)A}(a,a') 
 &=\min\{ rsd_A(a,a'), 1 \} \\
 & \geq \min \{ r \min \{ sd_A(a,a'),1 \},1 \} \\
 &= d_{r(sA)}(a,a')
\end{align*}
and if $s\leq 1$ or $r\geq 1$, then the inequality is an equality. 

\section{Omitted proofs of Section~\ref{sec:fixed:points}}

\subsection*{Proof of Proposition~\ref{prop:fp}}
 Let $f,f' \colon cY \to Y$. Since $f(\semfix(f)) = \semfix(f)$,
\begin{align*}
 d(\semfix (f), f(\semfix (f')) & \leq cd(\semfix(f), \semfix(f')) \\
 d(f(\semfix(f')), \semfix(f')) & \leq d(f,f') 
\end{align*}
so that $d(\semfix (f), \semfix (f')) \leq cd(\semfix(f), \semfix(f')) + d(f,f')$, and 
\[
 (1-c)d(\semfix (f), \semfix (f')) \leq d(f,f')
\]
as desired. For the second statement, by currying $f$ and composing with $\semfp$ we see that
\[
  \semfp(f) \colon (1-c)X \to (1-c) Y
\]
which implies $\semfp(f) \colon X \to Y$ since $\frac1{1-c}((1-c)X) = X$.  

\section{Omitted proofs of Section~\ref{sec:prob:meas}}

The \emph{Radon probability monad} $\TD$ on $\CMet$, 
takes objects $X \in \CMet$ to the space of Randon probability distributions over $X$ and morphisms $f\colon X \to Y$ to $\D f \colon \D X \to \D Y$, defined as $
\D f (\mu)(E) = \mu (f^{-1}(E))$, for $\mu\in\D X$ and $E$ a Borel set in $Y$ (a.k.a., the pushforward measure along $f$). The unit $\delta_X \colon X \to \D X$ and multiplication $m_X \colon \D\D X \to \D X$, are defined as follows, for $x \in X$, $\Phi \in \D \D X $, 
and Borel set $E \subseteq X$
\begin{align*}
\delta_X(x) = \delta_x \,,
&&
m_X(\Phi)(E) = \int \, \mu(E) \; {\Phi(\mathrm{d}\mu)} \,,
\end{align*}
where $\delta_x$ is the Dirac measure at $x$.

In~\cite[Theorem~3.8]{MardarePP18}, it is shown that the Randon probability monad can be characterized as the free IB algebra over complete metric spaces. This provides us with the following universal property for $\D X$: For any IB algebra $Y$ and non-expansive map $X \to Y$, there exists a unique non-expansive map $f^\sharp \colon \D X \to Y$ making the following diagrams commute, for all $p \in (0,1)$.
\[
 \begin{tikzcd}
  p\D X \mprod \bar{p} \D X \ar{r} \arrow[d, "pf^\sharp \mprod \bar{p} f^\sharp"'] 
  & \D X \arrow[d, "f^\sharp"'] & X \arrow[l,"\delta_X"'] \arrow[dl, "f"]\\
  pY \mprod \bar{p}Y \ar{r} & Y
 \end{tikzcd}
\]
where $\bar p = (1-p)$ and the maps without labels are the interpretations of $\oplus_p$ in the respective IB algebras. The map $f^\sharp$ is called \emph{homomorphic extension} of $f$.

The proof of Proposition~\ref{prop:D:free:IB} follows from the universal
property above.
\subsection{Proof of Proposition~\ref{prop:D:free:IB}}
We first show the case where $\Gamma$ is $\unit$ (so can be ignored), and prove the cases $r=0$, $0 < r \leq 1$, and $1 < r < \infty$ separately.
\begin{description}[wide, font=\sc]
\item[Case $r = 0$:] 
Follows by defining $\bar{f} = f$ as $0 X = \unit = 0 \D X$. Uniqueness of the map also follows trivially.

\item[Case $0 < r \leq 1$:] 
Observe that in this cases, the spaces $r\D X$ and $\D(rX)$ are equal.
Indeed, $X$ and $rX$ have same topology and Borel $\sigma$-algebra. Therefore $\D X$ and $\D (rX)$ have same underlying set of Radon measures. Moreover, by linearity of the integral:
\begin{equation*}
    d_{r\D X}(\mu,\nu) 
    = r \left(\min_{\omega} \int \, d_X \; {\omega}\right)
    = \min_{\omega} \int \, r d_X \; {\omega} 
    = d_{\D(rX)}(\mu,\nu) \,.
\end{equation*}
Let $f^\sharp \colon \D (rX) \to Y$ be the homomorphic extension of $f \colon rX \to Y$. By the equality $r\D X = \D(rX)$, a good candidate for $\bar{f} \colon r\D X \to Y$ is $f^\sharp$. So let $\bar{f} = f^\sharp$. Clearly, $\bar{f}$ is non-expansive. As $X$ and $rX$ have some Borel sets, $r\delta_X(x) = \delta_x = \delta_{rX}(x)$, therefore 
$f = \bar{f} \circ r\delta_X$. Furthermore, the commutativity of the diagram
\[
 \begin{tikzcd}[column sep=10ex]
  (pr)\D X \mprod (\bar pr) \D X \ar{r}{(\asmap\mprod\asmap)\circ m} \arrow[d, "\asmap\mprod\asmap"'] & r(p\D X \mprod \bar p\D X) \ar{r} & r\D X \ar{d}{\bar{f}}  \\
  p(r\D X) \mprod \bar p(r\D X) \ar{r}{p\bar{f} \mprod \bar p \bar{f}} & pY \mprod \bar p Y \ar{r} & Y
 \end{tikzcd}
\] 
can be easily established by direct calculation using the fact that $f^\sharp$ is an homomorphism of IB algebras and the maps from Theorem~\ref{thm:graded:comonad} are identities.

\item[Case $1 < r < \infty$:] 
Observe that, for $1 \leq r < \infty$, there is an adjunction $r \scaling (-) \dashv \frac1r\scaling(-)$ on $\CMet$ between scaling functor. Thus, define
$\tilde f\colon X \to \frac1r Y$ the adjoint of $f\colon rX \to Y$. Note that $f$ and $\tilde f$ have the same underlying set map. As $\frac1r \leq 1$, by Lemma~\ref{lem:IB:mexp:r}, $\frac1rY$ is a IB algebra, with algebra structure 
$p (\frac1rY) \mprod \bar{p} (\frac1rY) =  \frac1r(pY) \mprod \frac1r(\bar{p} Y) \xrightarrow{m} \frac1r (p Y \mprod \bar{p} Y) \xrightarrow{\frac1r \oplus_p} \frac1r Y$. Note that the algebra structure of $B$ and $\frac1rB$ have the same underlying set-maps.
Then, it make sense to define $\bar f \colon r\D X \to Y$ as the adjoint of  $(\tilde{f})^\sharp \colon \D X \to \frac1r Y$, which is obtained as the unique homomorphic extension of $\tilde{f}$. The commutativity of the diagram
\[
 \begin{tikzcd}[column sep=10ex]
  (pr)\D X \mprod (\bar pr) \D X \ar{r}{(\asmap\mprod\asmap)\circ m} \arrow[d, "\asmap\mprod\asmap"'] & r(p\D X \mprod \bar p\D X) \ar{r} & r\D X \ar{d}{\bar{f}}  \\
  p(r\D X) \mprod \bar p(r\D X) \ar{r}{p\bar{f} \mprod \bar p \bar{f}} & pY \mprod \bar p Y \ar{r} & Y
 \end{tikzcd}
\] 
can be established by direct calculation using the fact that $(\tilde{f})^\sharp$ is an homomorphism of IB algebras and $f$ has the same underlying map of $(\tilde{f})^\sharp$, since the maps from Theorem~\ref{thm:graded:comonad} are identities.
\end{description}
The more general case, where $\Gamma$ is not assumed to be $\unit$, follows from the above by using the IB-algebra structure on $\Gamma \mexp Y$ obtained as in Lemma~\ref{lem:IB:mexp:r}.
Indeed, by uncurrying $\tilde f \colon \mexp \D X \to (\Gamma \mexp Y)$, we obtain the desired map. 

\begin{proof}[Proof of Lemma~\ref{lem:IB:mexp:r}]
The operation on $X \mprod X'$ can be defined from those of $X$ and $X'$ as the following isomorphisms holds, since $p < 1$,
\[ p(X \mprod X') \mprod (1-p)(X \mprod X') \iso
 (pX \mprod (1-p)X) \mprod (pX' \mprod (1-p)X') \,. \]
 
The operation on $Y \mexp X$ can be defined as the 
adjoint correspondent to the operation 
\[Y\mprod p(Y \mexp X)\mprod (1-p)(Y \mexp X) \to X\] 
defined as a composition using the maps 
\begin{align*}
 c_{p,1-p,Y} & : Y \to pY \mprod (1-p)Y \\
 m_{p,Y, Y \mexp X} & : pY \mprod p(Y \mexp X) \to p(Y\mprod (Y \mexp X)) 
\end{align*}
function evaluation and the IB algebra operation on $X$. 

For the algebra structure on $rB$, note that since $r\leq 1$, 
\[ p(rB) \mprod (1-p)(rB) \iso r(pB) \mprod r(1-p)B
\]
The algebra structure can therefore be defined using the map $m$ from Theorem~\ref{thm:graded:comonad}. 
\end{proof}

\section{Omitted proofs of Section~\ref{sec:calculus}}

\subsection*{Proof of Lemma~\ref{lm:weakening} (1)}
We show that $\hastype{\Gamma,\Gamma'}{t}{A}$ implies $\hastype{\Gamma,\Delta,\Gamma'}{t}{A}$.
The proof is a simple induction on the typing derivation of $\hastype{\Gamma,\Gamma'}{t}{A}$. We show only a few interesting cases:

\begin{description}[wide, font=\sc]
\item[Case (var):] Assume $\hastype{\Gamma,\Gamma'}{x}{A}$ was derived with an application of (\textsc{var}). This can happen in two cases: either $\Gamma = \Gamma_1,x:^r A,\Gamma_2$ or  $\Gamma' = \Gamma'_1,x:^r A,\Gamma'_2$, for some $r\geq 1$. Then also the following are derivable via an application of the same typing rule:
\begin{mathpar}
\infrule[var]{
    r \geq 1}{
\hastype{\Gamma_1,x:^r A,\Gamma_2,\Delta,\Gamma'}{x}{A} }
\and
\infrule[var]{
    r \geq 1}{
\hastype{\Gamma, \Delta, \Gamma'_1,x:^r A,\Gamma'_2}{x}{A} }
\end{mathpar}

\item[Case (abs):] Assume $\hastype{\Gamma,\Gamma'}{\lambda x. t}{A \lol_r B}$ was derived with an application of (\textsc{abs}). 
Then, necessarily $\hastype{\Gamma,\Gamma', x:^rA}{t}{B}$ is also derivable. By inductive hypothesis so is $\hastype{\Gamma,\Delta,\Gamma', x:^rA}{t}{B}$ (note that we can assume $x \notin \Delta$, otherwise apply $\alpha$-renaming to $t$). 
Then the following is derivable:
\begin{mathpar}
\infrule[abs]{
    \hastype{\Gamma,\Delta,\Gamma', x:^rA}{t}{B} }{
\hastype{\Gamma,\Delta,\Gamma'}{\lambda x. t}{A \lol_r B} }
\end{mathpar}

\item[Case (app):] Assume $\hastype{\Gamma,\Gamma'}{tu}{A}$ was derived with an application of (\textsc{app}). Then, necessarily, $\hastype{\Gamma_1,\Gamma'_1}{t}{B \lol_r A}$ and $\hastype{\Gamma_2,\Gamma'_2}{u}{B}$ are also derivable, for $\Gamma = \Gamma_1+r\Gamma_2$ and $\Gamma' = \Gamma'_1+r\Gamma'_2$.
Observe that any $\Delta$ can always be described as the sum $\Delta = \Delta_1 +r\Delta_2$, for appropriate $\Delta_1, \Delta_2$.
By inductive hypothesis, then the following is derivable:
\begin{mathpar}
\infrule[app]{
    \hastype{\Gamma_1,\Delta_1,\Gamma'_1}{t}{B \lol_p A} &
    \hastype{\Gamma_2,\Delta_2,\Gamma'_2}{u}{B} }{
\hastype{(\Gamma_1,\Delta_1,\Gamma'_1)+r(\Gamma_2',\Delta_2,\Gamma'_2)}{t\,u}{B} }
\end{mathpar}
and $\Gamma,\Delta,\Gamma' = (\Gamma_1,\Delta_1,\Gamma'_1)+r(\Gamma_2',\Delta_2,\Gamma'_2)$.

\item[Case ($\mprod$):] Assume $\hastype{\Gamma,\Gamma'}{(t,u)}{A \pair{r}{s} B}$ was derived with an application of ($\mprod$). Then, necessarily, $\hastype{\Gamma_1,\Gamma_2}{t}{A}$ and $\hastype{\Gamma'_1, \Gamma'_2}{u}{B}$ are also derivable, for
$\Gamma = r\Gamma_1+s\Gamma_2+\Gamma_3$ and $\Gamma' = r\Gamma'_1+s\Gamma'_2+\Gamma'_3$. Observe that $\Delta$ can described as the sum $\Delta= r\Delta_1+s\Delta_2+\Delta_3$, for 
appropriate $\Delta_1,\Delta_2,\Delta_3$. By inductive hypothesis, the following is derivable:
\begin{mathpar}
\infrule[$\mprod$]{
    \hastype{\Gamma_1,\Delta_1,\Gamma_2}{t}{A} &
    \hastype{\Gamma'_1,\Delta_2,\Gamma'_2}{u}{B} }{
\hastype{r(\Gamma_1,\Delta_1,\Gamma'_1) + s(\Gamma_2,\Delta_2,\Gamma'_2)+
(\Gamma_3,\Delta_3,\Gamma'_3)}{(t,u)}{A \pair{r}{s} B} }
\end{mathpar}
and $\Gamma,\Delta,\Gamma' = r(\Gamma_1,\Delta_1,\Gamma'_1) + s(\Gamma_2,\Delta_2,\Gamma'_2)+
(\Gamma_3,\Delta_3,\Gamma'_3)$.

\item[Case ($\oplus_p$):] Assume $\hastype{\Gamma,\Gamma'}{t \oplus_p u}{\TD A}$ was derived with an application of ($\oplus_p$).
Then, necessarily, $\hastype{\Gamma_1,\Gamma'_1}{t}{\TD A}$ and 
$\hastype{\Gamma_2, \Gamma'_2}{u}{\TD A}$ are also derivable, for
$\Gamma = p\Gamma_1 + (1-p)\Gamma_2$ and $\Gamma' = \Gamma'_1, \Gamma'_2$.
By inductive hypothesis, then the following is derivable:
\begin{mathpar}
\infrule[$\oplus_p$]{
    \hastype{\Gamma_1,\Delta,\Gamma'_1}{t}{\TD A} &
    \hastype{\Gamma_2,\Delta,\Gamma'_2}{u}{\TD A}
    & p\in(0,1) }{
\hastype{p(\Gamma_1,\Delta,\Gamma'_1) + (1-p)(\Gamma_2,\Delta,\Gamma'_2)}{t \oplus_p u}{\TD A} }
\end{mathpar}
and $\Gamma,\Delta,\Gamma' = p(\Gamma_1,\Delta,\Gamma'_1) + (1-p)(\Gamma_2,\Delta,\Gamma'_2)$.

\item[Case (fix):] Assume $\hastype{\Gamma,\Gamma'}{\fix{x}{t}}{A}$ was derived with an application of (\textsc{fix}) and that $x \notin \Delta$ (otherwise apply $\alpha$-renaming to $t$).
Then, necessarily, $\hastype{(1-p)(\Gamma,\Gamma'), x:^pA}{t}{A}$
is also derivable for $p < 1$. Then, by inductive hypothesis,
the following is also derivable
\begin{mathpar}
\infrule[fix]{
    \hastype{(1-p)(\Gamma,\Delta,\Gamma'), x:^pA}{t}{A} & 
    p < 1 }{
\hastype{\Gamma,\Delta,\Gamma'}{\fix{x}{t}}{A} }
\end{mathpar}
\end{description}
The induction follows similarly also for the terms of
the extended calculus presented in Section~\ref{sec:logic}.

\subsection*{Proof of Lemma~\ref{lm:weakening} (2)}
We show that $\hastype{\Gamma}{t}{A}$ implies $\hastype{\Gamma+\Delta}{t}{A}$, for any $\Delta$ compatible with
$\Gamma$.
The proof is by induction on the derivation of $\hastype{\Gamma}{t}{A}$. We show only a few interesting cases:

\begin{description}[wide, font=\sc]

\item[Case (var):] Assume $\hastype{\Gamma}{x}{A}$ was derived with an application of (\textsc{var}). Then, $\Gamma = \Gamma_1, x:^rA, \Gamma_2$ for $r \geq 1$. As $\Delta$ is compatible with $\Gamma$, it can be split as $\Delta = \Delta_1, x:^s A, \Delta_2$, for appropriate $\Delta_1,\Delta_2$ and $s\in[0,\infty]$. By an application of (\textsc{var}) we derive 
\begin{mathpar}
\infrule[var]{
    r+s \geq 1}{
\hastype{(\Gamma_1+\Delta_1), x:^{r+s}A, (\Gamma_2+\Delta_2)}{x}{A} }
\end{mathpar}
and $\Gamma + \Delta = (\Gamma_1+\Delta_1), x:^{r+s}A, (\Gamma_2+\Delta_2)$.

\item[Case (abs):] Assume $\hastype{\Gamma}{\lambda x. t}{A \lol_r B}$ was derived with an application of (\textsc{abs}). Then, necessarily, $\hastype{\Gamma, x:^rA}{t}{B}$. By compatibility with $\Gamma$, we have $x\notin \Delta$. Then, the following is also derivable 
\begin{mathpar}
\infrule[abs]{
    \hastype{(\Gamma+\Delta), x:^rA}{t}{B} }{
\hastype{\Gamma+\Delta}{\lambda x. t}{A \lol_r B} }
\end{mathpar}
because the premise is equivalent to $\hastype{(\Gamma, x:^rA)+\Delta'}{t}{B}$ for $\Delta' = \Delta, x:^0 A$, which is derivable by inductive hypothesis.

\item[Case (app):] Assume $\hastype{\Gamma}{t\,u}{A}$ was derived with an application of (\textsc{app}). Then, necessarily, $\hastype{\Gamma_1}{t}{B \lol_r A}$ and $\hastype{\Gamma_2}{u}{B}$,
for $\Gamma = \Gamma_1 + r\Gamma_2$. Then, by inductive hypothesis, the following is also derivable:
\begin{mathpar}
\infrule[app]{
    \hastype{\Gamma_1+\Delta}{t}{B \lol_p A} &
    \hastype{\Gamma_2}{u}{B} }{
\hastype{(\Gamma_1+\Delta) + r\Gamma_2}{t\,u}{A} }
\end{mathpar}
and $\Gamma+\Delta = (\Gamma_1+\Delta) + r\Gamma_2$.

\item[Case ($\mprod$):] Assume $\hastype{\Gamma}{(t,u)}{A \pair rs B}$ was derived with an application of ($\mprod$). Then, necessarily, $\hastype{\Gamma_1}{t}{A}$, and $\hastype{\Gamma_3}{u}{B}$, for 
$\Gamma = r\Gamma_1 + s\Gamma_2+\Gamma_3$.
Thus, the following is also derivable
\begin{mathpar}
\infrule[$\mprod$]{
    \hastype{\Gamma_1}{t}{A} &
    \hastype{\Gamma_3}{u}{B} }{
\hastype{\Gamma+\Delta}{(t,u)}{A \pair{r}{s} B} }
\end{mathpar}
Note that, there was no need of using the inductive hypothesis in this cases, as weakening is already built in the rule. (Note that if the rule did not allow weakening in the conclusion, the theorem would fail in this case $s=r=0$.) 

\item[Case ($\oplus_p$):] Assume $\hastype{\Gamma}{t \oplus_p u}{\TD A}$ was derived with an application of ($\oplus_p$). Then, necessarily, $\hastype{\Gamma_1}{t}{\TD A}$ and $\hastype{\Gamma_2}{u}{\TD A}$ for $\Gamma = p\Gamma_1 + (1-p)\Gamma_2$. As $p\in(0,1)$, the context $\frac{1}{p}\Delta$ is well defined. Thus, by inductive hypothesis, the following is derivable:
\begin{mathpar}
\infrule[$\oplus_p$]{
    \hastype{\Gamma_1+\frac{1}{p}\Delta}{t}{\TD A} &
    \hastype{\Gamma_2}{u}{\TD A}
    & p\in(0,1) }{
\hastype{p(\Gamma_1+\frac{1}{p}\Delta) + (1-p)\Gamma_2}{t \oplus_p u}{\TD A} }
\end{mathpar}
and $\Gamma+\Delta = p(\Gamma_1+\frac{1}{p}\Delta) + (1-p)\Gamma_2$.

\item[Case (fix):] Assume $\hastype{\Gamma}{\fix{x}{t}}{A}$ was derived with an application of (\textsc{fix}). Then, we must have that $\hastype{(1-p)\scaling\Gamma, x:^pA}{t}{A}$ is derivable and 
$p<1$. By compatibility with $\Gamma$, also for $\Delta$ we have
that $x\notin \Delta$. Then, the following is derivable
\begin{mathpar}
\infrule[fix]{
    \hastype{(1-p)(\Gamma+\Delta), x:^pA}{t}{A} & 
    p < 1 }{
\hastype{\Gamma+\Delta}{\fix{x}{t}}{A} }
\end{mathpar}
as the premise is equivalent to $\hastype{((1-p)\Gamma, x:^pA)+\Delta'}{t}{A}$ for $\Delta' = (1-p)\Delta, x:^0 A$, which is derivable by inductive hypothesis. 
\end{description}
The induction follows similarly also for the terms of
the extended calculus presented in Section~\ref{sec:logic}.

\subsection*{Proof of Lemma~\ref{lm:substitution}}
We show that $\hastype{\Gamma, x:^rA, \Gamma'}{t}{B}$
and $\hastype{\Delta}{u}{A}$ implies the derivability of $\hastype{(\Gamma,\Gamma')+r\Delta}{t[u/x]}{B}$ whenever $(\Gamma,\Gamma')+r\Delta$ is well-defined. The proof is by induction on the derivation of $\hastype{\Gamma, x:^rA,\Gamma'}{t}{B}$. We show only a few interesting cases.

\begin{description}[wide, font=\sc]

\item[Case (var):] Assume the derivation of $\hastype{\Gamma, x:^rA,\Gamma'}{t}{B}$ ends with an application of (\textsc{var}). Then $t = y$ for some variable $y$. We distinguish two cases:
\begin{itemize}
    \item If $y=x$, then $t[u/x] = u$, $r\geq 1$, and $B = A$. Then, the thesis follows from $\hastype{\Delta}{u}{A}$ by Lemma~\ref{lm:weakening}(2) (weakening), as we obtain $\hastype{(\Gamma,\Gamma')+r\Delta}{u}{A}$.

    \item If $y\neq x$, then $t[u/x] = y$ and either $\Gamma = \Gamma_1, y:^s B, \Gamma_2$ or $\Gamma' = \Gamma'_1, y:^s B, \Gamma'_2$, for some $s\geq 1$. In both cases,
    apply (\textsc{var}) and then Lemma~\ref{lm:weakening}(2) (weakening) to obtain 
    $\hastype{(\Gamma,\Gamma')+r\Delta}{y}{B}$.
\end{itemize}

\item[Case (abs):] Assume the derivation of $\hastype{\Gamma, x:^r A,\Gamma'}{t}{B}$ ends with an application of (\textsc{abs}). Then the situation is as follows:
\begin{mathpar}
\infrule[abs]{
    \hastype{\Gamma, x:^rA, \Gamma', y:^sC}{v}{D} }{
\hastype{\Gamma,x:^rA,\Gamma'}{\underbrace{\lambda y. v}_{t}}{\underbrace{C \lol_s D}_{B} } }
\end{mathpar}
Note that $y \notin \Delta$, since $\Delta$ contains the same variables as $\Gamma,x:^rA,\Gamma'$. By Lemma~\ref{lm:weakening}(1) (weakening), 
from $\hastype{\Delta}{u}{A}$, we have $\hastype{\Delta,y:^0C}{u}{A}$. Then, the following is derivable:
\begin{mathpar}
\infrule[abs]{
    \hastype{((\Gamma,\Gamma')+p\Delta), y:^sC}{v[u/x]}{D} }{
\hastype{(\Gamma,\Gamma')+r\Delta}{\underbrace{\lambda y. v[u/x]}_{t[u/x]}}{\underbrace{C \lol_s D}_{B} } }
\end{mathpar}
as the premise follows by inductive hypothesis on 
$\hastype{\Gamma, x:^rA, \Gamma', y:^sC}{v}{D}$ with $\hastype{\Delta,y:^0C}{u}{A}$.

\item[Case (app):] Assume the derivation of $\hastype{\Gamma, x:^r A,\Gamma'}{t}{B}$ ends with an application of (\textsc{app}). Then the situation is as follows:
\begin{mathpar}
\infrule[app]{
    \hastype{\Gamma_1,x:^{r_1}A,\Gamma'_1}{v}{C \lol_s B} &
    \hastype{\Gamma_2,x:^{r_2}A,\Gamma'_2}{w}{C} }{
\hastype{\Gamma, x:^r A, \Gamma'}{\underbrace{v\,w}_t}{B} }
\end{mathpar}
for $\Gamma = \Gamma_1+s\Gamma_2$, $\Gamma' = \Gamma'_1+s\Gamma'_2$, and $r = r_1+sr_2$.
Then, by a direct application of the inductive hypothesis on
the premises of the above, the following is
also derivable:
\begin{mathpar}
\infrule[app]{
    \hastype{(\Gamma_1,\Gamma'_1)+r_1\Delta}{v[u/x]}{C \lol_s B} &
    \hastype{(\Gamma_2,\Gamma'_2)+r_2\Delta}{w[u/x]}{C} }{
\hastype{\underbrace{((\Gamma_1,\Gamma'_1)+r_1\Delta) + s((\Gamma_2,\Gamma'_2)+r_2\Delta)}_{(\Gamma,\Gamma')+r\Delta}}{\underbrace{v[u/x]\,w[u/x]}_{t[u/x]}}{B} }
\end{mathpar}

\item[Case (fix):] Assume the derivation of $\hastype{\Gamma, x:^r A,\Gamma'}{t}{B}$ ends with an application of (\textsc{fix}). Then the situation is as follows:
\begin{mathpar}
\infrule[fix]{
    \hastype{(1-p)(\Gamma, x:^rA, \Gamma'), y:^pB}{v}{B} & 
    p < 1 }{
\hastype{\Gamma, x:^rA, \Gamma'}{\underbrace{\fix{y}{v}}_t}{B} }
\end{mathpar}
By Lemma~\ref{lm:weakening}(1) (weakening), from $\hastype{\Delta}{u}{A}$, we have $\hastype{\Delta,y:^0B}{u}{A}$. Then, the following is derivable
\begin{mathpar}
\infrule[fix]{
    \hastype{(1-p)((\Gamma,\Gamma')+r\Delta), y:^pB}{v[u/x]}{B} & 
    p < 1 }{
\hastype{(\Gamma,\Gamma')+r\Delta}{\underbrace{\fix{y}{v[u/x]}}_{t[u/x]}}{B} }
\end{mathpar}
as the premise is derivable by inductive hypothesis on 
$\hastype{(1-p)(\Gamma, x:^rA, \Gamma'), y:^pB}{v}{B}$ with $\hastype{\Delta,y:^0C}{u}{A}$.
\end{description}
The induction follows similarly also for the terms of
the extended calculus presented in Section~\ref{sec:logic}.

\subsection{Omitted details about the denotation of terms}

We start by providing the definitions of the morphisms $\splt$, $\distribute$, and $\semweak$.
\begin{proposition}
For typing contexts $\Gamma,\Gamma',\Delta$ and $r \in \extpreals$ 
there exist morphisms 
\begin{align*}
\splt_{\Gamma,\Delta} &\colon \den{\Gamma+\Delta} 
    \to \den{\Gamma} \mprod \den{\Delta}
&
\proj_{\Gamma,\Delta,\Gamma'} &\colon \den{\Gamma,\Delta,\Gamma'} 
    \to \den{\Gamma,\Gamma'}
\\
\distribute_{r,\Gamma} &\colon \den{r\scaling\Gamma}
    \to r\scaling\den{\Gamma}
&
\semweak_{\Gamma,\Delta} &\colon \den{\Gamma+\Delta} 
    \to \den{\Gamma}
\end{align*}
\end{proposition}
\begin{proof}
We define the maps $\splt$, $\distribute$, and $\semweak$ by induction on the size of the typing contexts as follows
\begin{align*}
\splt_{\emptyctx,\emptyctx} = \lambda_\unit \,, &\quad
\splt_{(\Gamma, x:^rA)+(\Delta,x:^sA)} = \mathsf{match} \circ 
    (\splt_{\Gamma,\Delta} \mprod c) \,;
\\
\distribute_{r,\emptyctx} = id_{\unit} \,, &\quad
\distribute_{r,(\Gamma,x:^sA)} =  m \circ (\distribute_{r,\Gamma} \mprod \asmap) \,;
\\
\semweak_{\emptyctx,\emptyctx} = id_1 \,, &\quad
\semweak_{(\Gamma+\Delta, x:^{r+s}A),(\Delta,x:^sA)} =  
    \semweak_{\Gamma,\Delta} \mprod \kappa \,;
\end{align*}
where $c$, $\asmap$, $m$, and $\kappa$ are from
Theorem~\ref{thm:graded:comonad}, $\lambda_A \colon A \xrightarrow{\iso} \unit \mprod A$ is the left unitor, and $\mathsf{match}_{A,B,C,D} \colon (A \mprod B)\mprod (C \mprod D) \xrightarrow{\iso} (A \mprod C)\mprod (B \mprod D)$ is obtained from the associator and the swap morphisms coming from
the symmetric monoidal structure on $\CMet$ for $\mprod$.
As for the map $\proj$, this is defined by
\[
\proj_{\Gamma,\Delta,\Gamma'} = \den\Gamma\mprod ( \lambda^{-1} \circ ({!_1}\mprod\den{\Gamma'}))  \,;
\]

(In the above, the subscripts are omitted for the sake of readability, but can be easily
inferred from the types of the maps.)
\end{proof}

Observe that $\semweak$ factorizes through itself 
\begin{equation*}
\newcommand{\GD}{\den{\Gamma+\Delta}}
\newcommand{\GDT}{\den{(\Gamma+\Delta)+\Theta}}
\begin{tikzcd}[column sep=1cm, row sep=0.7cm]
    \GDT \arrow[r, "\semweak"] 
          \arrow[dr, "\semweak"']
    & \GD \arrow[d, "\semweak"] \\
    & \den{\Gamma}
\end{tikzcd}
\end{equation*}
and commutes with $\splt$ and $\distribute$ in the following sense
\begin{equation*}
\newcommand{\GGDD}{\den{(\Gamma+\Gamma')+(\Delta+\Delta')}}
\newcommand{\GD}{\den{\Gamma+\Delta}}
\newcommand{\GGxDD}{\den{\Gamma+\Gamma'}\mprod\den{\Delta+\Delta'}}
\newcommand{\GxD}{\den{\Gamma}\mprod\den{\Delta}}
\newcommand{\rGD}{\den{r(\Gamma+\Delta)}}
\newcommand{\rG}{\den{r\Gamma}}
\newcommand{\G}{\den{\Gamma}}
\begin{tikzcd}[column sep=1cm, row sep=0.7cm]
    \GGDD \arrow[r, "\splt"] 
              \arrow[d, "\semweak"']
    & \GGxDD \arrow[d, "\semweak\mprod\semweak"] \\
    \GD \arrow[r, "\splt"] & \GxD
\end{tikzcd}
\quad
\begin{tikzcd}[column sep=1cm, row sep=0.7cm]
    \rGD \arrow[r, "\distribute"] 
              \arrow[d, "\semweak"']
    & r\GD \arrow[d, "r\semweak"] \\
    \rG \arrow[r, "\distribute"] & r\G
\end{tikzcd}
\end{equation*}
This follows by an easy induction as per definition of $\semweak$, from the laws that characterizes scaling as a $\extpreals$-graded comonad (see remark before Theorem~\ref{thm:graded:comonad}).

\medskip
Next we provide the omitted details about the interpretation of
the terms of the calculus. Recall that only well-types terms have
an interpretation, which is given as a morphism in $\CMet$ of type
\begin{equation*}
    \den{\hastype{\Gamma}{t}{A}} \colon \den{\Gamma} \to \den{A}
\end{equation*}
The map above is defined by induction on the derivation of the typing judgment $\hastype{\Gamma}{t}{A}$, build according to the typing rules in Figure~\ref{fig:typingrules}. 
The exercise in giving the definition of the map above is ensuring it is non-expansive. The proof of this fact will be implicit as we define the maps by composing non-expansive maps. 
Once non-expansiveness is established, it will be often more convenient to work directly with their underlying set-maps, for which we give an explicit characterization in Figure~\ref{fig:setmap-semantics}.

\newcommand{\denSet}[1]{\llparenthesis\, #1 \,\rrparenthesis}
\begin{figure}[tb]
\centering
\begin{align*}
\denSet{\hastype{\Gamma,x:^rA,\Gamma'}{x}{A}}(\gamma,a,\gamma') 
    &= a \\
\denSet{\hastype{\Gamma}{\lambda x.t}{A \lol_r B}}\gamma 
    &= \curry(\denSet{\hastype{\Gamma,x:^rA}{t}{B}})\gamma \\
\denSet{\hastype{\Gamma+r\Gamma'}{t\,u}{B}}\gamma 
    &= 
    \begin{cases}
    \eval(\denSet{\hastype{\Gamma}{t}{A \lol_r B}}\gamma, 
        \denSet{\hastype{\Gamma'}{u}{A}}\gamma) &\text{if $r>0$}\\
    \eval(\denSet{\hastype{\Gamma}{t}{A \lol_0 B}}\gamma, *) &\text{otherwise}
    \end{cases}
    \\
\denSet{\hastype{\Gamma}{()}{\unitT}}\gamma &= * \\
\denSet{\hastype{\Gamma}{\tuple{t,u}}{A\times B}}\gamma 
    &= (\denSet{\hastype{\Gamma}{t}{A}}\gamma,\denSet{\hastype{\Gamma}{u}{B}}\gamma) \\
\denSet{\hastype{\Gamma}{\pi_i t}{A_i}}\gamma 
    &= \pi_i(\denSet{\hastype{\Gamma}{t}{A_1\times A_2}}\gamma) \\
\denSet{\hastype{\Gamma}{\inj_i t}{A_1+A_2}}\gamma 
    &= \inj_i (\denSet{\hastype{\Gamma}{t}{A_i}}\gamma) \\
\denSet{\hastype{\Gamma+r\Gamma'}{\case{t}{x}{u}{y}{v}}{C}} 
    &= 
    \begin{cases}
    \mathsf{k}(\gamma,
        \den{\hastype{\Gamma'}{t}{A+B}}\gamma) &\text{if $r>0$}
    \\
    \mathsf{k}(\gamma,*) &\text{otherwise}
    \end{cases}    
    \\
    &\text{where $\mathsf{k} = [\denSet{\hastype{\Gamma,x:^rA}{u}{C}}, \denSet{\hastype{\Gamma,y:^rB}{v}{C}}]$} \\
\denSet{\hastype{r\Gamma+s\Gamma'+\Gamma''}{(t,u)}{A\pair rs B}}\gamma 
    &= 
    \begin{cases}
    (\denSet{\hastype{\Gamma}{t}{A}}\gamma, \denSet{\hastype{\Gamma'}{u}{B}}\gamma) &\text{if $r,s > 0$}\\
    (*, \denSet{\hastype{\Gamma'}{u}{B}}\gamma) &\text{if $r=0$, $s>0$}\\
    (\denSet{\hastype{\Gamma}{t}{A}}\gamma,*) &\text{if $r>0$, $s=0$}\\
    (*,*) &\text{otherwise}
    \end{cases}
    \\
\denSet{\hastype{\Gamma+\Gamma'}{\letIn{(x,y)=u}{t}}{C}}\gamma 
    &= \denSet{\hastype{\Gamma,x:^rA,y:^sA}{t}{C}}(\gamma, 
        \denSet{\hastype{\Gamma'}{u}{A \pair rs B}}\gamma ) \\
\denSet{\hastype{\Gamma}{\delta(t)}{\TD A}}\gamma 
    &= \delta_{\den{A}}(\denSet{\hastype{\Gamma}{t}{A}}\gamma) \\
\denSet{\hastype{p\Gamma+(1-p)\Gamma'}{t\oplus_p u}{\TD A}}\gamma 
    &= \denSet{\hastype{\Gamma}{t}{\TD A}}\gamma 
    \oplus_p \denSet{\hastype{\Gamma'}{u}{\TD A}}\gamma \\
\denSet{\hastype{\Gamma+r\Gamma'}{\letIn{x=u}{t}}{E}}\gamma 
    &= 
    \begin{cases}
    \overline{\denSet{\hastype{\Gamma,x:^rA}{t}{E}}}(\gamma, 
        \denSet{\hastype{\Gamma'}{u}{\TD A}}\gamma) &\text{if $r>0$}\\
    \overline{\denSet{\hastype{\Gamma,x:^0A}{t}{E}}}(\gamma, *) &\text{otherwise}
    \end{cases}\\
\denSet{\hastype{\Gamma}{\Zero}{\nat}}\gamma &= 0 \\
\denSet{\hastype{\Gamma}{\Succ{t}}{\nat}}\gamma 
    &= \denSet{\hastype{\Gamma}{t}{\nat}}\gamma +1 \\
\denSet{\Rec{z}{(x,y).s}{n}}\gamma 
    &= 
    \iterate[](\denSet z\gamma,
               (\curry \circ \curry)(\denSet s\gamma),
               \denSet{n}\gamma ) \\
    &\text{where } \iterate[](a,f,0) = a \\
    &\phantom{\text{where }} \iterate[](a,f,n+1) = f(\iterate[](a,f,n))(n) \\
\denSet{\hastype{\Gamma}{\fix{x}{t}}{A}}\gamma &= \semfp(\denSet{\hastype{(1-p)\Gamma,x:^pA}{t}{A}})(\gamma)
\end{align*}

\caption{Set-theoretic denotation of terms. We denote by $\denSet{\hastype{\Gamma}{t}{A}}$ the underlying set map of the denotation
$\den{\hastype{\Gamma}{t}{A}}$ in $\CMet$. Note that the sensitivities of terms play a role only when they are $0$ as scaling by zero has the effect of collapsing a metric space to the singleton $\unit$, with $*$ being its unique element.}
\label{fig:setmap-semantics}
\end{figure}

Note that, since the definition of the map above depends on the derivation of the judgment $\hastype{\Gamma}{t}{A}$, different derivations could, in principle, yield different maps. However, as we will see, this ambiguity is avoided by assuming that terms are annotated with sufficient typing information. We provide more details on this later.

Now we turn to the actual definition. This is given by cases, for each typing rule. For readability, we use $\den{t}$ as a shorthand for the more precise $\den{\hastype{\Gamma}{t}{A}}$, and assume that the types of terms and their subterms match exactly those shown in the rules of Figure~\ref{fig:typingrules}. 
(We will often use this shorthand elsewhere as well, whenever the types are clear from context.)

\begin{description}[font=\sc,wide]
\item[Case (var):]
\begin{equation*}
    \den{x} \colon 
    \den{\Gamma, x:^rA, \Gamma'} 
        \xrightarrow{\pi_{r\den{A}}} r\scaling\den{A}
        \xrightarrow{\kappa} 1\scaling\den{A}
        \xrightarrow{\epsilon} \den{A} \,.
\end{equation*}
where $\pi_{r\den{A}}\colon \den{\Gamma} \mprod r\den{A} \mprod \den{\Gamma'} \to r\den{A}$ is the projection on the argument corresponding to the variable $x$, and $\kappa$, $\epsilon$ are the maps from Theorem~\ref{thm:graded:comonad}. Note that $r\geq 1$.

\item[Case (abs):] 
\begin{equation*}
    \den{\lambda x.u} \colon 
    \den{\Gamma} 
        \xrightarrow{\curry(\den{u})} 
    \den{B \lol_r C} \,.
\end{equation*}
where $\curry$ denotes currying.

\item[Case (app):] 
\begin{equation*}
    \den{t\,u} \colon
    \den{\Gamma + r\scaling \Gamma'} 
        \xrightarrow{(\den\Gamma\mprod \distribute) \circ \splt} 
    \den\Gamma\mprod r \den{\Gamma'} 
        \xrightarrow{\den t\mprod r\den u}
    \den{A \lol_r B} \mprod r\den A
        \xrightarrow{\mathsf{ev}} 
    \den{B} \,.
\end{equation*}
where $\eval$ denotes function application.

\item[Case (unit):]
\begin{equation*}
    \den{()} \colon \den{\Gamma} \xrightarrow{!} \unit
\end{equation*}
where $!$ is the final map to $\unit$.

\item[Case (pair):] 
\begin{equation*}
    \den{\tuple{t,u}} \colon \den{\Gamma}
        \xrightarrow{\tuple{\den{t},\den{u}}}
    \den{A \times B}
\end{equation*}

\item[Case ($\pi_i$):]
\begin{equation*}
    \den{\pi_it} \colon \den{\Gamma}
        \xrightarrow{\den{t}}
    \den{A_1} \times \den{A_2}
        \xrightarrow{\pi_i}
    \den{A_i}
\end{equation*}

\item[Case ($\inj_i$):]
\begin{equation*}
    \den{\inj_it} \colon \den{\Gamma}
        \xrightarrow{\den{t}}
    \den{A_i}
        \xrightarrow{\inj_i}
    \den{A_1+A_2}
\end{equation*}

\item[Case (case):]
\begin{multline*}
    \den{\case{t}{x}{u}{y}{v} } \colon 
    \den{\Gamma+r\Gamma'}
        \xrightarrow{(\den{\Gamma}\mprod\distribute)\circ \splt}
    \den{\Gamma} \mprod r\den{\Gamma'} \\
    \dots    \xrightarrow{\den{\Gamma}\mprod r\den t} 
    \den{\Gamma} \mprod r\den{A+B} 
        \xrightarrow{\iso}
    (\den{\Gamma}\mprod r\den{A}) + (\den{\Gamma}\mprod r\den{B})
        \xrightarrow{[\den{u}, \den{v}]}
    \den{C}
\end{multline*}
where we observe that the isomorphism above follows as both 
$(\den{\Gamma} \mprod -)$ and $(r \cdot -)$ (for $r\geq 1$) commutes with coproducts.

\item[Case ($\mprod$):] 
\begin{equation*}
    \den{(t,u)}\colon 
    \den{r\Gamma+s\Gamma'+\Gamma''}
        \xrightarrow{\semweak}
    \den{r\Gamma+s\Gamma'}
        \xrightarrow{(\distribute\mprod\distribute)\circ\splt}
    r \den{\Gamma} \mprod s\den{\Gamma'}
        \xrightarrow{r\den{t}\mprod s\den{u}}
    \den{A \pair{r}{s} B}
\end{equation*}

\item[Case (let-$\mprod$):]
\begin{equation*}
    \den{\letIn{(x,y)=u}{t}}\colon 
    \den{\Gamma+\Gamma'}
        \xrightarrow{\splt}
    \den{\Gamma} \mprod \den{\Gamma'}
        \xrightarrow{\den{\Gamma}\mprod\den{u}}
    \den{\Gamma} \mprod \den{A \pair{r}{s} B}
        \xrightarrow{\den{t}}
    \den{C}
\end{equation*}

\item[Case ($\delta$):]
\begin{equation*}
    \den{\delta(t)} \colon 
    \den{\Gamma}
        \xrightarrow{\den{t}}
    \den{A}
        \xrightarrow{\delta}
    \den{\TD A}
\end{equation*}
where $\delta$ is the unit of the monad $\D$.

\item[Case ($\oplus_p$):]
\begin{equation*}
    \den{t \oplus_p u} \colon 
    \den{p\Gamma+\overline{p}\Gamma'}
        \xrightarrow{(\distribute\mprod\distribute)\circ\splt}
    p\den{\Gamma}\mprod \overline{p}\den{\Gamma'}
        \xrightarrow{p\den{t} \mprod \overline{p}\den{u}}
    p\D\den{A}\mprod \overline{p}\D\den{A}
        \xrightarrow{\oplus_p}
    \den{\TD A}
\end{equation*}
where $\overline{p} = 1-p$ and $\oplus_p$ is the IB algebra operation
on $\D \den{A}$.

\item[Case (let):]
\begin{equation*}
    \den{\letIn{x=u}{t}} \colon \den{\Gamma+r\Gamma'}
        \xrightarrow{(\den{\Gamma}\mprod\distribute)\circ\splt}
    \den{\Gamma}\mprod r\den{\Gamma'}
        \xrightarrow{\den{\Gamma}\mprod r\den{u}}
    \den{\Gamma}\mprod r\D\den{A}
        \xrightarrow{\overline{\den{t}}}
    \den{E}
\end{equation*}
where from $\overline{\den{t}}$ is obtained from $\den{t}$ by Proposition~\ref{prop:D:free:IB} as $\den{E}$ is an IB algebra.

\item[Case (zero):]
\begin{equation*}
    \den{\Zero} \colon \den{\Gamma} 
        \xrightarrow{\semweak}
    \unit
        \xrightarrow{0}
    \den{\nat}
\end{equation*}

\item[Case (succ):]
\begin{equation*}
    \den{\Succ{t}} \colon \den{\Gamma}
        \xrightarrow{\den{t}}
    \nat
        \xrightarrow{-+1}
    \den{\nat}
\end{equation*}   

\item[Case (rec):]
To define the interpretation, we will first define the auxiliary map
$\iterate[]$ and argue why it is non-expansive.
Let $F = \den{A}\mexp\nat\mexp\den{A}$.
Define the morphisms $\iterate \colon \den{A}\mprod \infty\scaling F \to \den{A}$ by induction on $m\in \nat$ as follows, 
\begin{align*}
    \iterate[0] &\colon \den{A}\mprod \infty\scaling F \xrightarrow{\pi_1} \den{A} \\
    \iterate[m+1] &\colon \den{A}\mprod \infty\scaling F 
        \xrightarrow{(\den{A}\mprod\cp)} 
    \den{A} \mprod \infty\scaling F \mprod \infty\scaling F \\
    &\phantom{\colon{}} \dots
        \xrightarrow{\iterate\mprod\kappa}
    \den{A} \mprod F 
        \xrightarrow{\eval}
    (\nat \mexp \den{A})
        \xrightarrow{m\mprod (\nat \mexp \den{A})}
    \nat \mprod (\nat \mexp \den{A})
        \xrightarrow{\eval}
    \den{A}
\end{align*}
where $m\colon \unit \to \nat$ is the constant $m$-map and 
$\cp \colon \infty\scaling F \to \infty\scaling F \mprod \infty\scaling F$ is
the diagonal function, which is non-expansive as $\infty\scaling F$ is discrete.
Observe that for discrete spaces $X$, the isomorphism $Y \mprod X \iso \coprod_{x\in X} Y$ holds. Thus, as $\nat$ is discrete, we can combine the morphisms above, into a single one defined as follows
\begin{equation*}
    \iterate[] \colon 
    \den{A} \mprod \infty\scaling F \mprod \nat 
        \xrightarrow{\iso}
    \coprod_{m \in \nat} (\den{A} \mprod \infty\scaling F)
        \xrightarrow{\coprod_{n}\iterate[m]} \den{A} 
\end{equation*}

Now we are ready to define the interpretation of $\Rec{z}{(x,y).s}{n}$
\begin{multline*}
    \den{\Rec{z}{(x,y).s}{n}} \colon 
    \den{\Gamma+\infty\Gamma'+\Gamma''}
        \xrightarrow{
        (\den{z} \mprod \infty s' \mprod \den{n}) \circ 
        (\den\Gamma \mprod \distribute \mprod \den{\Gamma''}) 
        \circ \splt
        }
    \den{A} \mprod \infty\scaling F \mprod \nat
        \xrightarrow{\iterate[]}
    \den{A}
\end{multline*}
where $s' = (\curry \circ \curry)(\den{s})$ is the doubly-curried version of $\den{s}$.
    
\item[Case (fix)] 
\begin{equation*}
    \den{\fix{x}{u}} \colon \Gamma
        \xrightarrow{\semfp\big((\distribute \mprod p\den{A})\circ \den{u}\big)}
    \den{A}
\end{equation*}
where the map $\semfp$ is from Proposition~\ref{prop:fp}. Note that
as $\semfp$ is applied properly as $p<1$.
\end{description}

\subsection*{Proof of Theorem~\ref{thm:SemIndependence} (Independence of the choice of derivation)}
The non-expansiveness of the maps is clear from the above definitions. 
However, as discussed in Remark~\ref{rem:independence}, to formally prove that the semantics is independent of the choice of derivation, 
we must be able to infer the types of subterms in a judgment $\hastype{\Gamma}{t}{A}$ solely from the information in $\Gamma$, $t$, $A$. 
This is not possible with the informal syntax presented in Section~\ref{sec:calculus}. For instance, as illustrated by the two derivations below, 
the subterm $u$ in $\hastype{\Gamma}{\pi_1\tuple{t,u}}{A}$ can have any type:
\begin{mathpar}
\infrule[$\pi_1$]{
    \infrule[pair]{
        \hastype{\Gamma}{t}{A} &
        \hastype{\Gamma}{u}{B} }{
    \hastype{\Gamma}{\tuple{t,u}}{A \times B}}
}{\hastype{\Gamma}{\pi_1\tuple{t,u}}{A} }
\and 
\infrule[$\pi_1$]{
    \infrule[pair]{
        \hastype{\Gamma}{t}{A} &
        \hastype{\Gamma}{u}{C} }{
    \hastype{\Gamma}{\tuple{t,u}}{A \times C}}
}{\hastype{\Gamma}{\pi_1\tuple{t,u}}{A}}
\end{mathpar}  
This is a well-known issue in semantics, and the standard solution is to annotate terms with additional information to prevent ambiguities like the one above. The \emph{official} syntax of the calculus, which avoids such problems, is presented below (the typing rules in Figure~\ref{fig:typingrules} are implicitly updated accordingly) 
---the syntax for logical predicates presented in Section~\ref{sec:logic} does not require modifications.
\newcommand{\app}[4]{\mathsf{app}_{#1,#2}(#3,#4)}
\begin{align*}
 t, u, v &{} ::= {}
  x \mid \lambda x. t \mid \app rAtu \mid () \\
 &\mid \tuple{t,u} \mid \pi^{A,B}_1 t \mid \pi^{A,B}_2 t \\
 &\mid \inj_1 t \mid \inj_2 t \mid \case{t}{(x:A)}{u}{(y:B)}{v}
 \\
 &\mid  (t,u) \mid \letIn{(x{:^r}A,y{:^s}B) = u }{t} 
 \mid \delta t \mid  t \oplus_p u \mid \letIn{x:A = u}{t} \\
 &\mid \Zero \mid \Succ{t} \mid \Rec{u}{(x,y).t}{v} 
 \mid \fix{x}{t}
\end{align*}

Note, however, that although the syntax has been enriched with type information, this is not sufficient to ensure unique derivations. The issue arises because a context 
$\Gamma$ can be split as a sum $\Gamma_1+\Gamma_2$ in multiple ways. The two derivations below illustrate this ambiguity clearly:
\begin{mathpar}
\small
\infrule[app]{
    \infrule[abs]{
        \infrule[var]{ }{\hastype{y:^3A,x:^1A}{x}{A}}
    }{\hastype{y:^3A}{\lambda x.x }{A \times B}}
    &
    \infrule[var]{ }{\hastype{y:^1A}{x}{A}}
}{\hastype{y:^4A}{ \app{1}{A}{(\lambda x.x)}{y} }{A} }
\and 
\infrule[app]{
    \infrule[abs]{
        \infrule[var]{ }{\hastype{y:^1A,x:^1A}{x}{A}}
    }{\hastype{y:^1A}{\lambda x.x }{A \times B}}
    &
    \infrule[var]{ }{\hastype{y:^3A}{x}{A}}
}{\hastype{y:^4A}{ \app{1}{A}{(\lambda x.x)}{y} }{A} }
\end{mathpar} 
Next we show that this ambiguity does not affect
the denotation. 

Formally, we show that, whenever a judgement $\hastype{\Gamma}{t}{A}$ has two different derivations, say $\nabla$ and $\nabla'$,
then the two denotations obtained from them are the same. We visualize this as
\begin{equation*}
    \denD{\frac{\nabla}{\hastype{\Gamma}{t}{A}}} = 
    \denD{\frac{\nabla'}{\hastype{\Gamma}{t}{A}}}
\end{equation*}
to show explicitly the use of different derivations in computing the denotations.

The proof is by induction on the term $t$. We show only 
the cases that are representative of the techniques that can be used to solve in a similar fashion all the others.
For the proof we will need to use Lemma~\ref{lem:sem:weak}(2). That lemma therefore needs to be proved first, 
as a statement on the interpretation of derivations, before we can prove the present theorem. This is not a problem,
because the proof of Lemma~\ref{lem:sem:weak} proves that the interpretation of the weakening of a derivation 
constructed in Lemma~\ref{lm:weakening} is the semantic weakening of the interpretation of the original derivation. 
\begin{description}[font=\sc,wide]
\item[Case ($x$):] This case is clear as the derivation is unique.

\item[Case ($\lambda x.t$):] Both derivations end with an application of (\textsc{abs}) and the situation is as follows:
\begin{mathpar}
\infrule[abs]{
    \infrule{\nabla}{
    \hastype{\Gamma, x:^rA}{t}{B}} }{
\hastype{\Gamma}{\lambda x. t}{A \lol_r B} }
\and
\infrule[abs]{
    \infrule{\nabla'}{
    \hastype{\Gamma, x:^rA}{t}{B}} }{
\hastype{\Gamma}{\lambda x. t}{A \lol_r B} }
\end{mathpar}
In this case, from the information in the conclusion the choice of the premise is unique. Thus, independence follows by inductive hypothesis on $u$.

\item[Case ($\app{r}{A}{t}{u}$):] 
Both derivations end with an application of (\textsc{app}) and the situation is as follows:
\begin{mathpar}
\infrule[app]{
    \infrule{\nabla_1}{
    \hastype{\Gamma_1}{t}{A \lol_r B}} 
    &
    \infrule{\nabla_2}{
    \hastype{\Gamma_2}{u}{A}} }{
\hastype{\Gamma_1 + r\Gamma_2}{\app{r}{A}{t}{u}}{B} }
\and
\infrule[app]{
    \infrule{\nabla'_1}{
    \hastype{\Gamma'_1}{t}{A \lol_r B}} 
    &
    \infrule{\nabla'_2}{
    \hastype{\Gamma'_2}{u}{A}} }{
\hastype{\Gamma'_1 + r\Gamma'_2}{\app{r}{A}{t}{u}}{B} }
\end{mathpar}
where $\Gamma= \Gamma_1 + r\Gamma_2 = \Gamma'_1 + r\Gamma'_2$. 
Let $\Gamma^\infty$ be the version of $\Gamma$ where all sensitivities are set to $\infty$. 
Then from the sub-derivations of $\hastype{\Gamma_1}{t}{A \lol_r B}$ and $\hastype{\Gamma'_1}{t}{A \lol_r B}$ occurring above, by Lemma~\ref{lem:sem:weak}, we obtain new derivations such that the following holds:
\begin{mathpar}
\denD{\frac{\nabla_1}{\hastype{\Gamma_1}{t}{A \lol_r B}}} 
\circ \semweak_{\Gamma_1,\Gamma^\infty} = 
\denD{\frac{\tilde\nabla_1}{\hastype{\Gamma^\infty}{t}{A \lol_r B}}}
\and
\denD{\frac{\nabla'_1}{\hastype{\Gamma'_1}{t}{A \lol_r B}}} 
\circ \semweak_{\Gamma'_1,\Gamma^\infty} = 
\denD{\frac{\tilde\nabla'_1}{\hastype{\Gamma^\infty}{t}{A \lol_r B}}}
\end{mathpar}
As the derivations on the right-hand side of the equation above have the same conclusion, by inductive hypothesis on $t$, we get 
$\den{\tilde\nabla_1} = \den{\tilde\nabla'_1}$.
Similarly we get 
\begin{mathpar}
\denD{\frac{\nabla_2}{\hastype{\Gamma_2}{u}{A}}} 
\circ \semweak_{\Gamma_2,\Gamma^\infty} = 
\denD{\frac{\tilde\nabla_2}{\hastype{\Gamma^\infty}{u}{A}}}
\and
\denD{\frac{\nabla'_2}{\hastype{\Gamma'_2}{u}{A}}} 
\circ \semweak_{\Gamma'_2,\Gamma^\infty} = 
\denD{\frac{\tilde\nabla'_2}{\hastype{\Gamma^\infty}{u}{A}}}
\end{mathpar}
such that $\den{\tilde\nabla_2} = \den{\tilde\nabla'_2}$. 
We must show that 
\begin{equation} \label{eq:lem:indep:deriv:app:goal}
\mathsf{ev} \circ (\den {\nabla_1}\mprod r\den{\nabla_2}) \circ (\den{\Gamma_1}\mprod \distribute) \circ \splt
= \mathsf{ev} \circ (\den {\nabla_1'}\mprod r\den{\nabla_2'}) \circ (\den{\Gamma_1'}\mprod \distribute) \circ \splt
\end{equation}
Since $\semweak$ commutes with $\splt$ and $\distribute$, it is easy to see that 
\begin{align*}
 \mathsf{ev} \circ (\den {\nabla_1}\mprod r\den{\nabla_2}) \circ (\den{\Gamma_1}\mprod \distribute) \circ \splt \circ \semweak_{\Gamma,\Gamma^\infty}
 & = \mathsf{ev} \circ (\den {\tilde\nabla_1}\mprod r\den{\tilde\nabla_2}) \circ (\den{\Gamma^\infty}\mprod \distribute) \circ \splt 
\end{align*} 
where the $\splt$ on the right hand side has type $\den{\Gamma^\infty} \to \den{\Gamma^\infty} \mprod \den{r \Gamma^\infty}$. 
Similarly
\begin{align*}
\mathsf{ev} \circ (\den {\nabla_1'}\mprod r\den{\nabla_2'}) \circ (\den{\Gamma_1'}\mprod \distribute) \circ \splt \circ \semweak_{\Gamma,\Gamma^\infty}
 & = \mathsf{ev} \circ (\den {\tilde\nabla_1'}\mprod r\den{\tilde\nabla_2'}) \circ (\den{\Gamma^\infty}\mprod \distribute) \circ \splt 
\end{align*}
So, since the underlying map of $\semweak_{\Gamma,\Gamma^\infty}$ is surjective we conclude (\ref{eq:lem:indep:deriv:app:goal}).
\end{description}
The induction follows similarly also for the terms of the extended calculus
presented in Section~\ref{sec:logic}.

\subsection*{Proof of Theorem~\ref{thm:Soundness:Term:Interpr} (Soundness)}
As the underlying set-maps of the interpretation of terms is for most parts the usual set-theoretic one, we will avoid showing soundness of the standard judgmental equalities, such as $\beta$ and $\eta$ rules function types, and those for product and sum types, and the one for the unit.
\begin{description}[font=\sc,wide]
\item[Case (let-$\mprod$):] We prove soundness of the judgmental equalities corresponding to the $\beta$ and $\eta$ rules for tensor type. 
We will assume $r,s>0$, and avoid discussing the other cases as they can be
obtained via simple adaptations of the one shown below:

The first one is: $\letIn{(x,y) = (u,v)}t \jeq t[u/x,v/y]$.
\begin{align*}
    \denSet{\hastype{\Delta&+r\Gamma+s\Gamma}{t[u/x,v/y]}{C}}\gamma = {} \\
    &= \denSet{\hastype{\Delta, x:^rA, y:^sB}{t}{C}}
        (\gamma,\den{\hastype{\Gamma}{u}{A}}\gamma,\den{\hastype{\Gamma'}{v}{B}}\gamma)
        \tag{Lemma~\ref{lem:sem:subst}} \\
    &= \den{\hastype{\Delta, x:^rA, y:^sB}{t}{C}}
        (\gamma,\den{\hastype{r\Gamma+s\Gamma'}{(u,v)}{A \pair{r}{s} B}}
            \gamma) \\
    &= \denSet{\hastype{\Delta+r\Gamma+s\Gamma}{\letIn{(x,y)=(u,v)}{t}}{C}}\gamma
\end{align*}
    
The second one is: $u[t/z] \jeq \letIn{(x,y) = t}{u[(x,y)/z]}$.
\begin{align*}
    &\denSet{\hastype{\Delta+\Gamma}{{}\letIn{(x,y)=t}{u[(x,y)/z]}}{C}}\gamma \\
    &= \denSet{\hastype{\Delta,x:^rA,y:^sB}{u[(x,y)/z]}{C}}(\gamma, \denSet{\hastype{\Gamma}{t}{A \pair{r}{s} B}}\gamma) \\
    &= \denSet{\hastype{\Delta,z:^1 A \pair{r}{s} B}{u}{C}}
        (\gamma, \denSet{\hastype{x:^rA, y:^sB}{(x,y)}{A \pair{r}{s} B}}(\denSet{\hastype{\Gamma}{t}{A \pair{r}{s} B}}\gamma)) 
    \tag{Lemma~\ref{lem:sem:subst}} \\
    &= \denSet{\hastype{\Delta,z:^1 A \pair{r}{s} B}{u}{C}}(\gamma,\denSet{\hastype{\Gamma}{t}{A \pair{r}{s} B}}\gamma) 
    \tag{*} \label{eq:Jeq:Soundness:Let:Tensor} \\
    &= \denSet{\hastype{\Delta+\Gamma}{u[t/z]}{C}}\gamma 
    \tag{Lemma~\ref{lem:sem:subst}}
\end{align*}
where \eqref{eq:Jeq:Soundness:Let:Tensor} is justified by the following equality, where $(a,b) = \denSet{\hastype{\Gamma}{t}{A \pair{r}{s} B}}(\gamma)$
\begin{align*}
    \denSet{\hastype{x:^rA, y:^sB}{(x,y)}{A \pair{r}{s} B}}(a,b)
    = (\denSet{\hastype{x:^1A}{x}{A}}a, \denSet{\hastype{y:^1B}{y}{B}}b)
    = (a,b)
\end{align*}

\item[Case (let):] We prove soundness of the judgmental equalities 
regarding the let-binding for $\TD A$.

Unit law: $\letIn{x = \delta(t)}{u} \jeq u[t/x]$:
\begin{align*}
    \denSet{\hastype{\Gamma+r\Gamma'}{&\letIn{x=\delta(t)}{u}}{B}}\gamma = {} \\
    &=\overline{\denSet{\hastype{\Gamma,x:^rA}{u}{B}}} 
        (\gamma,\denSet{\hastype{\Gamma'}{\delta(t)}{\TD A}}\gamma) \\
    &=\overline{\denSet{\hastype{\Gamma,x:^rA}{u}{B}}}
        (\gamma,\delta_{\den{A}}(\denSet{\hastype{\Gamma'}{t}{A}}\gamma)) \\
    &=\denSet{\hastype{\Gamma,x:^rA}{u}{B}}
        (\gamma,\denSet{\hastype{\Gamma'}{t}{A}}\gamma)
    \tag{Proposition~\ref{prop:D:free:IB}} \\
    &= \denSet{\hastype{\Gamma+r\Gamma'}{u[t/x]}{B}}\gamma
    \tag{Lemma~\ref{lem:sem:subst}}
\end{align*}
    
Associativity law: $(\letIn{x = (\letIn{y=v}{t})}{u}) \jeq (\letIn{y=v}{(\letIn{x=t}{u})})$:
\begin{align*}
    \denSet{\hastype{\Gamma+r(\Delta+s\Delta')}{{}&\letIn{x = (\letIn{y=v}{t})}{u}}{A}}\gamma \\
    &=\overline{\denSet{\hastype{\Gamma,x:^rA}{u}{C}}}
    (\gamma,
    \denSet{\hastype{\Delta+s\Delta'}{\letIn{y=v}{t}}{\TD A}}\gamma) \\
    &=\overline{\denSet{\hastype{\Gamma,x:^rA}{u}{C}}}
    (\gamma,
    \overline{\denSet{\hastype{\Delta,y:^sB}{t}{\TD A}}}
    (\gamma, \denSet{\hastype{\Delta'}{v}{\TD B}}\gamma)
    ) \\
    &=\overline{
        \overline{\denSet{\hastype{\Gamma,x:^rA}{u}{C}}}
            (\gamma,\denSet{\hastype{\Delta,y:^sB}{t}{\TD A}}}
                (\gamma,\denSet{\hastype{\Delta'}{v}{\TD B}}\gamma)
            ) 
    \tag{*} \label{eq:Jeq:Soundness:Let:TD:Assoc}\\
    &=\overline{\denSet{\hastype{(\Gamma+r\Delta),y:^{rs}B}{\letIn{x=t}{u}}{C}}}
    (\gamma,\denSet{\hastype{\Delta'}{v}{\TD B}}\gamma) \\
    &=\denSet{\hastype{\Gamma+r(\Delta+s\Delta')}{\letIn{y=v}{(\letIn{x=t}{u})}}{C}}\gamma
\end{align*}
where \eqref{eq:Jeq:Soundness:Let:TD:Assoc} is justified by the fact that 
$\overline{\den{\hastype{\Gamma,x:^rA}{u}{C}}} 
\circ (\den{\Gamma}\mprod r\overline{\den{\hastype{\Delta,y:^sB}{t}{\TD A}}}\circ m)$
 is the unique homomorphic extension of 
$\overline{\den{\hastype{\Gamma,x:^rA}{u}{C}}}
\circ(\den{\Gamma}\mprod r\den{\hastype{\Delta,y:^sB}{t}{\TD A}} \circ m)$
 in the sense of Proposition~\ref{prop:D:free:IB}
as the following commuting diagram shows and the fact that the
natural transformation $m$ from Theorem~\ref{thm:graded:comonad} has the identity as underlying set-map:

\begin{equation*}
\newcommand{\B}{\den{B}}
\newcommand{\Gx}{\den{\Gamma}\mprod}
\newcommand{\GpDxpq}{\Gx r\den{\Delta} \mprod rs}
\newcommand{\GxpDxq}[1]{\Gx r(\den{\Delta} \mprod s #1)}
\begin{tikzcd}[column sep=4cm, row sep=0.7cm]
    \GpDxpq\B \arrow[r, "\GpDxpq\delta_{\B}"] 
              \arrow[d, "\Gx m"']
              \arrow[dr, phantom, "\text{(naturality)}"]
    & \GpDxpq\TD\B \arrow[d, "\Gx m"] \\
    \GxpDxq\B \arrow[r, "\GxpDxq{\delta_{\B}}"]
              \arrow[rd, bend right=10, "\Gx r\den{t}"']
              \arrow[dr, phantom, "\text{(Prop.~\ref{prop:D:free:IB})}"]
    & \GxpDxq{\TD\B} \arrow[d, "\Gx r\overline{\den{t}}"] \\
    & \Gx r\TD\den{A} \arrow[d, "\overline{\den{u}}"] \\
    & \den{C}
\end{tikzcd}
\end{equation*}

Homomorphism: 
$\letIn{x = v\oplus_p t}u \jeq (\letIn{x = v}u)\oplus_p(\letIn{x=t}u)$:
\begin{align*}
    \denSet{\hastype{\Gamma+r&(p\Delta+(1-p)\Delta')}{\letIn{x = v\oplus_p t}u}{\TD A}}\gamma \\
    &=\overline{\mathsf{u}}
        (\gamma,\denSet{\hastype{(p\Delta+(1-p)\Delta')}{v \oplus_p t}{\TD B}}\gamma) \\
    &=\overline{\mathsf{u}}(\gamma,
        \denSet{\hastype{\Delta}{v}{\TD B}}(\gamma)
        \oplus_p
        \denSet{\hastype{\Delta'}{t}{\TD B}}(\gamma)) \\
    &= \overline{\mathsf{u}}
        (\gamma, \denSet{\hastype{\Delta}{v}{\TD B}}\gamma)
       \oplus_p
       \overline{\mathsf{u}}
        (\gamma, \denSet{\hastype{\Delta'}{t}{\TD B}}\gamma) 
    \tag{$\overline{\mathsf{u}}$ homo} \\
    &= \denSet{\hastype{\Gamma+r\Delta}{\letIn{x=v}{u}}{\TD A}}\gamma
       \oplus_p
       \denSet{\hastype{\Gamma+r\Delta'}{\letIn{x=t}{u}}{\TD A}}\gamma \\
    &=\denSet{\hastype{\Gamma+r(p\Delta+(1-p)\Delta')}{(\letIn{x = v}u)\oplus_p(\letIn{x=t}u)}{\TD A}}\gamma
\end{align*}
where $\mathsf{u} = \denSet{\hastype{\Gamma,x:^r\TD B}{u}{\TD A}}$ and, by Proposition~\ref{prop:D:free:IB}, $\overline{\mathsf{u}}$ is an homomorphism of IB algebras on its second component.

    \item[Case rec:] The judgmental equalities regarding recursion on natural numbers 
    \begin{align*}
        \Rec{z}{(x,y).s}{\Zero} & \jeq z 
        &  
        \Rec z{(x,y).s}{\Succ n} & \jeq s[\Rec z{(x,y).s}n/x, n/y] 
    \end{align*}
    follow by definition of the $\iterate[]$ map
    and Lemma~\ref{lem:sem:subst}.
    The only thing to observe is that both sides of the equalities
    are well typed in the same context thanks to weakening (Lemma~\ref{lem:sem:subst}).

    \item[Case fix:] The judgmental equality regarding fixed points
    \begin{align*}
        \fix xt & \peq t[\fix xt/x]
    \end{align*}
    follows directly by definition of the fixed point operator $\semfp$
    and Lemma~\ref{lem:sem:subst}.
    \qedhere
\end{description}

\subsection*{Proof of Lemma~\ref{lem:sem:weak} (1)}
Two arrows in $\CMet$ are equal if so are the underlying set-maps.
Thus, the statement to prove simplifies to 
\begin{equation*}
    \denSet{\hastype{\Gamma,\Delta,\Gamma}tA}(\gamma,\delta,\gamma') = \denSet{\hastype{\Gamma,\Gamma}tA}(\gamma,\gamma') \,.
\end{equation*}
where $\denSet{\hastype{\Gamma}{t}{A}}$ denotes the underlying set-map of $\den{\hastype{\Gamma}{t}{A}}$, as per Figure~\ref{fig:setmap-semantics}.
The proof of the above is by induction on the derivation of $\hastype{\Gamma,\Gamma'}tA$. We will closely follow the steps of the proof of Lemma~\ref{lm:weakening}(1) as the derivations used here are the same. We select only a few interesting cases:
\begin{description}[font=\sc,wide]
\item[Case (var):] Assume $\hastype{\Gamma,\Gamma'}{x}{A}$ was derived with an application of (\textsc{var}). This can happen in two cases. We consider one of them as the other is similar. Let $\Gamma = \Gamma_1,x:^r A,\Gamma_2$. Then
\begin{align*}
\denSet{\hastype{\Gamma_1,x:^r A,\Gamma_2,\Gamma'}xA}(\gamma_1,a,\gamma_2,\gamma') 
= a
= \denSet{\hastype{\Gamma_1,x:^r A,\Gamma_2,\Delta,\Gamma'}xA}(\gamma_1,a,\gamma_2,\delta,\gamma')\,.
\end{align*}

\item[Case (abs):] Assume $\hastype{\Gamma,\Gamma'}{\lambda x. t}{A \lol_r B}$ was derived with an application of (\textsc{abs}). 
Then
\begin{align*}
\denSet{\hastype{\Gamma,\Gamma'}{\lambda x. t}{A \lol_r B}}(\gamma,\gamma')
&= \curry(\denSet{\hastype{\Gamma,\Gamma',x:^rA}{t}{B}})(\gamma,\gamma') \\
&= \curry(\denSet{\hastype{\Gamma,\Delta,\Gamma',x:^rA}{t}{B}})(\gamma,\delta,\gamma')
    \tag{*} \\
&= \denSet{\hastype{\Gamma,\Delta,\Gamma'}{\lambda x. t}{A \lol_r B}}(\gamma,\delta,\gamma')
\end{align*}
where (*) is justified by universality of the adjunction $(-\mprod X \dashv X \mexp -)$ and the following equality 
\begin{equation*}
\eval \circ ((\curry(\den{\hastype{\Gamma,\Gamma',x:^rA}{t}{B}}) \circ \proj) \mprod r\den A)
= \den{\hastype{\Gamma,\Delta,\Gamma',x:^rA}{t}{B}}
\end{equation*}
which easily follows by inductive hypothesis.

\item[Case (app):] Assume $\hastype{\Gamma,\Gamma'}{tu}{A}$ was derived with an application of (\textsc{app}). Then, for the case when $r>0$, by inductive hypothesis we have that
\begin{align*}
\denSet{\hastype{\Gamma,\Gamma'}{tu}{A}}(\gamma,\gamma')
&= \eval( \denSet{\hastype{\Gamma_1,\Gamma'_1}{t}{B \lol_r A}}(\gamma,\gamma') , 
         \denSet{\hastype{\Gamma_2,\Gamma'_2}{u}{B}}(\gamma,\gamma') ) \\
&= \eval( \denSet{\hastype{\Gamma_1,\Delta_1,\Gamma'_1}{t}{B \lol_r A}}(\gamma,\delta,\gamma') , 
         \denSet{\hastype{\Gamma_2,\Delta_2,\Gamma'_2}{u}{B}}(\gamma,\delta,\gamma') ) \\
&= \denSet{\hastype{\Gamma,\Delta,\Gamma'}{tu}{A}}(\gamma,\delta,\gamma') \,.
\end{align*}
where $\Gamma = \Gamma_1+r\Gamma_2$, $\Gamma' = \Gamma'_1+r\Gamma'_2$, and $\Delta = \Delta_1 +r\Delta_2$. The case when $r=0$ is similar.

\item[Case ($\mprod$):] Assume $\hastype{\Gamma,\Gamma'}{(t,u)}{A \pair{r}{s} B}$ was derived with an application of ($\mprod$). 
Then, for the case when $r,s>0$, by inductive hypothesis we get 
\begin{align*}
\denSet{\hastype{\Gamma,\Gamma'}{(t,u)}{A \pair{r}{s} B}}(\gamma,\gamma')
&= (\denSet{\hastype{\Gamma_1,\Gamma_2}{t}{A}}(\gamma,\gamma'),
  \denSet{\hastype{\Gamma'_1, \Gamma'_2}{u}{B}}(\gamma,\gamma') ) \\
&= (\denSet{\hastype{\Gamma_1,\Delta_1\Gamma_2}{t}{A}}(\gamma,\delta,\gamma'),
  \denSet{\hastype{\Gamma'_1,\Delta_2,\Gamma'_2}{u}{B}}(\gamma,\delta,\gamma') ) \\
&= \denSet{\hastype{\Gamma,\Delta,\Gamma'}{(t,u)}{A \pair{r}{s} B}}(\gamma,\delta,\gamma')
\end{align*}
where $\Gamma = r\Gamma_1+s\Gamma_2+\Gamma_3$, $\Gamma' = r\Gamma'_1+s\Gamma'_2+\Gamma'_3$, and $\Delta=r\Delta_1+s\Delta_2+\Delta_3$. The other cases follow similarly.

\item[Case (let):] Assume $\hastype{\Gamma,\Gamma'}{ \letIn{x=u}{t} }{E}$ was derived with an application of (\textsc{let}). Then, for the case when $r>0$, we get that 
\begin{align*}
\denSet{\hastype{\Gamma,\Gamma}{\letIn{x=u}{t}}{E}}(\gamma,\gamma')
&= \overline{\denSet{\hastype{\Gamma_1,\Gamma'_1,x:^rA}{t}{E}}}
    (\gamma,\gamma', 
    \denSet{\hastype{\Gamma_2,\Gamma'_2}{u}{\TD A}}(\gamma,\gamma')) \\
&= \overline{\denSet{\hastype{\Gamma_1,\Delta_1,\Gamma'_1,x:^rA}{t}{E}}}
    (\gamma,\delta,\gamma', 
    \denSet{\hastype{\Gamma_2,\Delta_2,\Gamma'_2}{u}{\TD A}}(\gamma,\delta,\gamma')) \\
&= \denSet{\hastype{\Gamma,\Delta,\Gamma}{\letIn{x=u}{t}}{E}}(\gamma,\delta,\gamma')
\end{align*}
where 
$\Gamma = \Gamma_1+r\Gamma_2$, $\Gamma' = \Gamma'_1+r\Gamma'_2$, and $\Delta=\Delta_1+r\Delta_2$ and the second equality follows by inductive hypothesis and the equality
$
  \overline{\den{\hastype{\Gamma_1,\Gamma'_1,x:^rA}{t}{E}}} \circ \proj
  = 
  \overline{\den{\hastype{\Gamma_1,\Delta,\Gamma'_1,x:^rA}{t}{E}}}
$, which is a consequence of the uniqueness of the homomorphic extension in the sense of  Proposition~\ref{prop:D:free:IB}, as the following diagram commutes
    
    \begin{equation*}
    \newcommand{\A}{\den{A}}
    \newcommand{\GDGxr}{\den{\Gamma_1,\Delta,\Gamma'_1}\mprod r}
    \newcommand{\GGxr}{\den{\Gamma_1,\Gamma'_1}\mprod r}
    \begin{tikzcd}[column sep=4cm, row sep=0.7cm]
        \GDGxr\A \arrow[r, "\GDGxr\delta"] 
                  \arrow[d, "\proj"']
                  \arrow[dr, phantom, "\text{(naturality)}"]
        & \GDGxr\TD\A \arrow[d, "\proj"] \\
        \GGxr\A \arrow[r, "\GGxr\delta"]
                  \arrow[rd, near start, bend right=15, "\den{t}"']
                  \arrow[dr, phantom, "\text{(Prop.~\ref{prop:D:free:IB})}"]
        & \GGxr\TD\A \arrow[d, "\overline{\den{t}}"] \\
        & \den{E}
    \end{tikzcd}
    \end{equation*}
The case when $r=0$ follows similarly.

\item[Case (rec):] Assume $\hastype{\Gamma,\Gamma'}{ \Rec{z}{(x,y).s}{n} }{A}$ was derived with an application of (\textsc{rec}) and that $x,y \notin \Delta$ (otherwise apply $\alpha$-renaming to $s$).
Then, by inductive hypothesis we get
\begin{align*}
\denSet{\hastype{\Gamma,\Gamma'}{\Rec{z}{(x,y).s}{n}}{A}}(\gamma,\gamma')
&= \iterate[](\mathsf{z}(\gamma,\gamma'),\mathsf{s}(\gamma,\gamma'),\mathsf{n}(\gamma,\gamma')) \\
&= \iterate[](\mathsf{zw}(\gamma,\delta,\gamma'),\mathsf{sw}(\gamma,\delta,\gamma'),\mathsf{nw}(\gamma,\delta,\gamma')) \\
&= \denSet{\hastype{\Gamma,\Delta,\Gamma'}{\Rec{z}{(x,y).s}{n}}{A}}(\gamma,\delta,\gamma') \,,
\end{align*}
where $\Gamma = \Gamma_1+\infty\Gamma_2+\Gamma_3$, 
$\Gamma' = \Gamma'_1+\infty\Gamma'_2+\Gamma'_3$, $\Delta=\Delta_1+\infty\Delta_2+\Delta_3$, and 
\begin{align*}
\mathsf{z} &= \denSet{\hastype{\Gamma_1,\Gamma'_1}{z}{A}} \\
\mathsf{zw} &= \denSet{\hastype{\Gamma_1,\Delta_1,\Gamma'_1}{z}{A}} \\
\mathsf{s} &= (\curry\circ\curry)(\denSet{\hastype{\Gamma_2,\Gamma'_2,x:^1A,y:^1\nat}{s}{\nat}}) \\
\mathsf{sw} &= (\curry\circ\curry)(\denSet{\hastype{\Gamma_2,\Delta_2,\Gamma'_2,x:^1A,y:^1\nat}{s}{\nat}}) \\
\mathsf{n} &= \denSet{\hastype{\Gamma_3,\Gamma'_3}{n}{\nat}} \\
\mathsf{nw} &= \denSet{\hastype{\Gamma_3,\Delta_3,\Gamma'_3}{n}{\nat}}
\end{align*}
Note that the equality $\mathsf{s}(\gamma,\gamma') = \mathsf{sw}(\gamma,\delta,\gamma')$ is justified by the universality of the
adjunction $(- \mprod X \dashv X \mexp -)$ and the fact that
$\den{\hastype{\Gamma,\Delta,\Gamma',x:^1A,y:^1\nat}{s}{A}}$ is equal
to 
\begin{equation*}
\eval \circ ((\curry\circ\curry)(\den{\hastype{\Gamma,\Gamma',x:^1A,y:^1\nat}{s}{A}}) \circ \proj) \mprod \den A \mprod \nat)
\end{equation*}
which easily follows by inductive hypothesis.

\item[Case (fix):] Assume $\hastype{\Gamma,\Gamma'}{\fix{x}{t}}{A}$ was derived with an application of (\textsc{fix}). Then
\begin{align*}
\denSet{\hastype{\Gamma,\Gamma'}{\fix{x}{t}}{A}}(\gamma,\gamma')
&=\semfix(\den{\hastype{(1-p)(\Gamma,\Gamma'),x^p:A}{t}{A}})(\gamma,\gamma') \\
&=\semfix(\den{\hastype{(1-p)(\Gamma,\Delta,\Gamma'),x^p:A}{t} {A}})(\gamma,\delta,\gamma') \\
&= \denSet{\hastype{\Gamma,\Delta,\Gamma'}{\fix{x}{t}}{A}}(\gamma,\delta,\gamma')
\end{align*}
where the second equality follows by inductive hypothesis and unicity of the fixed point, as per definition of the map $\semfix$.
\end{description}
The induction follows similarly also for the terms of the extended 
calculus presented in Section~\ref{sec:logic}.

\subsection*{Proof of Lemma~\ref{lem:sem:weak}(2)}
The proof is by induction on the derivation of $\hastype{\Gamma}tA$. We closely follow the steps of the proof of Lemma~\ref{lm:weakening}(2), using the same derivation obtained for $\hastype{\Gamma+\Delta}{t}{A}$ of which we use its parts implicitly. 
We select only a few interesting cases:
\begin{description}[wide, font=\sc]

\item[Case (var):] Assume $\hastype{\Gamma}{x}{A}$ was derived with an application of (\textsc{var}). Then, $\Gamma = \Gamma_1, x:^rA, \Gamma_2$ for $r \geq 1$. Therefore,
\begin{align*}
\den{\hastype{\Gamma+\Delta}{x}{A}}
&= \epsilon_{\den A} \circ \kappa_{(r+s),1,\den A} \circ \pi_{(r+s)\den A} \\
&= \epsilon_{\den A} \circ \kappa_{r,1,\den A} \circ \kappa_{(r+s),r,\den A} \circ \pi_{(r+s)\den A} \\
&= \epsilon_{\den A} \circ \kappa_{r,1,\den A} \circ \pi_{r\den A}\circ \semweak_{\Gamma,\Delta} \tag{*}\\
&= \den{\hastype{\Gamma}{x}{A}}\circ \semweak_{\Gamma,\Delta}
\end{align*}
for $\Delta = \Delta_1, x:^sA, \Delta_2$, 
where (*) follows the fact 
that the projections $\pi_i \colon A_1\mprod A_2 \to A_i$ in $\CMet$ can be defined via the unitors: $\pi_1 = \lambda_A \circ (A \mprod !)$ and $\pi_2 = \rho_B \circ (! \mprod B)$.

\item[Case (abs):] Assume $\hastype{\Gamma}{\lambda x.t}{A \lol_r B}$ was derived with an application of (\textsc{abs}). Then
\begin{align*}
\den{\hastype{\Gamma}{\lambda x. t}{A \lol_r B}} \circ \semweak_{\Gamma,\Delta}
&= \curry(\denSet{\hastype{\Gamma,x:^rA}{t}{B}}) \circ \semweak_{\Gamma,\Delta} \\
&= \curry(\denSet{\hastype{\Gamma+\Delta,x:^rA}{t}{B}}) \tag{*} \\
&= \den{\hastype{\Gamma+\Delta}{\lambda x. t}{A \lol_r B}}
\end{align*}
where (*) is justified by universality of the adjunction $(-\mprod X \dashv X \mexp -)$ and the following equality 
\begin{equation*}
\eval \circ ((\curry(\den{\hastype{\Gamma,x:^rA}{t}{B}}) \circ \semweak_{\Gamma,\Delta}) \mprod r\den A)
= \den{\hastype{\Gamma+\Delta,x:^rA}{t}{B}}
\end{equation*}
which easily follows by inductive hypothesis.

\item[Case (app):] Assume $\hastype{\Gamma}{t\,u}{A}$ was derived with an application of (\textsc{app}). Then, by inductive hypothesis we get that
\begin{align*}
\den{\hastype{\Gamma}{t\,u}{A}} &\circ \semweak_{\Gamma,\Delta} = {} \\
&= \eval \circ (\den{\hastype{\Gamma_1}{t}{A \lol_r B}} \mprod (\distribute\circ r\den{\hastype{\Gamma_2}{u}{A}})) \circ \splt_{\Gamma_1,r\Gamma_2} \circ \semweak_{\Gamma,\Delta} \\
&= \eval \circ ((\den{\hastype{\Gamma_1}{t}{A \lol_r B}}\circ \semweak_{\Gamma_1,\Delta}) \mprod (\distribute\circ r\den{\hastype{\Gamma_2}{u}{A}})) \circ \splt_{(\Gamma_1+\Delta),r\Gamma_2} \\
&= \eval \circ (\den{\hastype{\Gamma_1+\Delta}{t}{A \lol_r B}} \mprod 
 (\distribute\circ r\den{\hastype{\Gamma_2}{u}{A}})) \circ \splt_{(\Gamma_1+\Delta),r\Gamma_2} \\
&=\den{\hastype{\Gamma+\Delta}{t\,u}{A}}
\end{align*}
for $\Gamma = \Gamma_1 + r\Gamma_2$ and where the second equality follows as $\semweak$ commutes with $\splt$ and $\semweak_{\Gamma,\emptyctx}$ is the identity.

\item[Case ($\mprod$):] Assume $\hastype{\Gamma}{(t,u)}{A \pair sr B}$ was derived with an application of ($\mprod$). Then,
\begin{align*}
\den{\hastype{\Gamma}{(t,u)}{A \pair sr B}} \circ \semweak_{\Gamma,\Delta}
&= \mathsf{pre} \circ \semweak_{(r\Gamma_1+s\Gamma_2),\Gamma_3} \circ \semweak_{\Gamma,\Delta} \\
&= \mathsf{pre} \circ \semweak_{(r\Gamma_1+s\Gamma_2),(\Gamma_3+\Delta)}\\
&=\den{\hastype{\Gamma+\Delta}{(t,u)}{A \pair sr B}}
\end{align*}
for $\Gamma = r\Gamma_1+s\Gamma_2+\Gamma_3$ and $\mathsf{pre} = (r\den{\hastype{\Gamma_1}{t}{A}} \mprod r\den{\hastype{\Gamma_2}{u}{B}}) \circ
(\distribute\mprod\distribute)\circ\splt$, where the second equality follows as $\semweak$ factorizes through itself.
Note that we did not use the inductive hypothesis in this step, for the same reason as in the proof of Lemma~\ref{lem:weak}.

\item[Case ($\oplus_p$):] Assume $\hastype{\Gamma}{t \oplus_p u}{\TD A}$ was derived with an application of ($\oplus_p$). Then, by inductive hypothesis we get that
\begin{align*}
&\den{\hastype{\Gamma}{t \oplus_p u}{\TD A}} \circ \semweak_{\Gamma,\Delta} = {} \\
&= \oplus_p \circ 
(p\den{\hastype{\Gamma_1}{t}{\TD A}} \mprod 
 \overline{p}\den{\hastype{\Gamma_1}{t}{\TD A}}) \circ
 (\distribute_{p,\Gamma_1}\mprod\distribute_{\overline{p},\Gamma_2})
    \circ\splt_{p\Gamma_1,\overline p\Gamma_2} \circ \semweak_{\Gamma,\Delta} \\
&= \oplus_p \circ 
(p(\den{\hastype{\Gamma_1}{t}{\TD A}}\circ\semweak_{\Gamma_1,\Delta}) \mprod 
 \overline{p}\den{\hastype{\Gamma_1}{t}{\TD A}}) \circ
 (\distribute_{p,(\Gamma_1+\Delta)} \mprod \distribute_{\overline{p},\Gamma_2})
    \circ\splt_{p(\Gamma_1+\Delta),\overline p\Gamma_2} \\
&= \oplus_p \circ 
(p(\den{\hastype{\Gamma_1+\Delta}{t}{\TD A}}) \mprod 
 \overline{p}\den{\hastype{\Gamma_1}{t}{\TD A}}) \circ
 (\distribute_{p,(\Gamma_1+\Delta)} \mprod \distribute_{\overline{p},\Gamma_2})
    \circ\splt_{p(\Gamma_1+\Delta),\overline p\Gamma_2} \\
&=\den{\hastype{\Gamma+\Delta}{t \oplus_p u}{\TD A}}
\end{align*}
where $\overline p = 1-p$, $\Gamma = p\Gamma_1 + \overline{p}\Gamma_2$ and the second equality follows as weak commutes with $\splt$ and $\distribute$, and $\semweak_{\overline{p}\Gamma_2,\emptyctx}$ is the identity.
\end{description}
The induction follows similarly also for the terms of the extended calculus 
presented in Section~\ref{sec:logic}.

\subsection*{Proof of Lemma~\ref{lem:sem:subst} (Semantic Substitution)}
By using the underlying set-maps given in Figure~\ref{fig:setmap-semantics}, the statement to prove simplifies to:
\begin{align*}
\denSet{\hastype{(\Gamma,\Gamma')+r\Delta}{t[u/x]}{B}}(\gamma,\gamma')
&= \denSet{\hastype{\Gamma,x:^rA,\Gamma'}tB}(\gamma,\denSet{\hastype{\Delta}{u}{A}}(\gamma,\gamma'),\gamma') \tag{for $r>0$} \\
\denSet{\hastype{\Gamma,\Gamma'}{t[u/x]}{B}}(\gamma,\gamma')
&= \denSet{\hastype{\Gamma,x:^0A,\Gamma'}tB}(\gamma,*,\gamma') \tag{for $r=0$}
\end{align*}

Note that the formulation of the lemma assumes implicitly that  statement above is $(\Gamma,\Gamma')+r\Delta$ is well-defined, meaning that $(\Gamma,\Gamma')$ 
and $r\Delta$ have the same variables of the same types in the same order. This does not, however, completely justify the notation above which 
assumes that  $(\gamma,\gamma')$ is an element in both $\den{\Gamma,\Gamma'}$ and $\den{r\Delta}$, because some sensitivities could be zero in one context, but not the other.
In that case, some elements in $(\gamma, \gamma')$ should be replaced by the unique element in $1$. 
Strictly speaking, all such applications should therefore be precomposed with appropriate weakenings. For simplicity of notation, we will not do that however,
and also note that since we have proved weakening (Lemma~\ref{lem:sem:weak} (2)) we know that such weakenings do not essentially change the interpretation of terms. 

The proof is by induction on the derivation of $\hastype{\Gamma, x:^rA,\Gamma'}{t}{B}$. 
We closely follow the steps of the proof of Lemma~\ref{lm:substitution}, and use the derivations therein implicitly.
We show a few interesting cases.
\begin{description}[font=\sc,wide]
\item[Case (var):] Assume the derivation of $\hastype{\Gamma, x:^rA,\Gamma'}{t}{B}$ ends with an application of (\textsc{var}). Then $t = y$ for some variable $y$. We distinguish two cases:
\begin{itemize}
    \item If $y=x$, then $t[u/x] = u$, $r\geq 1$, and $B = A$. Therefore,
    \begin{align*}
    \denSet{\hastype{(\Gamma,\Gamma')+\Delta}{x[u/x]}{B}}(\gamma,\gamma')
    &= \denSet{\hastype{r\Delta}{u}{B}}(\gamma,\gamma') 
    \tag{Lemma~\ref{lem:sem:weak}(2)} \\
    &=\denSet{\hastype{(\Gamma,x:^rA,\Gamma')}{x}{A}}(\gamma,\denSet{\hastype{\Delta}{u}{B}}(\gamma,\gamma'),\gamma')
    \end{align*}

    \item If $y\neq x$, then $t[u/x] = y$ and either $\Gamma = \Gamma_1, y:^q B, \Gamma_2$ or $\Gamma' = \Gamma'_1, y:^q B, \Gamma'_2$, for some $q\geq 1$. 
    Consider the first situation, where $\gamma = (\gamma_1,b,\gamma_2)$
    \begin{align*}
    \denSet{\hastype{(\Gamma,\Gamma')+r\Delta}{y[u/x]}{B}}(\gamma,\gamma') 
    &= \denSet{\hastype{(\Gamma_1,y:^qB,\Gamma_2,\Gamma')+r\Delta}{y}{B}}(\gamma_1,b,\gamma_2,\gamma') \\
    &=b \\
    &=\denSet{\hastype{(\Gamma_1,y:^qB,\Gamma_2,x:^rA,\Gamma')}{y[u/x]}{B}}(\gamma_1,b,\gamma_2,a,\gamma')
    \end{align*}
    Where $a$ is $*$ for $r=0$ and $\denSet{\hastype{\Delta}{u}{A}(\gamma,\gamma')}$ otherwise.
\end{itemize}    

\item[Case (abs):] Assume the derivation of $\hastype{\Gamma, x:^r A,\Gamma'}{t}{B}$ ends with an application of (\textsc{abs}). Then the situation is as described in the corresponding step of the proof of Lemma~\ref{lm:substitution},
in particular $t=\lambda y.v$ and $B = C \lol_s D$, and $y\notin\Delta$.
Then we have
\begin{align*}
\denSet{\hastype{(\Gamma,\Gamma')&+r\Delta}{t[u/x]}{B}}(\gamma,\gamma') = {}\\
&= \denSet{\hastype{(\Gamma,\Gamma')+r\Delta}{\lambda y. v[u/x]}{C \lol_s D}}(\gamma,\gamma') \\
&=\curry(\denSet{\hastype{((\Gamma,\Gamma')+r\Delta),y:^sC}{v[u/x]}{D}})(\gamma,\gamma') \\
&=\curry(\denSet{\hastype{\Gamma,x:^rA,\Gamma',y:^sB}{v}{D}})(\gamma,
\denSet{\hastype{\Delta,y:^0C}{u}{A}}(\gamma,\gamma',*)
,\gamma') \tag{*} \\
&=\denSet{\hastype{\Gamma,x:^rA,\Gamma'}{\lambda y.v}{C \lol_s D}}(\gamma,
\denSet{\hastype{\Delta,y:^0C}{u}{A}}(\gamma,\gamma',*)
,\gamma') \\
&=\denSet{\hastype{\Gamma,x:^rA,\Gamma'}{t}{B}}(\gamma,
\denSet{\hastype{\Delta}{u}{A}}(\gamma,\gamma')
,\gamma') \tag{Lemma~\ref{lem:sem:weak}(1)}
\end{align*}
where (*) follows by inductive hypothesis and the universality of the counit
$\eval$ of the adjunction $(- \mprod X) \dashv (X \lol -)$.

\item[Case (app):] Assume the derivation of $\hastype{\Gamma, x:^p A,\Gamma'}{t}{B}$ ends with an application of (\textsc{app}). Then, the situation is as described in the corresponding step of the proof of Lemma~\ref{lm:substitution},
in particular $t=v w$. Assume $r,s>0$. Then, by inductive hypothesis, we have
\begin{align*}
\denSet{\hastype{(\Gamma,\Gamma')&+r\Delta}{t[u/x]}{B}}(\gamma,\gamma') = {} \\
&= \denSet{\hastype{(\Gamma,\Gamma')+r\Delta}{v[u/x] w[u/x]}{B}}(\gamma,\gamma') \\
&=\eval(\denSet{\hastype{(\Gamma_1,\Gamma'_1)+r_1\Delta}{v[u/x]}{C \lol_s B}}(\gamma,\gamma'), \denSet{\hastype{(\Gamma_2,\Gamma'_2)+r_2\Delta}{w[u/x]}{C}}(\gamma,\gamma')) \\
&=\eval(
\denSet{\hastype{\Gamma_1,x:^{r_1}A,\Gamma'_1}{v}{C \lol_s B}}
(\gamma,a, \gamma'), 
\denSet{\hastype{\Gamma_2,x:^{r_2}A,\Gamma'_2}{w}{C}}
(\gamma,a,\gamma')) \\
&=\denSet{\hastype{\Gamma,x:^r,\Gamma'}{t}{B}}
(\gamma,\denSet{\hastype{\Delta}{u}{A}}(\gamma,\gamma'),\gamma')
\end{align*}
where $a = \denSet{\hastype{\Delta}{u}{A}}(\gamma,\gamma')$.
The cases ($r=0<s$), ($s=0<r$), and ($r=s=0$) are obtained as an adaptation of the above.

\item[Case (fix):] Assume the derivation of $\hastype{\Gamma, x:^p A,\Gamma'}{t}{B}$ ends with an application of (\textsc{fix}). Then, the situation is as described in the corresponding step of the proof of Lemma~\ref{lm:substitution},
in particular $t=\fix{y}{v}$, $q<1$, and $y\notin \Delta$.
Assume $r>0$. Then, we have
\begin{align*}
\denSet{\hastype{(\Gamma,\Gamma')+&r\Delta}{t[u/x]}{B}}(\gamma,\gamma') = {} \\
&= \denSet{\hastype{(\Gamma,\Gamma')+r\Delta}{\fix{y}{v[u/x]}}{B}}(\gamma,\gamma') \\
&=\semfp(\denSet{\hastype{(1-q)((\Gamma,\Gamma')+r\Delta),y:^qB}{v[u/x]}{B}})(\gamma,\gamma')\\
&=\semfp(\denSet{\hastype{(1-q)(\Gamma,x:^rA,\Gamma'),y:^qB}{v}{B}})
(\gamma,\denSet{\hastype{\Delta,y:^0B}{u}{A}}(\gamma,\gamma',*),\gamma') 
\tag{*} \\
&=\semfp(\denSet{\hastype{(1-q)(\Gamma,x:^rA,\Gamma'),y:^qB}{v}{B}})
(\gamma,\denSet{\hastype{\Delta}{u}{A}}(\gamma,\gamma'),\gamma') 
\tag{Lemma~\ref{lem:sem:weak}(1)} \\
&= \denSet{\hastype{\Gamma,\Gamma'}{\fix{y}{v}}{B}}(\gamma,\denSet{\hastype{\Delta}{u}{A}}(\gamma,\gamma'),\gamma')
\end{align*}
where (*) is justified by Proposition~\ref{lem:just:one:fix} and the 
equality below, which holds by inductive hypothesis:
\begin{align*}
\denSet{\hastype{(1-q)((&\Gamma,\Gamma')+r\Delta),y:^qB}{v[u/x]}{B}}(\gamma,\gamma',b) = {} \\
&=\denSet{\hastype{(1-q)(\Gamma,x:^rA,\Gamma'),y:^qB}{v}{B}}
(\gamma,\denSet{\hastype{\Delta,y:^0B}{u}{A}}(\gamma,\gamma',*),\gamma',b) \,.
\end{align*}
The case for $r=0$ is a simple adaptation of the above.

\end{description}
We remark that the proof follows similarly also for the terms of the calculus extended as per Section~\ref{sec:logic}.

\section{Omitted proofs of Section~\ref{sec:logic}}

\begin{lemma}
The following are non-expansive maps
\begin{align*}
 {\oplus}, {\lol} &\colon  \Prop \mprod \Prop \to \Prop \\
 \min\{p\cdot -, 1\} &\colon p \scaling \Prop \to \Prop   \quad (\text{for $p \in (0,\infty]$})\\
 {\sup}, {\inf} &\colon (\infty \scaling X \mexp \Prop) \to \Prop \\
 \max, \min &\colon \Prop \times \Prop \to \Prop \\
 d_X &\colon X \mprod X \to \Prop
\end{align*}
\end{lemma}
\begin{proof}
We prove non-expansiveness separately for each function:
\begin{description}[font=\sc,wide]
    \item[Case $\oplus$:]
    Let $a,a',b,b'\in \Prop$. We need to prove the following two inequalities:
    \begin{align*}
        (|a - a'| \oplus |b-b'|) + (a\oplus b) \geq a' \oplus b'
        &&\text{and}&&
        (|a - a'| \oplus |b-b'|) + (a'\oplus b') \geq a \oplus b \,.
    \end{align*}
    We prove only the latter, as the former follows similarly as $|x - y| = |y - x|$.
    
    If $|a - a'| + |b-b'| + a'+b' \geq 1$ we are done. Otherwise 
    \begin{align*}
        (|a - a'| \oplus |b-b'|) + (a' \oplus b')
        &= |a - a'| + |b-b'| + a'+b' \\
        &\geq (a - a') + (b-b') + a'+b' \\
        &=  a+b \\
        &\geq a \oplus b \,.
    \end{align*}   

    \item[Case $\lol$:]
    Let $a,a',b,b'\in \Prop$. We need to prove the following two inequalities:
    \begin{align*}
        (|a - a'| \oplus |b-b'|) + (a\lol b) \geq a' \lol b'
        &&\text{and}&&
        (|a - a'| \oplus |b-b'|) + (a'\lol b') \geq a \lol b \,.
    \end{align*}
    We prove only the former, as the latter follows.

    If $|a - a'| + |b-b'| \geq 1$ or $b'- a' \leq 0$ we are done, as both sides of the inequalities are in $[0,1]$. Otherwise 
    \begin{align*}
        (|a - a'| \oplus |b-b'|) + (a \lol b)
        &= |a - a'| + |b'-b| + (a \lol b) \\
        &\geq  (a - a') + (b'-b) + (b - a) \\
        &=  b'- a' \\
        &= a' \lol b' \,.
    \end{align*}

    \item[Case $\min\{ p\cdot -, 1\}$:]
    Let $a,b\in \Prop$ and $p \in (0,\infty]$. We need to prove the following two
    inequalities
    \begin{align*}
        \min \{ p|a-b|,1\} + \min\{pa, 1\} \geq \min\{pb, 1\}
        &&\text{and}&&
        \min \{ p|a-b|, 1\} + \min\{pb, 1\} \geq \min\{pa, 1\}
    \end{align*}
    We prove only the latter, as the other follows similarly using the fact
    that $|x-y| = |y-x|$.

    If $p|a-b| + pb \geq 1$ we are done. Otherwise,
    \begin{align*}
        \min \{ p|a-b|, 1\} + \min\{pb, 1\}
        &= p|a-b| + pb \\
        &\geq p(a-b) + pb \\
        &= pa \\
        &\geq \min\{pa, 1\} \,.
    \end{align*}
    
    \item[Case $\sup$:]
    Let $f,g \colon X \to [0,1]$ be set-maps (thus, elements of  the metric space $\infty\scaling X \mexp \Prop$). We need to prove that
    \begin{equation*}
        \sup_{x} |f(x)-g(x)| \geq |\sup_x f(x) - \sup_x g(x)| \,.
    \end{equation*}
    First notice that 
    \begin{enumerate}[leftmargin=3\parindent, labelindent=\parindent]
        \item For any $A\subseteq [-1,1]$, $|\sup A| \leq \sup |A|$, where $|A| = \{|a| \mid a \in A\}$. 
        
        Indeed, because $a \leq |a|$, we have $\sup A \leq \sup |A|$. 
        If $\sup A \geq 0$ we are done. 
        If $\sup A < 0$, let $- A \defeq \{ -a \mid a \in A \}$.
        Then $|A| = -A$. Moreover,
        $|\sup A| = - \sup A = \inf(-A) = \inf|A| \leq \sup |A|$.
        \item For any $h,i\colon X \to [0,1]$, $\sup_x (h(x)+i(x)) \leq \sup_x h(x) + \sup_x i(x)$.

        Indeed, $\sup_x (h(x)+i(x)) = \sup B$ and $\sup_x h(x) + \sup_x i(x) = \sup B'$
        where $B \defeq \{h(x)+i(y) \mid x,y\in X \text{ and } x=y \}$ is a subset of $B' \defeq \{h(x)+i(y) \mid x,y\in X \}$.
    \end{enumerate}
    By instantiating (1) with $A = \{f(x) - g(x) \mid x \in X\}$, we get $\sup_x |f(x) - g(x)| \geq |\sup_x (f(x) - g(x))|$; and by instantiating (2) with $h = f-g$ and $i = g$
    we get $\sup_x (f(x)-g(x)) \geq \sup_x f(x) - \sup_x g(x)$.
    
    From the above we conclude that
    \begin{equation} \label{eq:sup:nonexp}
        \sup_{x} |f(x)-g(x)| 
        \geq 
        |\sup_{x} (f(x)-g(x))|
        \geq
        \sup_{x} (f(x)-g(x))
        \geq 
        \sup_x f(x) - \sup_x g(x) \,.
    \end{equation}
    If $\sup_x f(x) \geq \sup_x g(x)$ we are done. Otherwise, if $\sup_x f(x) < \sup_x g(x)$, we have that 
    \begin{equation*}
        \sup_{x} |f(x)-g(x)| 
        = \sup_{x} |g(x)-f(x)|
        \geq \sup_x g(x) - \sup_x f(x) 
        = |\sup_x f(x) - \sup_x g(x)|
    \end{equation*}
    where the inequality follows as for \eqref{eq:sup:nonexp} by inverting the roles of $f$ and $g$.
    
    \item[Case $\inf$:] 
    Let $f,g \colon X \to [0,1]$ be set-maps (thus, elements of  the metric space $\infty\scaling X \mexp \Prop$). Non-expansiveness follows by non-expansiveness of $\sup$ (shown above) by noticing the following equalities
    \begin{align*}
        \sup_{x} |f(x)-g(x)|
        &= \sup_{x} |(1-f)(x)-(1-g)(x)| \\
        &\geq |\sup_x (1-f)(x) - \sup_x (1-g)(x)| 
        \tag{$\sup$ non-exp} \\
        &= |(1-\inf_x f(x)) - (1 - \inf_x g(x))| \\
        &= |\inf_x f(x) - \inf_x g(x)| \,.
    \end{align*}

    \item[Case $\max$:]
    Let $a,a',b,b'\in \Prop$. We need to prove the following inequalities:
    \begin{align*}
        \max\{|a - a'|, |b-b'|\} + \max\{a, b\} &\geq \max\{a', b'\}
        \\
        \max\{|a - a'|, |b-b'|\} + \max\{a', b'\} &\geq \max\{a, b\} \,.
    \end{align*}
    We prove only the latter, as the former follows similarly by using the fact
    that $|x - y| = |y - x|$.
    Notice that 
    \begin{align*}
        \max\{|a - a'|, |b-b'|\} + \max\{a', b'\}
        &\geq \max\{|a - a'|, |b-b'|\} + a' \\
        &\geq (a-a') + a' \\
        &= a
    \end{align*}
    and similarly, $\max\{|a - a'|, |b-b'|\} + \max\{a', b'\} \geq b$.
    Thus, conclude $\max\{|a - a'|, |b-b'|\} + \max\{a', b'\} \geq \max\{a, b\}$.

    \item[Case $\min$:]
    It follows from non-expansiveness of $\max$ since
    $\min\{a,b\} = 1-\max\{1-a,1-b\}$ and  $|a-b| = |(1-a)-(1-b)|$.

    \item[Case $d_X$:]
    Let $a,a',b,b'\in \Prop$. We need to prove the following two inequalities:
    \begin{align*}
        (d_X(a,a') \oplus d_X(a,b')) + d_X(a,b) &\geq d_X(a',b')
        //
        (d_X(a,a') \oplus d_X(a,b')) + d_X(a',b') &\geq d_X(a,b) \,.
    \end{align*}
    We prove only the former, as the latter follows similarly.
    If $d_X(a,a') \oplus d_X(a,b') \geq 1$, we are done. Otherwise,
    \begin{align*}
        (d_X(a,a') \oplus d_X(a,b')) + d_X(a,b)
        &= d_X(a,a') + d_X(a,b') + d_X(a,b) \\
        &= d_X(a',a) + d_X(a,b) + d_X(a,b') \tag{symmetry} \\
        &\geq d_X(a',b') \,. \tag{triangular inequality}
    \end{align*}
    
\end{description}
\end{proof}

Next we provide the explicit definition for the underlying set maps for the interpretations of logical predicates. These equations should be added to those
in Figure~\ref{fig:setmap-semantics}.
\begin{align*}
    \denSet{\hastype{\Gamma}{\true}{\Prop}}\gamma &= 0 \\
    \denSet{\hastype{\Gamma}{\false}{\Prop}}\gamma &= 1 \\
    \denSet{\hastype{\Gamma+\Gamma'}{t \peq_A u}{\Prop}}\gamma 
        &= d_{\den A}(\denSet{\hastype{\Gamma}{t}{A}}\gamma, \denSet{\hastype{\Gamma'}{u}{A}}\gamma) \\
    \denSet{\hastype{\Gamma+\Gamma'}{\varphi \ltensor \psi}{\Prop}}\gamma 
        &=  
        \denSet{\hastype{\Gamma}{\varphi}{\Prop}}\gamma \oplus \denSet{\hastype{\Gamma'}{\psi}{\Prop}}\gamma \\
    \denSet{\hastype{\Gamma+\Gamma'}{\varphi \lexp \psi}{\Prop}}\gamma 
        &= 
        \denSet{\hastype{\Gamma}{\varphi}{\Prop}}\gamma \lol \denSet{\hastype{\Gamma'}{\psi}{\Prop}}\gamma \\
    \denSet{\hastype{r\Gamma+\Gamma'}{r\varphi}{\Prop}}\gamma 
        &= \min \{r \cdot \denSet{\hastype{\Gamma}{\varphi}{\Prop}}\gamma, 1\} \\
    \denSet{\hastype{\Gamma}{\neg\varphi}{\Prop}}\gamma 
        &= 1-\denSet{\hastype{\Gamma}{\varphi}{\Prop}}\gamma \\
    \denSet{\hastype{\Gamma}{\varphi \wedge \psi}{\Prop}}\gamma 
        &= 
        \max \{ \denSet{\hastype{\Gamma}{\varphi}{\Prop}}\gamma, \denSet{\hastype{\Gamma}{\psi}{\Prop}}\gamma \} \\
    \denSet{\hastype{\Gamma}{\varphi \vee \psi}{\Prop}}\gamma 
        &= 
        \min \{ \denSet{\hastype{\Gamma}{\varphi}{\Prop}}\gamma, \denSet{\hastype{\Gamma}{\psi}{\Prop}}\gamma \} \\
    \denSet{\hastype{\Gamma}{\exists x: A. \varphi}{\Prop}}\gamma 
        &= 
        \inf_{a\in \den{A}} \denSet{\hastype{\Gamma,x:^\infty A}{\varphi}{\Prop}}(\gamma,a) \\
    \denSet{\hastype{\Gamma}{\forall x: A. \varphi}{\Prop}}\gamma 
        &= 
        \sup_{a\in \denSet{A}} \denSet{\hastype{\Gamma,x:^\infty A}{\varphi}{\Prop}}(\gamma,a) 
\end{align*} 

\subsection*{Proof of Theorem~\ref{thm:Soundness:Logic} (Soundness)}
We need to show that all the rules of inference in Figure~\ref{fig:logicrules} preserve soundness.
Observe that the rules (\textsc{true}), (\textsc{false}), (\textsc{ex}),
(\textsc{$\ltensor$-i}), (\textsc{$\ltensor$-e}), (\textsc{$\lexp$-i}), 
(\textsc{$\lexp$-e}), (\textsc{$\wedge$-i}), (\textsc{$\wedge$-el/r}), 
(\textsc{$\vee$-i/r}), and (\textsc{$\vee$-e}) are sound with the expected interpretation on a generic commutative quantale. Thus they are sound also in $\Prop$.
The remaining cases are considered below. 

As for most cases, the typing context is understood, we simply write
$\denSet{\varphi}$ for $\denSet{\hastype{\Delta}{\varphi}{\Prop}}$. Also, for $\Psi = \psi_1, \dots, \psi_n$ a list of predicates, we define 
$\denSet{\Psi}\delta \defeq \bigoplus_i\denSet{\psi_i}\delta$.

\begin{description}[font=\sc,wide]
    \item[Rule (pr):] Assume $\denSet{\Psi}\delta \geq \denSet{\varphi}\delta$, for all $\delta \in \denSet\Delta$.
    Let $\delta \in \Delta$.
    If $\denSet{r\Psi}\delta = 1$ then $\denSet{r\Psi}\delta \geq \denSet{r\varphi}\delta$ holds trivially. Let $\Psi = \psi_1, \dots, \psi_n$ and assume $\denSet{r\Psi}\delta < 1$. Clearly, $\denSet{r\psi_i}\delta < 1$, for all $1\leq i \leq n$. Moreover,
    \begin{align*}
        \denSet{r\Psi}\delta 
        = \bigoplus_i \min\{ r\cdot \denSet{\psi_i}\delta, 1\} 
        = r\cdot \sum_i \denSet{\psi_i}\delta 
        \geq
        r \cdot \denSet\Psi\delta 
        \geq r \cdot \denSet\varphi\delta 
        \geq 
        \denSet{r\varphi}\delta\,.
    \end{align*}

    \item[Rule (dup):] Soundness of both directions of the rule follows by
    the following equality
    \begin{align*}
        \denSet{r\varphi}\delta \oplus \denSet{s\varphi}\delta 
        &= \min \{ r\denSet\varphi\delta,1\} \oplus \min \{ s\denSet\varphi\delta,1\} \\
        &= \min \{ r\denSet\varphi\delta+s\denSet\varphi\delta,1\} \\
        &= \min \{ (r+s)\denSet\varphi\delta,1\} \\
        &= \denSet{(r+s)\varphi}\delta \,.
    \end{align*}

    \item[Rule (der):] Soundness of both directions follows by
     \begin{align*}
        \denSet{\psi}\delta
        = \min \{ \denSet{\psi}\delta,1\}
        = \min \{ 1 \cdot \denSet{\psi}\delta,1\}
        = \denSet{1\psi}\delta
     \end{align*}

    \item[Rule (inc):] Soundness follows by monotonicity of $\oplus$
    (in both arguments) and the fact that, whenever $r \leq s$, we have
    \begin{align*}
        \denSet{r\psi}\delta
        = \min \{ r \cdot \denSet{\psi}\delta,1\}
        \leq \min \{ s \cdot \denSet{\psi}\delta,1\}
        = \denSet{s\psi}\delta \,.
    \end{align*}

    \item[Rule (assoc${}_1$):] Soundness follows by monotonicity of $\oplus$
    (in both arguments) and the following inequality 
    \begin{align*}
        \denSet{r(s\psi)}\delta
        = \min \{ r \cdot \min\{s\denSet\psi,1\}, 1\} 
        = \min \{ rs \cdot \denSet\psi, r, 1\}
        \leq \min \{ rs \cdot \denSet\psi, 1\}
        =\denSet{(rs)\psi} \,.
    \end{align*}

    \item[Rule (assoc${}_2$):] Soundness follows by monotonicity of $\oplus$
    (in both arguments) and the fact that the inequality proven in the case \textsc{(assoc${}_1$)} is an equality when $s\leq 1$ or $r\geq 1$.

    \item[Rule (g-rec):] If $\denSet\Psi\delta =1$, then the conclusion $\denSet\Psi\delta\geq \denSet \varphi\delta$ is trivial. Otherwise,
    \[\denSet{(1-p)\Psi}\delta = (1-p)\denSet\Psi\delta\] and $(1-p)\denSet\Psi\delta \oplus p \denSet\varphi\delta = (1-p)\denSet\Psi\delta + p \denSet\varphi\delta$. So we conclude by
    \begin{align*}
        (1-p)\denSet\Psi\delta + p \denSet\varphi\delta \geq \denSet\varphi\delta
        &\iff
        (1-p)\denSet\Psi\delta \geq (1-p)\denSet\varphi\delta \\
        &\iff
        \denSet\Psi\delta \geq \denSet\varphi\delta \,.
    \end{align*}

    \item[Rule ($\neg$-i):] Soundness follows by the fact that $\lol$ is right adjoint to $\oplus$ and the interpretation of $\false$. Indeed,
    \begin{align*}
        \denSet\Psi\delta \oplus \denSet\varphi\delta \geq \denSet\false\delta
        &\iff
        \denSet\Psi\delta \geq \denSet\varphi\delta \lol \denSet\false\delta
    \end{align*}
    and $\denSet\varphi\delta \lol \denSet\false\delta = \max \{1 -\denSet\varphi\delta, 0\} = 1 -\denSet\varphi\delta = \denSet{\neg\varphi}\delta$.

    \item[Rule ($\neg$-e):] Soundness follows by the fact that $\lol$ is right adjoint to $\oplus$ and the interpretation of $\false$. Indeed,
    \begin{align*}
        \denSet\Psi\delta \oplus \denSet{\neg\varphi}\delta \geq \denSet\false\delta
        &\iff
        \denSet\Psi\delta \geq \denSet{\neg\varphi}\delta \lol \denSet\false\delta
    \end{align*}
    and $\denSet{\neg\varphi}\delta \lol \denSet\false\delta = \max \{1 -\denSet{\neg\varphi}\delta, 0\} = 
    \max \{\denSet{\varphi}\delta, 0\} =
    \denSet{\varphi}\delta$.

    \item[Rule ($\exists$-i)] Soundness follows by the following:    \begin{align*}
        \denSet{\hastype{\Delta}{\Psi}{\Prop}}\delta 
        &\geq \denSet{\hastype{\Delta}{\varphi[t/x]}{\Prop}}\delta 
        \\
        &= \denSet{\hastype{\Delta+\infty\Delta}{\varphi[t/x]}{\Prop}}\delta\\
        &=\denSet{\hastype{\Delta,x:A}{\varphi}{\Prop}}
            (\delta,\denSet{\hastype{\Delta}{t}{A}}\delta)
        \tag{Lemma~\ref{lem:sem:subst}}\\
        &\geq \inf_{a\in\den A}
        \denSet{\hastype{\Delta,x:A}{\varphi}{\Prop}}(\delta,a) \\
        &=\denSet{\hastype{\Delta}{\exists x:A.\varphi}{\Prop}}\delta \,.
    \end{align*}
    
    \item[Rule ($\exists$-e):] First note that $\inf\{r\cdot a\mid a\in A\} = r\cdot \inf A$ for any subset $A \subseteq \uinterval$
    because $r \cdot (-)$ is continuous and monotone for $r <\infty$. 
    (This fails for $r = \infty$ as $\inf\{\infty\cdot a\mid a\in (0,1]\} = 1$, and $\infty \cdot \inf (0,1] = \infty \cdot 0 = 0$.)
    Now assume 
    \[\denSet{\hastype{\Delta,x:A}{\Psi}{\Prop}}(\delta,a) \oplus 
    r \cdot\denSet{\hastype{\Delta,x:A}{\varphi}{\Prop}}(\delta,a) \geq \denSet{\hastype{\Delta,x:A}{\psi}{\Prop}}(\delta,a),\] 
    for all $\delta\in\den\Delta$ and $a\in \den{A}$.
    As the logical judgment appearing in the conclusion of the rule 
    is well-formed, we know that $\hastype{\Delta}{\Psi}{\Prop}$ 
    and $\hastype{\Delta}{\psi}{\Prop}$. 
    Thus, by Lemma~\ref{lem:sem:subst}, we have
    \begin{align*}
        \denSet{\hastype{\Delta,x:A}{\Psi}{\Prop}}(\delta,a) & = 
        \denSet{\hastype{\Delta}{\Psi}{\Prop}}\delta \\
        \denSet{\hastype{\Delta,x:A}{\psi}{\Prop}}(\delta,a) & =
        \denSet{\hastype{\Delta}{\psi}{\Prop}}\delta
    \end{align*}
    As a consequence, 
    \begin{align*}
    \denSet{\hastype{\Delta}{\psi}{\Prop}}\delta
    &\leq \inf_{a\in\den{A}} \big(\denSet{\hastype{\Delta}{\Psi}{\Prop}}\delta
    \oplus r\cdot\denSet{\hastype{\Delta,x:A}{\psi}{\Prop}}(\delta,a) \big)
    \\
    &= \denSet{\hastype{\Delta}{\Psi}{\Prop}}\delta
    \oplus \inf_{a\in\den{A}} r\cdot \denSet{\hastype{\Delta,x:A}{\psi}{\Prop}}(\delta,a) \\
    &= \denSet{\hastype{\Delta}{\Psi}{\Prop}}\delta
    \oplus r\cdot\inf_{a\in\den{A}} \denSet{\hastype{\Delta,x:A}{\psi}{\Prop}}(\delta,a) \\
    &= \denSet{\hastype{\Delta}{\Psi}{\Prop}}\delta
    \oplus \denSet{\hastype{\Delta}{r(\exists x:A.\varphi)}{\Prop}} \,.
    \end{align*}

    \item[Rule ($\forall$-i/e):] The proof of soundness follows similarly to those for existential quantifier, only that reasoning has to be thought with the order reversed. Note however,
    that in this case $\sup \{r\cdot a \mid a\in A\} = r\cdot \sup A$ holds for all $r$ including $r= \infty$. 

    \item[Case (eq-i):] Soundness follows as a direct consequence of the interpretation of equality as a distance. Indeed
    \begin{align*}
        \denSet{t \peq_A t}\delta = d_{\den{A}}(\denSet{t}\delta,\denSet{t}\delta) = 0 \leq \denSet{\Psi}\delta \,.
    \end{align*}

    \item[Rule (eq-e):] For the soundness, observe that the judgments appearing in the 
    rule are well-formed. Thus, $\Delta$ is discrete and all formulas are well-typed
    for the typing context of the logical judgment where they appear. Thus, 
    \begin{align*}
        \denSet{\hastype{\Delta}{{}&\Psi}{\Prop}}\delta \oplus \denSet{\hastype{\Delta}{\Psi'}{\Prop}}\delta \\
        &= \denSet{\hastype{\Delta+r\Delta}{\Psi}{\Prop}}\delta \oplus 
            \denSet{\hastype{r(\Delta+\Delta)}{\Psi'}{\Prop}}\delta
        \tag{Lemma~\ref{lem:sem:weak}} \\
        &\geq \denSet{\hastype{\Delta+r\Delta}{\varphi[t/x]}{\Prop}}\delta \oplus 
            \denSet{\hastype{r(\Delta+\Delta)}{r(t=s)}{\Prop}}\delta 
        \tag{premise} \\
        &= \denSet{\hastype{\Delta,x:^rA}{\varphi}{\Prop}}
            (\delta,\denSet{\hastype{\Delta}{t}{A}}\delta) \oplus
           \denSet{\hastype{r(\Delta+\Delta)}{r(t=s)}{\Prop}}\delta
        \tag{Lemma~\ref{lem:sem:subst}} \\
        \intertext{If $\denSet{\hastype{r(\Delta+\Delta)}{r(t=s)}{\Prop}}\delta = 1$, we immediately conclude. Otherwise:}
        &= \denSet{\hastype{\Delta,x:^rA}{\varphi}{\Prop}}
            (\delta,\denSet{\hastype{\Delta}{t}{A}}\delta) \oplus
            r\cdot \denSet{\hastype{\Delta+\Delta}{t=s}{\Prop}}\delta \\
        &= \denSet{\hastype{\Delta,x:^rA}{\varphi}{\Prop}}
            (\delta,\denSet{\hastype{\Delta}{t}{A}}\delta) \oplus
            d_{r\den A}(\denSet{\hastype{\Delta}{t}{A}}\delta,
                               \denSet{\hastype{\Delta}{s}{A}}\delta) \\
        &\geq \denSet{\hastype{\Delta,x:^rA}{\varphi}{\Prop}}
            (\delta,\denSet{\hastype{\Delta}{s}{A}}\delta)
        \tag{*} \label{eq:sound:EQ-e} \\
        &= \denSet{\hastype{\Delta+r\Delta}{\varphi[s/x]}{\Prop}}\delta
        \tag{Lemma~\ref{lem:sem:subst}} \\
        &= \denSet{\hastype{\Delta}{\varphi[s/x]}{\Prop}}\delta
        \tag{Lemma~\ref{lem:sem:weak}}
    \end{align*}
    where \eqref{eq:sound:EQ-e} follows as $\denSet{\hastype{\Delta,x:^rA}{\varphi}{\Prop}}(\delta, -) \colon r\den{A} \to \uinterval$ is non-expansive.

    \item[Rule (ind${}_\mprod$):] Suppose 
    \[\denSet{\hastype{\Delta,x:A,y:B}{\Psi[(x,y)/z]}{\Prop}}(\delta,a,b)
    \geq \denSet{\hastype{\Delta,x:A,y:B}{\varphi[(x,y)/z]}{\Prop}}(\delta,a,b)\] 
    holds for all
    $\delta \in \den{\Delta}$, $a \in \infty\den A$, and $b \in \infty\den B$. Then, by Lemma~\ref{lem:sem:subst} and interpretation of terms, also the following 
    inequality holds for all $\delta \in \den{\Delta}$, $c \in \den {A \pair{p}{q} B}$:
    \begin{equation}
    \label{eq:sound:ind:tensor}
        \denSet{\hastype{\Delta,z: A \pair{p}{q} B}{\Psi}{\Prop}}(\delta,c)
    \geq \denSet{\hastype{\Delta,z: A \pair{p}{q} B}{\varphi}{\Prop}}(\delta,c)
    \end{equation}
    From this we obtain
    \begin{align*}
        \denSet{\hastype{\Delta}{\Psi[t/z]}{\Prop}}\delta
        &= \denSet{\hastype{\Delta,z: A \pair{p}{q} B}{\Psi}{\Prop}}
        (\delta,\denSet{\hastype{\Delta}{t}{A \pair{p}{q} B}}) 
        \tag{Lemma~\ref{lem:sem:subst}} \\
        &\geq \denSet{\hastype{\Delta,z: A \pair{p}{q} B}{\varphi}{\Prop}}
        (\delta,\denSet{\hastype{\Delta}{t}{A \pair{p}{q} B}})
        \tag{by \eqref{eq:sound:ind:tensor}} \\
        &=\denSet{\hastype{\Delta}{\varphi[t/z]}{\Prop}}\delta
        \tag{Lemma~\ref{lem:sem:subst}}
    \end{align*}

    \item[Rule (ind$_+$):] The proof is similar to that of (\textsc{ind}${}_\mprod$).

    \item[Rule (ind$_\nat$):] From the assumption on the soundness of the premises and interpretation of the terms $\Zero$ and $\Succ{n}$ we have that for all $\delta \in \den\Delta$ the following two properties hold:
    \begin{itemize}
        \item $\denSet{\hastype{\Delta, n:\nat}{\varphi}{\Prop}}(\delta,0) 
        \leq \denSet{\hastype{\Delta}{\Psi}{\Prop}}\delta$;
        \item $\denSet{\hastype{\Delta, n:\nat}{\varphi}{\Prop}}(\delta,m+1) 
        \leq \denSet{\hastype{\Delta, n:\nat}{\varphi}{\Prop}}(\delta,m)$, for all $m \in \nat$.
    \end{itemize}
    Thus, by induction 
    the following inequality holds for all $m \in \nat$:
    \begin{equation*}
        \denSet{\hastype{\Delta, n:\nat}{\varphi}{\Prop}}(\delta,m) 
        \leq \denSet{\hastype{\Delta}{\Psi}{\Prop}}\delta
    \end{equation*}
    From the above and the fact that $\Delta$ is a discrete context we obtain
    \begin{align*}
        \denSet{\hastype{\Delta}{\Psi}{\Prop}}\delta
        &\geq \denSet{\hastype{\Delta, n:\nat}{\varphi}{\Prop}}
            (\delta,\denSet{\hastype{\Delta}{t}{\nat}}\delta) \\
        &= \denSet{\hastype{\Delta+\infty\Delta}{\varphi[t/n]}{\Prop}}\delta
        \tag{Lemma~\ref{lem:sem:subst}}\\
        &= \denSet{\hastype{\Delta}{\varphi[t/n]}{\Prop}}\delta \,.
    \end{align*}

   \item[Rule (ind$_{\TD}$):] 
    We prove the soundness of the rule as restated in Remark~\ref{rem:finiteTD-Induction}, as this is more general than the one presented in Figure~\ref{fig:logicrules}.
    
    To this end, first observe that the finitely supported dyadic Borel probability measures
    can be expressed as a $\oplus_{\frac 12}$-convex sum of Dirac distributions 
    at elements from its support. 
    
    From the soundness of the premises we obtain that the following two properties hold, for all $e \in \den\Delta$, $a\in \den A$, 
    and $\mu,\nu \in \D \den{A}$,
    \begin{itemize}
        \item $\denSet{\hastype{\Delta, {x:}\D A}{\varphi}{\Prop}}(e,\delta_{\denSet{A}}(a)) 
        \leq \denSet{\hastype{\Delta}{\Psi}{\Prop}}(e)$;
        \item $\denSet{\hastype{\Delta, {x:}\D A}{\varphi}{\Prop}}(e,\mu \oplus_{\frac12} \nu ) 
        \leq \frac12 \denSet{\hastype{\Delta, {x:}\D A}{\varphi}{\Prop}}(e,\mu) + 
            \frac12 \denSet{\hastype{\Delta, {x:}\D A}{\varphi}{\Prop}}(e,\nu)$.
    \end{itemize}
    From the above, by an easy induction, we have that 
    \begin{equation} \label{eq:sound:ind:TD}
        \denSet{\hastype{\Delta, {x:}\D A}{\varphi}{\Prop}}(e,\mu) \leq \denSet{\hastype{\Delta}{\Psi}{\Prop}}(e)
    \end{equation}
    holds for any finitely supported dyadic Borel probability measure $\mu$ on $\den A$.

   Since finitely supported dyadic Borel probability measures are dense in $\D A$ (see proof of \cite[Theorem 21]{BreugelHMW07} and references therein) and by assumption $\denSet{\hastype{\Delta, {x:}\D A}{\varphi}{\Prop}}(\delta, -) \colon r \D\den{A} \to [0,1]$ is $r$-Lipschitz for $r < \infty$ (thus, continuous), we have that the inequality
   \eqref{eq:sound:ind:TD} extends to the entire set of Radon probability measures in $\D \den A$.
   From this and the fact that $\Delta$ is discrete, for all $e\in \den\Delta$, we have
    \begin{align*}
        \denSet{\hastype{\Delta}{\Psi}{\Prop}}(e)
        &\geq \denSet{\hastype{\Delta, x:\D A}{\varphi}{\Prop}}(e,
        \denSet{\hastype{\Delta}{t}{\D A}}(e)) \\
        &= \denSet{\hastype{\Delta+\infty\Delta}{\varphi[t/x]}{\Prop}}(e)
        \tag{Lemma~\ref{lem:sem:subst}}\\
        &= \denSet{\hastype{\Delta}{\varphi[t/x]}{\Prop}}(e) \,.
        \qedhere
    \end{align*}
\end{description}

\subsection*{Proof of Lemma~\ref{lem:r:exists}}
Note that by (\textsc{der}), $1\varphi$ is equivalent to $\varphi$. In the proofs below, this equivalence will be used implicitly when the scaling factor $r$ used in the rules is $r = 1$. 
\begin{itemize}
\item Both the directions of the equivalence between $r(\varphi \wedge \psi)$ and $r\varphi \wedge r\psi$ follows by an application of (\textsc{$\wedge$-i}) from $r(\varphi \wedge \psi) \vdash r\varphi$ and $r(\varphi \wedge \psi) \vdash r\psi$. In turn, these are obtained from $\varphi \wedge \psi \vdash \varphi$ and $\varphi \wedge \psi \vdash \psi$, which are clearly valid, by an application of (\textsc{pr}).

\item By (\textsc{$\vee$-e}), to conclude $r(\varphi \vee \psi) \vdash r\varphi \vee r\psi$ it 
suffices to prove $r\varphi \vdash r\varphi \vee r\psi$ and $r\psi \vdash r\varphi \vee r\psi$, which are certanly true.
Conversely, $r\varphi \vee r\psi \vdash r(\varphi \vee \psi)$ follows by an application of (\textsc{$\vee$-e}) from $r\varphi \vdash r(\varphi \vee \psi)$ and $r\psi \vdash r(\varphi \vee \psi)$, which can be concluded from the clearly valid statements $\varphi \vdash \varphi \vee \psi$ and $\psi \vdash \varphi \vee \psi$ via an application of (\textsc{pr}).

\item By (\textsc{$\forall$-i}), to prove $r(\forall x:A. \varphi) \vdash \forall x:A. r\varphi$ it 
suffices to show that $r(\forall x:A. \varphi) \vdash r\varphi$, which follows from $\forall x:A. \varphi \vdash \varphi$ by (\textsc{pr}). 
Conversely, $\forall x:A. r\varphi \vdash r(\forall x:A. \varphi)$ is obtained from 
$\forall x:A. r\varphi \vdash r\varphi$ via an application of (\textsc{$\forall$-i}), since  
$\forall x:A. r\varphi \vdash r\varphi$ can be concluded via (\textsc{$\forall$-e}) from $\forall x:A. r\varphi \vdash \forall x:A. r\varphi$, by substituting $x$ for itself in $r\varphi$.

\item Let $r < \infty$. $r(\exists x:A. \varphi) \vdash \exists x:A. r\varphi$ is obtained from 
$r\varphi \vdash \exists x:A. r\varphi$ via an application of (\textsc{$\exists$-e}), as 
$r\varphi \vdash \exists x:A. r\varphi$ can be concluded from $r\varphi \vdash r\varphi$ via (\textsc{$\exists$-i}), by substituting $x$ for itself in $r\varphi$.
Conversely, $\exists x:A. r\varphi \vdash r(\exists x:A. \varphi)$ is proven by (\textsc{$\exists$-e})
from $r\varphi \vdash r(\exists x:A. \varphi)$, which in turn follows by $\varphi \vdash \exists x:A. \varphi$ via an application of (\textsc{pr}).
\end{itemize}

\subsection*{Proof of Lemma~\ref{lem:logic:weak} (Weakening of logical contexts)}
We want to show that, 
{\centering
if $\logicJ[\Delta]{\Psi}{\varphi}$ is derivable, 
so is $\logicJ[\Delta]{\Psi, \vartheta}{\varphi}$.
}

The proof is by induction on the derivation of $\logicJ{\Psi}{\varphi}$.
Below we show only a few selected interesting cases, as the others are proved similarly.
The cases below correspond to the last rule applied in the derivation. For convenience we use the same symbols appearing in Figure~\ref{fig:logicrules}
and reserve $\vartheta$ as the generic formula for which we do weakening.
\begin{description}[font=\sc,wide]
\item[Case (false):] Apply (\textsc{false}) to obtain 
$\logicJ{\Psi,\vartheta,\false}{\varphi}$ and then (\textsc{ex}) to
obtain the format required by the statement.

\item[Case (ex):] Then $\logicJ{\Psi,\varphi,\psi,\Psi'}{\rho}$ is derivable
and so is $\logicJ{\Psi,\varphi,\psi,\Psi',\vartheta}{\rho}$, by inductive hypothesis. Apply (\textsc{ex}) to obtain
\begin{mathpar}
\infrule[ex]{
    \logicJ{\Psi,\varphi,\psi,\Psi',\vartheta}{\rho}}{
\logicJ{\Psi,\psi,\varphi,\Psi',\vartheta}{\rho}}
\end{mathpar}

\item[Case (dup):] Then $\logicJ{\Psi,(r+s)\varphi}{\psi}$ is derivable.
Then, also $\logicJ{\Psi,(r+s)\varphi,\vartheta}{\psi}$ is derivable by
inductive hypothesis.
Apply (\textsc{ex}) and (\textsc{dup}) to get 
\begin{mathpar}
\infrule[dup]{
\infrule[ex]{
    \logicJ{\Psi,(r+s)\varphi,\vartheta}{\psi}
}{\logicJ{\Psi,\vartheta,(r+s)\varphi}{\psi}} }{
\logicJ{\Psi,\vartheta,r\varphi,s\varphi}{\psi}}
\end{mathpar}
and then (\textsc{ex}) twice to recover the syntactic format 
required by the statement.

\item[Case (g-rec):] Then, $\logicJ{(1-p)\Psi,p\varphi}{\varphi}$ is derivable, for some $p \in (0,1)$. Apply the inductive hypothesis on 
$(1-p)\vartheta$
to obtain $\logicJ{(1-p)\Psi,p\varphi, (1-p)\vartheta}{\varphi}$. 
Apply (\textsc{ex}) and then (\textsc{g-rec}) to get
\begin{mathpar}
\infrule[g-rec]{
\infrule[ex]{
    \logicJ{(1-p)\Psi,p\varphi, (1-p)\vartheta}{\varphi}
}{\logicJ{(1-p)(\Psi\vartheta),p\varphi}{\varphi}}
}{\logicJ{\Psi,\vartheta}{\varphi}}
\end{mathpar}

\item[Case ($\ltensor$-i):] Then $\logicJ{\Psi}{\varphi}$ and $\logicJ{\Psi'}{\varphi'}$ are derivable. Apply the inductive hypothesis on the second one to get $\logicJ{\Psi',\vartheta}{\varphi'}$, and then (\textsc{$\ltensor$-i})
to finally derive
\begin{mathpar}
\infrule[$\ltensor$-i]{
    \logicJ{\Psi}{\varphi}
    &
    \logicJ{\Psi',\vartheta}{\varphi'}
}{\logicJ{\Psi,\Psi',\vartheta}{\varphi\ltensor\varphi'}}
\end{mathpar}

\item[Case (eq-e):] Then $\logicJ[\Delta]{\Psi}{\varphi[t/x]}$ and 
$\logicJ[\Delta]{\Psi'}{r(t=_A u)}$ are derivable for well-typed
predicate $\hastype{\Delta, x:^r A}{\varphi}{\Prop}$ and terms
$\hastype{\Delta}{t}{A}$, $\hastype{\Delta}{u}{A}$. 
Apply the inductive hypothesis to $\logicJ[\Delta]{\Psi'}{r(t=_A u)}$ to
obtain $\logicJ[\Delta]{\Psi',\vartheta}{r(t=_A u)}$ and the 
(\textsc{eq-e}) to obtain
\begin{mathpar}
\infrule[eq-e]{
    \hastype{\Delta, x:^r A}{\varphi}{\Prop}
    &\hastype{\Delta}{t}{A} 
    &\hastype{\Delta}{u}{A}
    &\logicJ[\Delta]{\Psi}{\varphi[t/x]}
    &\logicJ[\Delta]{\Psi',\vartheta}{r(t=_A u)}
}{\logicJ[\Delta]{\Psi, \Psi',\vartheta}{\varphi[u/x]} }
\end{mathpar}

\item[Case (ind${}_\mprod$):] Then 
$\hastype{\Delta}{t}{A\! \pair{r}{s} B}$ and $\logicJ[\Delta, x : A, y : B]{\Psi}{\varphi[(x,y)/z]}$ is derivable. Apply the inductive 
hypothesis to obtain $\logicJ[\Delta, x : A, y : B]{\Psi,\vartheta}{\varphi[(x,y)/z]}$, and then (\textsc{ind${}_\mprod$}) to obtain
\begin{mathpar}
\infrule[ind${}_\mprod$]{
    \logicJ[\Delta, x : A, y : B]{\Psi,\vartheta}{\varphi[(x,y)/z]} 
    & \hastype{\Delta}{t}{A\! \pair{r}{s} B}
}{\logicJ{\Psi,\vartheta}{\varphi[t/z]} }
\end{mathpar}
\qedhere
\end{description}

\subsection*{Proof of Lemma~\ref{lem:weak}.1 (Logic -
weakening of typing context)}
We need to show that whenever $\logicJ[\Delta]{\Psi}{\varphi}$ is derivable, 
so is $\logicJ[\Delta, e:E]{\Psi}{\varphi}$, for $e \notin \Delta$.

The proof is by induction on the derivation of $\logicJ[\Delta]{\Psi}{\varphi}$. In most cases, the result follows directly from a straightforward application of the inductive hypothesis. Below, we focus only on two interesting cases, as the others are similar. For clarity, we use the same symbols as in Figure~\ref{fig:logicrules} and designate $e:E$ as the generic 
variable used for weakening.
\begin{description}[font=\sc,wide]
\item[Case (eq-i):] Then, $\hastype{\Delta}{t \jeq s}{A}$. By Lemma~\ref{lm:weakening}(1) we get $\hastype{\Delta, e:E}{t \jeq s}{A}$.
Apply (\textsc{eq-i}) to obtain 
\begin{mathpar}
\infrule[eq-i]{
    \hastype{\Delta,e:E}{t \jeq s}{A}
}{\logicJ[\Delta,e:E]{\Psi}{t=_A s} }
\end{mathpar}

\item[Case (eq-e):] Then $\logicJ[\Delta]{\Psi}{\varphi[t/x]}$ and 
$\logicJ[\Delta]{\Psi'}{r(t=_A u)}$ are derivable for well-typed
predicate $\hastype{\Delta, x:^r A}{\varphi}{\Prop}$ and terms
$\hastype{\Delta}{t}{A}$, $\hastype{\Delta}{u}{A}$. 
By Lemma~\ref{lm:weakening}(1) we have 
\begin{mathpar}
\hastype{\Delta, e:E, x:^r A}{\varphi}{\Prop} \,,
\and
\hastype{\Delta,e:E}{t}{A}\,, \text{and}
\and
\hastype{\Delta, e:E}{u}{A} \,.
\end{mathpar}
By inductive hypothesis, we $\logicJ[\Delta,e:E]{\Psi}{\varphi[t/x]}$ and $\logicJ[\Delta,e:E]{\Psi'}{r(t=_A s)}$ are also derivable. 
Apply (\textsc{eq-e}) to obtain
\begin{mathpar}
\infrule[eq-e]{
    \begin{array}{r@{}l}
    \hastype{\Delta,e:E, x:^r A}{\varphi}{{}& \Prop} \\
    \hastype{\Delta,e:E}{t}{{}& A} \\
    \hastype{\Delta,e:E}{u}{{}& A}
    \end{array}
    \and
    \begin{array}{r@{}l}
    \\
    \logicJ[\Delta,e:E]{\Psi}{{}&\varphi[t/x]} \\
    \logicJ[\Delta,e:E]{\Psi'}{{}&r(t=_A u)}
    \end{array}
}{\logicJ[\Delta,e:E]{\Psi, \Psi'}{\varphi[u/x]} }
\end{mathpar}
\end{description}

\subsection*{Proof of Lemma~\ref{lem:weak}.2 (Logic - substitution lemma)}
We show that if $\logicJ[\Delta, e:E]{\Psi}{\varphi}$ is derivable 
and $\hastype{\Delta}uE$, then $\logicJ[\Delta]{\Psi[u/e]}{\varphi[u/e]}$
is also derivable.

The proof proceeds by induction on the derivation of $\logicJ[\Delta]{\Psi}{\varphi}$. Below, we focus on the most interesting cases, as the others are similar. For clarity, we use the same symbols as in Figure~\ref{fig:logicrules} with an exception for typing context which we
assume to be $\Delta,e:E$. We designate $u$ as the generic term to
be substituted for $e$.
\begin{description}[font=\sc,wide]
\item[Case (true):] Apply (\textsc{true}) to obtain 
$\logicJ{\Psi[u/e]}{\true}$ and note that $\true = \true[u/e]$.
Moreover, as $\logicJ[\Delta,e:E]{\Psi}{\true}$ is well-formed, so is
$\logicJ[\Delta]{\Psi[u/e]}{\true}$ by Lemma~\ref{lm:substitution}.

\item[Case (pr):] Then, $\logicJ[\Delta,e:E]{\Psi}{\varphi}$.
By inductive hypothesis, $\logicJ{\Psi[u/e]}{\varphi[u/e]}$ is also derivable. From this, apply (\textsc{pr}) to obtain 
\begin{mathpar}
    \infrule[pr]{
        \logicJ{\Psi[u/e]}{\varphi[u/e]} }{
    \logicJ{p\Psi[u/e]}{p\varphi[u/e]}}
\end{mathpar}

\item[Case ($\exists$-i):] Then, $\hastype{\Delta,e:E}{t}{A}$ and
$\logicJ[\Delta,e:E]{\Psi}{\varphi[t/x]}$. 
Assume without loss of generality that $x \neq e$ (otherwise apply
$\alpha$-renaming to $\exists x:A.\varphi$).
By Lemma~\ref{lm:substitution} we have $\hastype{\Delta}{t[u/e]}{A}$
and by inductive hypothesis $\logicJ{\Psi[u/e]}{(\varphi[t/x])[u/e]}$ is
also derivable. By composition of substitutions
\[
(\varphi[t/x])[u/e] = (\varphi[u/e])[t[u/e]/x] \,.
\]
Thus we can apply (\textsc{$\exists$-i}) to obtain
\begin{mathpar}
    \infrule[$\exists$-i]{
        \hastype{\Delta}{t[u/e]}{A} &
        \logicJ{\Psi[u/e]}{(\varphi[u/e])[t[u/e]/x]}
    }{ \logicJ{\Psi[u/e]}{\exists x: A. \varphi[e/u]}}
\end{mathpar} 
and notice that, as $x \neq e$, we have $\exists x: A. \varphi[e/u] = (\exists x: A. \varphi)[e/u]$.

\item[Case ($\exists$-e):] Then, $\logicJ[\Delta,e:E,x:A]{\Psi,r\varphi}{\psi}$. By an application of the inductive hypothesis, also $\logicJ[\Delta,x:A]{\Psi[u/e],r\varphi[u/e]}{\psi[u/e]}$ is derivable.
Now apply (\textsc{$\exists$-e}) to obtain 
\begin{mathpar}
\infrule[$\exists$-e]{
    \logicJ[\Delta,x:A]{\Psi[u/e],r\varphi[u/e]}{\psi[u/e]} }{
\logicJ{\Psi[u/e],r(\exists x:A \varphi[u/e])}{\psi[u/e]} }
\end{mathpar}
and notice that, as $x\neq e$, we have $r(\exists x: A. \varphi[e/u]) = r(\exists x: A. \varphi)[e/u]$.

\item[Case (eq-i):] Then $\hastype{\Delta,e:E}{t \jeq s}{A}$. By Lemma~\ref{lm:substitution} we obtain $\hastype{\Delta}{t[u/e]}{A}$ and $\hastype{\Delta}{s[u/e]}{A}$.
Apply (\textsc{eq-i}) to obtain 
\begin{mathpar}
\infrule[eq-i]{
    \hastype{\Delta}{t[u/e] \jeq s[u/e]}{A}
}{\logicJ{\Psi[u/e]}{t[u/e] =_A s[u/e]} }
\end{mathpar}
and observe that $t[u/e] = t[u/e] = (t = t)[u/e]$.

\item[Case (eq-e):] Then $\logicJ[\Delta,e:E]{\Psi}{\varphi[t/x]}$ and $\logicJ[\Delta,e:E]{\Psi'}{r(t=_A s)}$ are derivable for the well-typed predicate $\hastype{\Delta,e:E,x:^rA}{\varphi}{\Prop}$ and terms
$\hastype{\Delta,e:E}{t}{A}$ and $\hastype{\Delta,e:E}{s}{A}$.
By Lemma~\ref{lm:substitution}, we respectively get $\hastype{\Delta,x:^rA}{\varphi[u/e]}{\Prop}$, for the predicate, and $\hastype{\Delta}{t[u/e]}{A}$, and $\hastype{\Delta}{s[u/e]}{A}$ for the terms.

By inductive hypothesis, 
$\logicJ{\Psi[u/e]}{(\varphi[t/x])[u/e]}$ and $\logicJ{\Psi'[u/e]}{r(t[u/e]=_A s[u/e])}$ are also derivable.
By composition of substitutions
\[
(\varphi[t/x])[u/e] = (\varphi[u/e])[t[u/e]/x] \,.
\]
Thus, we can apply (\textsc{eq-e}) to obtain
\begin{mathpar}
\infrule[eq-e]{
    \begin{array}{r@{}l}
    \hastype{\Delta,e:E, x:^r A}{\varphi[u/e]}{{}& \Prop} \\
    \hastype{\Delta,e:E}{t[u/e]}{{}& A} \\
    \hastype{\Delta,e:E}{s[u/e]}{{}& A}
    \end{array}
    \and
    \begin{array}{r@{}l}
    \\
    \logicJ[\Delta,e:E]{\Psi[u/e]}{{}&(\varphi[u/e])[t[u/e]/x]} \\
    \logicJ[\Delta,e:E]{\Psi'[u/e]}{{}&r(t[u/e]=_A s[u/e])}
    \end{array}
}{\logicJ{(\Psi,\Psi')[u/e]}{(\varphi[u/e])[s[u/e]/x]} }
\end{mathpar}
Once again, $(\varphi[u/e])[s[u/e]/x] = (\varphi[s/x])[u/e]$, so the conclusion is in the right format.

\item[Case (ind${}_\mprod$):] Then $\logicJ[\Delta,e:E,x: A,y: B]{\Psi}{\varphi[(x,y)/z]}$ and $\hastype{\Delta,e:E}{t}{A \pair{r}{s} B}$. By applying Lemma~\ref{lm:substitution} to the latter, we get $\hastype{\Delta}{t[u/e]}{A \pair{r}{s} B}$. By applying the inductive hypothesis to the former, we get that also $\logicJ[\Delta,x:A,y:B]{\Psi[u/e]}{(\varphi[(x,y)/z])[u/e]}$ is derivable. 
By composition of substitutions and the fact that $x,y\neq e$ the following holds
\[
(\varphi[(x,y)/z])[u/e] = (\varphi[u/e])[(x,y)/z] \,.
\]
Thus, we can apply (\textsc{ind${}_\mprod$}) to obtain
\begin{mathpar}
\infrule[ind${}_\mprod$]{
    \logicJ[\Delta,x:A,y:B]{\Psi[u/e]}{(\varphi[u/e])[(x,y)/z]} & 
    \hastype{\Delta}{t[u/e]}{A \pair{r}{s} B}
}{  \logicJ{\Psi[u/e]}{(\varphi[u/e])[t[u/e]/z]} }
\end{mathpar}
Once again, by composition of substitution, 
$(\varphi[u/e])[t[u/e]/z] = (\varphi[t/z])[u/e]$, so the conclusion
is in the expected format.
\end{description}
The remaining cases are similar.

\section{Omitted proofs of Section~\ref{sec:basic:prop}}

\subsection*{Proof of Lemma~\ref{lem:ext:tensor}}
 Let $\Delta = x,y : A \pair{r}{s} B$. The left-to-right implication is proved by (\textsc{eq-e}), so it suffices to prove 
 \[\logicJ[\Delta]{\cdot}{\letIn{(a,b) = x, (a',b') = x}{r\scaling(a\peq a') \ltensor s (b \peq b')} } \,, \]
 and by (\textsc{ind${}_\mprod$}) we can assume $x$ is of the form $(a,b)$, so the proof obligation reduces to
 $\logicJ[\Delta, a:A, b:B]{\cdot}{r\scaling(a\peq a) \ltensor s (b \peq b)}$, 
 which is true. For the other direction, by (\textsc{ind${}_\mprod$}) it suffices to prove that 
\[ \logicJ[\Delta, a,a':A, b,b':B]{r\scaling(a\peq a') , s (b \peq b')}{(a,b) \peq (a',b')} \,,\]
which follows by (\textsc{eq-e}) on the two equalities appearing in the logical context.

\subsection*{Proof of Lemma~\ref{lem:ext:sum}}
 If $x,y : A + B$ then $x \peq y$ is equivalent to 
\begin{align*}
& (\exists a,a'. (x \peq \inj_1 a) \ltensor (y \peq \inj_1 a') \ltensor (a \peq a')) \\
\vee & 
 (\exists b,b'. (x \peq \inj_2 b) \ltensor (y \peq \inj_2 b') \ltensor (b \peq b'))
\end{align*}
 For one direction, by equality induction, it suffices to prove 
\begin{align*}
& (\exists a,a'. (x \peq \inj_1 a) \ltensor (x \peq \inj_1 a') \ltensor (a \peq a')) \\
\vee & 
 (\exists b,b'. (x \peq \inj_2 b) \ltensor (x \peq \inj_2 b') \ltensor (b \peq b'))
\end{align*}
and by (\textsc{ind${}_+$}) we just need to prove it in the cases of $x = \inj_1(a)$ and $x = \inj_2(b)$. In the former case,
pick $a'=a$. Then we must prove $(\inj_1 a \peq \inj_1 a) \ltensor (\inj_1 a \peq \inj_1 a) \ltensor (a \peq a))$,
which is true. The other direction follows by transitivity (Proposition~\ref{prop:eq:symm:trans}).

\subsection*{Proof of Proposition~\ref{prop:Leibniz}}
That $r(x \peq_A y)$ implies $\forall \phi : \PredQ Ar . \phi(x) \lexp \phi(y)$ follows directly from (\textsc{eq-e}). The other direction 
follows by application of (\textsc{$\forall$-e}) to $\lambda z. r(x \peq z) : \PredQ Ar$.

\section{Omitted proofs of Section~\ref{sec:c:markov}}

\subsection*{Proof of Proposition~\ref{prop:bisim:eq}}

First note the following.

\begin{lemma} \label{lem:coup:imp}
 Let $R,S : \Rel AB$. Then 
 $\forall a,b . R(a,b) \lexp S(a,b), \coup{R}\rho\mu\nu \ts \coup{S}\rho\mu\nu.$
\end{lemma}
\begin{proof}
 The statement immediately reduces to 
 \[
   \forall a,b . R(a,b) \lexp S(a,b)\ts \mean{(a,b)}{\rho}{R(a,b)} \lexp \mean{(a,b)}{\rho}{S(a,b)}
 \]
 which can be proved by induction on $\rho$: The base case of $\rho = \delta(a,b)$ reduces to the assumption 
 $R(a,b) \lexp S(a,b)$. For the induction case we must prove that 
 \begin{align*}
 & p(\mean{(a,b)}{\rho}{R(a,b)} \lexp \mean{(a,b)}{\rho}{S(a,b)}), (1-p)(\mean{(a,b)}{\rho'}{R(a,b)} \lexp \mean{(a,b)}{\rho'}{S(a,b)})
 \\ & \ts \mean{(a,b)}{\rho\oplus_p\rho'}{R(a,b)} \lexp \mean{(a,b)}{\rho\oplus_p\rho'}{S(a,b)}
\end{align*}
 which follows by unfolding $\mean{(a,b)}{\rho\oplus_p\rho'}{R(a,b)} $ to 
 $p(\mean{(a,b)}{\rho}{R(a,b)})\ltensor (1-p)(\mean{(a,b)}{\rho'}{R(a,b)})$ and distributing $p$ and $(1-p)$ over $\lexp$ in
 the assumptions. 
\end{proof}

 The proof of Proposition~\ref{prop:bisim:eq} now proceeds as follows. To show that bisimilarity implies equality, we use 
 guarded recursion, so suppose
 $\discfact(\forall x,y . x\bisim y \lexp x \peq y), x \bisim y$. We must prove $x\peq y$. The assumption $x \bisim y$ unfolds to
 \[
 \letIn{(l, \mu) = \unfold(x), \,
    (l', \nu) = \unfold(y)}{
   l \peq l' \ltensor \discfact( \exists \rho . \coup{\bisim}\rho\mu\nu})
 \]
 and by Lemma~\ref{lem:coup:imp}, 
 and the guarded recursion hypothesis, the last part can be rewritten to 
 $\discfact(\exists \rho . \coup{\eq}\rho\mu\nu)$, which by Theorem~\ref{thm:internal:Kantorovic} is equivalent to 
 $\discfact(\mu \peq \nu)$. Using (\textsc{ind}${}_\mprod$) we have proved that 
 \[
   \discfact(\forall x,y . x\bisim y \lexp x \peq y), x \bisim y \ts 
   \letIn{(l, \mu) = \unfold(x), \,
    (l', \nu) = \unfold(y)}{
   l \peq l' \ltensor \discfact( \mu \peq\nu})
 \]
 By Lemma~\ref{lem:ext:tensor} we conclude $\unfold(x) \peq \unfold(y)$, which implies $x\peq y$,
 because $\unfold$ is an isomorphism.
 
 For the other direction, to prove that $x \peq y$ implies $x \bisim y$, by equality induction it suffices to show that 
 bisimilarity is reflexive. This is done by guarded recursion by showing $\discfact(\forall x. x\bisim x) \ts \forall x. x \bisim x$.
 Introducing $x$, the proof obligation unfolds to 
 \[
   \logicJ[x : \Proc]{\discfact(\forall x. x\bisim x)}{ 
   \letIn{(l, \mu) = \unfold(x), \,
    (l', \nu) = \unfold(x)}{
   l \peq l' \ltensor \discfact(  \exists \rho . \coup{\bisim}\rho\mu\nu})}
 \]
 and by (\textsc{ind}${}_\mprod$) it suffices to show that 
 \[
   \logicJ[x : \Proc, l : \Labels, \mu : \D(\Proc)]{\discfact(\forall x. x\bisim x)}{ 
   l \peq l \ltensor \discfact(  \exists \rho . \coup{\bisim}\rho\mu\mu)}
 \]
 so we have reduced the proof to showing that $\forall x. x\bisim x$ implies $\forall\mu . \exists \rho . \coup{\bisim}\rho\mu\mu)$.
 This is done similarly to the proof of Theorem~\ref{thm:internal:Kantorovic}:  
  Define $\hastype{\mu:^2 \TD A}{\omega(\mu)}{\TD(\Proc \pair11 \Proc)}$ as
 $\letIn{a = \mu}{\delta(a,a)}$. Then $\TD\pi_i(\omega(\mu)) \peq \mu$ for $i=1,2$ can be proved by induction on $\mu$.
Moreover, $\cdot \ts \mean{(a,a')}{\omega(\mu)}{a \bisim a'}$ can be proved by induction on $\mu$ as follows. If $\mu = \delta(a)$, 
by \eqref{eq:mean:dirac} and $\omega(\mu) \jeq \delta(a,a)$, $\mean{(a,a')}{\omega(\mu)}{a \bisim a'}$ reduces to $a\bisim a$, which is true
by assumption. If $\mu = \mu_1 \oplus_p \mu_2$, we must show
$\mean{(a,a')}{\omega(\mu_1\oplus_p \mu_2)}{a \bisim a'}$ in context
\[
 p(\mean{(a,a')}{\omega(\mu_1)}{a \bisim a'}), (1-p)(\mean{(a,a')}{\omega(\mu_2)}{a \bisim a'}) 
\]
which holds by \eqref{eq:mean:conv} because $\omega(\mu_1\oplus_p \mu_2) \jeq \omega(\mu_1)\oplus_p \omega(\mu_2)$.

\section{Omitted proofs of Section~\ref{sec:hypercube}}

We just show the missing case for (\ref{eq:sum:flip}), where $p$ and $q$ differ in more than one position. Recall
that this states that 
\[
\frac{N\!-\!1}{N\!+\!1}(p\peq q) \ts \sum_{i=0}^N\frac1{N\!+\!1}(\flip_i\,p \peq \flip_{\bij(i)}\,q)
\]
 
If $p$ and $q$ differ in positions $i_1,\dots, i_n$ for $n>1$, let $\bijarg pq$ be the permutation
that cycles $i_1,\dots, i_n$, then, $p \peq q$ has value $\frac nN$. For $i \in \{i_1,\dots, i_n\}$, 
$\flip_i \,p$ and $ \flip_{\bij i} \,q$ differ in $n-2$ positions, so that $\flip_i \,p \peq \flip_{\bij i} \,q$ equals $\frac{n-2}N$.
If $i \notin \{i_1,\dots, i_n\}$, then $\flip_i\,p$ and $ \flip_{\bij i}\,q$ differ in $n$ positions. So 
\begin{align*}
 \sum_{i=0}^N\frac1{N\!+\!1}(\flip_i\,p \peq \flip_{\bij(i)}\,q) & =  \frac1{N\!+\!1}(n \cdot \frac{n-2}N + (N\!+\!1\!-\!n)\frac nN) \\
 & = \frac1{N\!+\!1}\left(\frac{n^2-2n+nN+n-n^2}N\right) \\
 & = \frac1{N\!+\!1}\left(\frac{(N-1)n}N\right) 
\end{align*}
which equals $\frac{N\!-\!1}{N\!+\!1}(p\peq q)$ proving
\[
\frac{N\!-\!1}{N\!+\!1}(p\peq q) \ts \sum_{i=0}^N\frac1{N\!+\!1}(\flip_i\,p \peq \flip_{\bij(i)}\,q)
\]

\end{document}
\endinput